\documentclass[12pt,letterpaper]{article}
\usepackage{fullpage}

\input{_preamble.sty}
\input{_graphics.sty}

\usepackage[subtle]{savetrees}

\begin{document}
\title{``Post'' Pre-Analysis Plans: \\Valid Inference for Non-Preregistered Specifications\thanks{We are grateful to Alberto Abadie, Isaiah Andrews, and Anna Mikusheva for their guidance and support throughout this project. We thank Sarah Moon, Ashesh Rambachan, and Jonathan Roth for valuable feedback and suggestions. We also thank Abhijit Banerjee, Esther Duflo, Ben Olken, and participants in lunch workshops at MIT for early conversations that helped situate this work within applied research settings. Finally, we are especially grateful to Frank Schilbach for generous and repeated feedback through the life of this project. Vod gratefully acknowledges support from the Jerry A. Hausman Fellowship.}}

\author{Reca Sarfati\thanks{Department of Economics, Massachusetts Institute of Technology, sarfati@mit.edu.} \and Vod Vilfort\thanks{Department of Economics, Massachusetts Institute of Technology, vod@mit.edu.}}

\date{\today}

\maketitle

\begin{abstract}
\noindent Pre-analysis plans (PAPs) have become standard in experimental economics research, but it is nevertheless common to see researchers deviating from their PAPs to supplement preregistered estimates with non-prespecified findings. While such ex-post analysis can yield valuable insights, there is broad uncertainty over how to interpret---or whether to even acknowledge---non-preregistered results. In this paper, we consider the case of a truth-seeking researcher who, after seeing the data, earnestly wishes to report additional estimates alongside those preregistered in their PAP. We show that, even absent ``nefarious'' behavior, conventional confidence intervals and point estimators are invalid due to the fact that non-preregistered estimates are only reported in a subset of potential data realizations. We propose inference procedures that account for this conditional reporting. We apply these procedures to \cite{bessone2021economic}, which studies the economic effects of increased sleep among the urban poor. We demonstrate that, depending on the reason for deviating, the adjustments from our procedures can range from having no difference to an economically significant difference relative to conventional practice. Finally, we consider the robustness of our procedure to certain forms of misspecification, motivating possible heuristic checks and norms for journals to adopt.
\end{abstract}

\newpage

\section*{Introduction}
Pre-analysis plans (PAPs) are time-stamped, publicly accessible documents that preregister hypotheses, experimental design, and planned statistical analyses before data are observed. Originating in clinical trials, PAPs have been championed as a key tool for improving credibility and transparency in experimental economics \citep{clinicalPAPs2004,olken2015promises}. Following the 2013 launch of the American Economic Association’s RCT Registry—a central repository tracking ongoing, completed, and withdrawn trials—the number of preregistrations with PAPs has grown remarkably \citep{ofosu2023pre}. Despite this growth, uncertainty remains about their requisite detail and comprehensiveness, as well as the extent to which researchers should be allowed to deviate from preregistered choices \citep{olken2015promises,banerjee2020praise}. In a 2023 survey of experimental economists, \citet{imai2025} report that while 83\% support deviations provided transparent disclosure, views diverge sharply on acceptable scope—ranging from unrestricted latitude to narrowly circumscribed departures. Overall, 74\% expressed a desire for clearer norms and guidance for how PAPs should be drafted and reviewed.

At root, the core tension underlying arguments for and against PAPs turns on the permitted scope of deviation: On the one hand, strict adherence limits researchers' degrees of freedom for data mining, specification searching, and ex post rationalization \citep{Simmons2011FalsePositive,Wicherts2016dof}. Formally, strict adherence sets the analysis rule ex ante, ensuring nominal frequentist guarantees (e.g., size, coverage) hold under identical replications of the experiment \citep{olken2015promises,kasy2023optimal}.
On the other hand, strict adherence constrains researchers' flexibility to handle unforeseen issues that arise during implementation, as well as ability to explore novel insights that emerge over the course of the study. Considering the high overhead costs to running a new experiment, not pursuing auxiliary analyses when the data suggest them also comes at a serious social cost \citep{Miguel2021TransparencyJEP, CoffmanNiederle2015, olken2015promises, banerjee2020praise}. In practice, writing a PAP that is sufficiently exhaustive to anticipate all possible contingencies is extremely time and labor intensive, and deviations are extremely common \citep{OfosuPosner2020PAPsPAndP, ofosu2023pre, Brodeur2024reduce}. The pervasiveness of such deviations presents an urgent need to understand how the reporting of preregistered estimates (henceforth, \textit{PAP estimates}) formally differs from that of non-preregistered estimates (henceforth, \textit{post estimates}), and by extension, how the latter ought to be presented in papers and interpreted by readers.

In this paper, we develop a framework to formalize the statistical properties of deviations from the PAP, thereby allowing us to provide certain inferential guarantees for the resulting post estimates. In doing so, we offer a means to resolve the core tension above: Researchers can still enjoy the benefits of PAP estimates, while leaving open the option to validly report and interpret post estimates.
Our approach frames deviations as instances of \textit{selective reporting}, which differ from prespecified analyses specifically in that the decision to report can depend on realized data. For example, a non-preregistered regression coefficient might be of interest only when it is found to be significant and large; a new specification pooling treatment arms might be reported only if the preregistered specification fails to reject the null; alternately, a non-preregistered economic model might be tested only if the data violate a core assumption of the preregistered model.
In modeling deviations in this way, we formalize an existing norm where researchers are expected to provide a rationale for deviating. Here, this rationale is expressed through what we term the \textit{deviation set}, which corresponds to the subset of data realizations satisfying the conditions for reporting a given post estimate. 
Using the deviation set, we show how one may obtain confidence intervals and point estimators that provide valid inference \textit{conditional} on the event that the researcher deviates. We argue that researchers should leverage this deviation set to report corrected (i.e., conditional) inferences alongside conventional (i.e., unconditional) inferences. We therefore provide general inference procedures for obtaining confidence intervals that have correct conditional coverage and point estimators that are conditionally unbiased.

The first and primary contribution of this paper is the framing of PAP deviations as selective reporting, which provides a formal language for addressing the concerns faced by researchers when reporting post estimates. In viewing deviations through this lens, our framework is able to leverage results and algorithms from the conditional inference literature to correct for the statistical distortions arising from PAP deviations, ensuring valid inference for post estimates. 

To develop our results, we consider a model where the estimates are normally distributed with a known covariance matrix. The inference procedures that we propose are based on the observation that, since a deviation only occurs in the data realizations from its deviation set, the corresponding post estimates are no longer normally distributed, owing to dependence between the post estimate and its deviation event. In the leading case where the deviation set is based on realizations of PAP estimates, this dependence is governed by the correlation between the PAP and post estimates. After accounting for this dependence, the distribution of post estimates corresponds to a class of truncated normal distributions, with truncation regions that depend on the deviation set. We show that for a broad class of polyhedral deviation sets, which encompass the deviations that occur in practice, this yields computationally tractable test statistics from which conditionally valid confidence intervals and point estimators can be derived---results from \cite{Pfanzagl1994} further imply the optimality of these procedures for our setting.

As a second contribution, we relax the assumptions of normality and known covariance to show that feasible analogues of the proposed procedures---where one plugs in asymptotically normal estimates and consistent covariance matrix estimators---are asymptotically valid in large samples, uniformly over a broad class of data generating processes. Thus, researchers can safely implement these plug-in procedures, which we detail explicitly. While such uniformity results have been established in \citet{markovic2017unifying}; \citet{tian2017asymptotics}; \citet{andrews2024inference}; and \citet{mccloskey2024hybrid} for the case of one polyhedra, our results establish uniformity for the general case of multiple polyhedra.\footnote{\citet{tibshirani2018uniform} develop uniform asymptotics for a model where the events of interest take the form of unions of polyhedra, but (i) impose that underlying mean parameters shrink in accordance with the sample size and (ii) employ procedures that condition on individual polyhedra before aggregating. The latter can lead to efficiency loss in the normal model, relative to the optimal procedure that we employ \citep{fithian2014optimal}.} In proving uniformity for unions of multiple polyhedra, we allow for the complexity and expressiveness of deviation sets that arise in practice—for instance, conditioning on two-sided significance requires conditioning on a union of two polyhedra. This uniformity result is of independent interest, beyond the PAP setting. For example, the selective inference literature is often interested in conditional inference for lasso model coefficients conditional on the lasso selection event, which can be represented as a union of polyhedra \citep{lee2016exact}.

We apply our framework to \citet{bessone2021economic}, which examines the economic effects of increased sleep among the urban poor. We construct three plausible deviation sets, based directly on text from the paper, and show that—depending on the deviation set and the realized data—accounting for conditional reporting can range from no change to economically meaningful departures from conventional practice. Intuitively, when the realized data lie near a reporting boundary implied by the deviation set, conditioning can have a large impact; when the data lie comfortably inside the region, the adjustments to conventional inference are minimal.

As an additional contribution, we examine the robustness of our procedure to certain forms of misspecification of the deviation set, encompassing cases of both strategic misreporting and earnest mistakes. First, we show that inference remains valid—albeit less precise—when the researcher can only describe a subset of their deviation set local to the realized data. Second, we present sensitivity analysis to assess robustness across alternative deviation sets, motivating potential norms that journals and researchers should adopt, such as agreed-upon significance cutoffs for reporting non-prespecified findings (e.g., a default significance level of 5\%).

To ease the adoption of our approach, we give examples of how to specify deviation sets for the most common rationales, demonstrating that the costs to formal articulation are low. For example, if a post estimate is reported because it is significant at the $5\%$ level, then the deviation set is characterized by the significance cutoff. We show that an immediate implication of our framework is that conventional and conditional inferences are equivalent when deviations occur (i) for reasons independent of the data (e.g., arrival of new econometric methods, honest mistakes, or oversights in PAP specification) or (ii) in events that are independent of the post estimates (e.g., deviations based on the results from a pilot study on an independent sample). In these special cases, there is no need to adjust one's conventional inferences.

\paragraph{Related Literature.}

Our paper contributes to an ongoing dialogue among economists on the costs and benefits of PAPs; see, for instance, \citet{McKenzie2012}, \citet{Humphreys2013}, \citet{CoffmanNiederle2015}, \citet{olken2015promises}, \citet{Glennerster2017}, \citet{ChristensenMiguel2018}, \citet{ChristensenFreeseMiguel2019}, \citet{banerjee2020praise}. 
These discussions are accompanied by an active literature meta-analyzing the empirical success of PAPs in achieving their stated goals of curbing publication bias and $p$-hacking, such as \cite{BrodeurCookHeyes2020AERMethodsMatter}, \cite{broduer2023unpack}, \cite{ofosu2023pre}, and \cite{Brodeur2024reduce}. On the whole, existing work has primarily centered around questions of (i) whether to use a PAP in the first place and (ii) its appropriate length, detail, and contents. Less explored, however, are the specific consequences of deviating from a PAP \textit{after} it is already preregistered (i.e., taken as fixed) and data observed. This question is the focus of our paper, isolated from considerations of publication bias and $p$-hacking.

There is also a literature discussing optimal PAP construction \citep{ludwig2019augmenting, banerjee2020theory, anderson2022highly, kasy2023optimal}. Our conditional inference framework instead takes the PAP as given and asks how one can validly interpret deviations from the PAP. That said, while we do not make formal statements in this paper about optimal PAP specification, we do provide suggestions to researchers and journals for potential norms to adopt. We leave formal consideration of the implications of our framework for optimal PAP construction to future work.

Methodologically, this paper draws from the conditional inference literature, particularly from work on optimal conditionally quantile-unbiased estimation for polyhedral conditioning events \citep{pfanzagl1979optimal, fithian2014optimal, lee2016exact, andrews2024inference, mccloskey2024hybrid}. The conditionally quantile-unbiased estimators form the building blocks for our proposed confidence intervals and point estimators. In terms of model setup, the asymptotic framework underlying our uniformity result most closely resembles the approaches of \cite{andrews2024inference} and \cite{mccloskey2024hybrid}, whose results can be used to prove the case of a single polyhedra. Here we prove the general case of multiple polyhedra.

\paragraph{Outline.} In Section \ref{d2:sec:setting}, we begin with a description of the research process surrounding PAP construction and deviations therefrom, using \citet{bessone2021economic} as an illustrative application. In Section \ref{d2:sec:model}, we introduce notation to formalizes this process; moreover, we interpret deviations through the lens of a simple Bayesian decision problem---though our inference results are valid without the Bayesian structure. In Section \ref{d2:sec:inference}, we introduce our conditional inference objectives in a model with normally distributed estimates, and present general procedures for constructing confidence intervals and point estimators that are conditionally valid. In Section \ref{d2:sec:feasible.inference}, we show how to implement these procedures in practice, allowing for non-normal estimates. In Section \ref{d2:sec:application}, we apply our results to \citet{bessone2021economic}. In Section \ref{d2:sec:robustness}, we discuss robustness of our proposed procedures to misspecification of the deviation set. We provide all proofs and supplementary results in the Appendix.

\section{Setting}\label{d2:sec:setting}
We begin with a general description of how PAPs are used in the status quo, accompanied by examples from the experimental economics literature. This descriptive timeline will motivate the formal model presented in Section \ref{d2:sec:model}. Throughout this paper, we focus on PAPs in the context of randomized control trials (RCTs), as (i) this is where PAPs are most used in economics, and (ii) it offers a well-defined division between preregistration and data collection. That said, these discussions in principle also apply to empirical research using preexisting datasets \citep{burlig2018improving,Miguel2021TransparencyJEP}, though such cases generally raise the additional challenge of verifying that researchers are indeed preregistering studies before analyzing the data\footnote{There do exist non-experimental settings where researchers are required to prespecify analyses prior to data access. For instance, restricted-use microdata such as those accessed through the U.S. Census Bureau often require submission and approval of detailed research proposals before any data extracts are made, with each additional data request or modification of scope typically necessitating a new proposal.} \citep{ChristensenMiguel2018}. To avoid distracting from the central focus of this paper, we take the delineation of ``pre'' and ``post'' data collection as given.

\paragraph{Priors.} Researchers embark on projects with priors informed by sources such as economic theory, existing literature, outcomes of a pilot study, discussions with experts, or first-hand observation. For example, \cite{bessone2021economic} benchmark the anticipated magnitude of their results as follows:
\begin{displayquote}
   ``While experimental evidence on the effect of increasing sleep in field settings is scarce, there is a widely held belief among researchers and the public that reducing sleep deprivation would lead to improvements in economic outcomes \citep{walker2017sleep}. To document these priors, we surveyed 119 experts from sleep science and economics who predicted sizable economic benefits, including a 7\% increase in work output, of increasing sleep by half an hour a night from the low levels observed in our setting.''
\end{displayquote}
At the highest level, priors can inform the researcher's choice of which hypotheses to test, estimands to consider, and outcomes to measure given theoretical salience or expected precision and effect size. Because empirical projects are constrained in scale and scope, priors can further discipline downstream design choices: how large of a sample to field; how many waves of follow-up to fund; which populations to emphasize; whether to invest in costly data enhancements; and how to allocate statistical power between estimating overall effects and probing heterogeneity or mechanism analysis.

\paragraph{Preregistration.} In advance of data collection, the researcher preregisters information about their study such as primary outcomes, experimental design, randomization method, randomization unit, clustering, and sample size (including total number of observations, number of clusters, and units per treatment arm). In addition to this information, the researcher may also register a PAP stipulating which specifications will be used to estimate quantities of interest. In practice, depending on the research topic and journal, there is variation in whether a PAP is required, and if so, to what level of specificity. At its simplest, a PAP might consist of a single list of linear regressions. On another extreme, a PAP might take the form of a complete contingent plan which maps every possible data realization to a particular set of specifications, mirroring the notion of a complete strategy in game theory. 

To illustrate, consider a simple example, shown in Figure \ref{fig:pap-tree}. After estimating the average treatment effect (ATE) of a deworming intervention on school attendance, a complete contingent PAP might stipulate that if the attendance effect is statistically insignificant, the researcher will analyze a particular set of secondary outcomes such as health measures (e.g., incidence of anemia, weight-for-age) to assess whether the intervention improved child well-being through channels other than schooling. By contrast, if the effect \textit{is} significant, the PAP specifies a set of persistence checks to perform, such as attendance in the following academic year.

\usetikzlibrary{arrows.meta,positioning,shapes.geometric}
\begin{figure}[htbp]
  \centering
    \begin{tikzpicture}[
      node distance=12mm and 10mm,
      >=Latex,
      every node/.style={font=\small},
      start/.style={rectangle, rounded corners, draw, align=center, inner sep=3pt, text width=3.6cm},
      decision/.style={diamond, draw, aspect=1.4, align=center, inner sep=1pt, text width=2.0cm}, 
      block/.style={rectangle, draw, align=center, inner sep=3pt, text width=5.0cm}, 
      yes/.style={midway, above, sloped, font=\small},
      no/.style={midway, below, sloped, font=\small}
    ]
    \node[start] (stage1) {\textbf{Stage 1: Primary Analysis}\\
    Estimate ATE of deworming on\\
    \emph{school attendance}};
    \node[decision, right=of stage1] (sig) {Attendance\\ effect\\ significant?};
    \node[block, above right=6mm and 20mm of sig] (persist) {\textbf{Stage 2A: \\ Persistence Checks}\\
    Estimate ATE on attendance in following academic year};
    \node[block, below right=6mm and 20mm of sig] (health) {\textbf{Stage 2B: Alt. Channels}\\
    Analyze other health outcomes (e.g., anemia incidence, weight-for-age)};
    \draw[->] (stage1) -- (sig);
    \draw[->] (sig) -- (persist) node[yes] {Yes};
    \draw[->] (sig) -- (health) node[no] {No};
    \end{tikzpicture}
    \caption{Complete contingent PAP for deworming intervention.}
  \label{fig:pap-tree}
\end{figure}
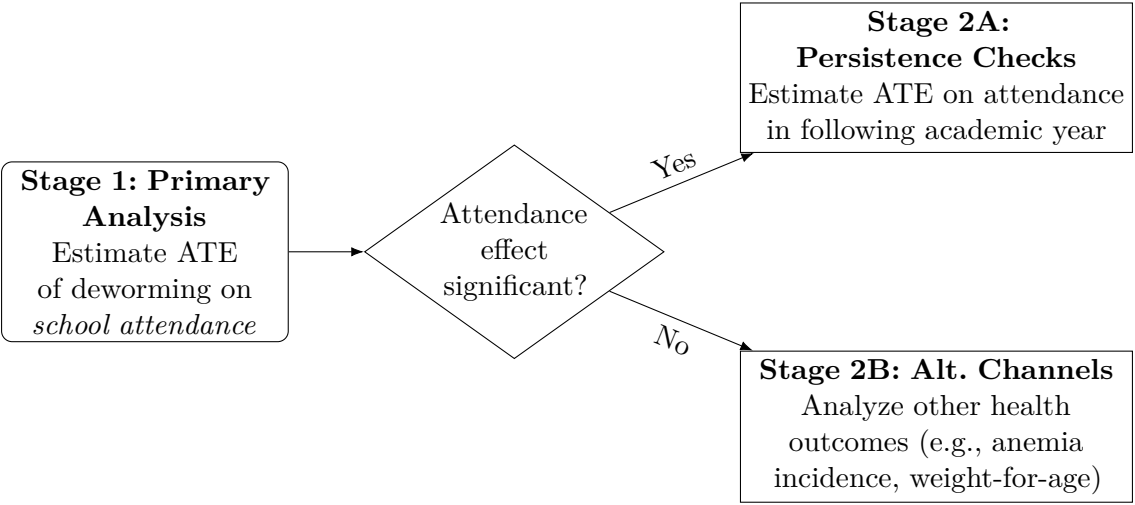
In the sample PAP above, the researcher not only specifies their primary analysis, but also their intended secondary analysis for both possible contingencies based on their initial findings. Even the simplest of empirical exercises, however, can be far more complex than the contingencies shown in Figure \ref{fig:pap-tree}. Operationally, most PAPs in practice lie somewhere in the middle: The researcher has a core set of primary specifications and also stipulates how these specifications will be adjusted for some—but not all—contingencies. 

An example of a partial contingent plan can be found in the (31-page) PAP for \citet{bessone2021economic}. One component of their analysis tests whether sleep affects how attentive individuals are to earnings-related incentives. The authors approached this question by introducing a ``salience variation'' in a particular data-entry task. The plan was first to establish that participants respond more strongly when incentives are presented in a high-salience rather than a low-salience way, and then to assess whether the sleep intervention reduces this gap.\footnote{If individuals are fully attentive, their responses to incentives should be the same whether the incentives are displayed in a high- or low-salience way. A larger gap between the two conditions therefore signals inattention. As sleep is hypothesized to improve alertness, if the treatment reduces this gap, that provides evidence that sleep mitigates inattention to incentives.} However, it was unknown ex ante whether (i) the salience manipulation would succeed and (ii) what form the response would take. As directly stated in their PAP:
\blockquote{
``We do not fully pre-specify this analysis since the appropriate analysis will depend upon first
establishing that the salience variation works ``as intended''—that is, the response to incentives is stronger under high-salience relative to low-salience. Moreover, the precise form of the response to salience—e.g. whether individuals notice and respond to the incentive change in 5 minutes versus 30 minutes—remains unknown at the time of this pre-registration, making it difficult to write down the full contingent analysis plan.''
}

\paragraph{Data Collection and Analysis.} After completing the preregistered analyses, a researcher may wish to go beyond the PAP and compute additional results. Common motivations might include exploratory work prompted by unexpected findings, knowledge gained during data collection, new statistical tools, or unanticipated issues with the original specification. As an instance of the latter, \citet{bessone2021economic} write:
\blockquote{
``We preregistered daily net savings as our main variable of interest. However, this measure suffers from an unanticipated design issue: participants make large one-time withdrawals right before the study ends, which mechanically drives down net savings. We believe deposits more accurately reflect differences in savings behavior, and the accrued interest captures the benefit of savings.''
}
More generally, because PAPs reflect prior beliefs about where the most interesting or strongest effects are likely to lie, the desire for additional analysis is particularly likely when the realized data substantially shift the researcher's prior. The next section formalizes this idea.

\paragraph{Reporting of Results and Deviations from PAP.} The researcher publicly reports the PAP and the results of all preregistered analyses.\footnote{We assume that PAP results are always reported (i.e., the researcher does not \textit{hide} any findings). For simplicity, we also do not model degrees of emphasis in reporting. In our framework, placing PAP results in an appendix is equivalent to including them in the main text.} At this time, there is broad uncertainty in the profession about whether and how to report \emph{non}-preregistered results. Prescriptions range from strict ``no reporting'' to more permissive approaches that allow reporting conditional on clearly flagging non-PAP findings as exploratory or suggestive.

We refer to the reporting of a non-preregistered result as a \emph{deviation} from the PAP. Uncertainty over their interpretation notwithstanding, deviations are extremely common in practice and arise for many reasons: understanding surprising results (exploration or mechanism probes), referee and editor requests (e.g., robustness or placebo checks), or the availability of improved statistical tools since PAP submission. We use \citet{bessone2021economic} as a running example of such deviations.

\begin{example}\label{d2:ex:sleep} 
\citet{bessone2021economic} study the economic effects of different policies targeted at increasing sleep among the urban poor. The authors cross-randomize workers in Chennai, India, to receive (i) one of two treatments designed to improve nighttime sleep and (ii) a treatment that provides nap breaks during workdays. 
\citet{bessone2021economic} preregister a set of specifications regressing indices of work, well-being, and cognition outcomes on interactions of the night sleep and nap treatments. However, as the authors describe in their paper,
\begin{displayquote}
``Those who received a night sleep treatment in addition to naps had very similar effects to those with naps only. This pattern of results suggests that naps have an overall positive effect on outcomes, whereas increases in night sleep do not. To increase statistical power and streamline the discussion of effects on the individual outcomes, we turn to an analysis that pools the two night sleep treatments and does not allow for an interaction effect of night sleep and nap treatments.'' 
\end{displayquote}
\end{example}
This pooling decision is substantively motivated: When two night-sleep arms are practically indistinguishable, combining them yields a more precisely estimated quantity. However, at this stage there is uncertainty in the profession over how to proceed: A blanket ban on deviations would preclude learning from the pooled results, yet reporting conventional (unconditional) point estimators and confidence intervals would ignore that the pooling decision was made selectively after seeing the data, risking an overstatement of precision. While the researcher can flag such results as exploratory and meant to be taken with a grain of salt, it is still ambiguous how to interpret deviations from the PAP. In the coming sections, we express the stages outlined above with a formal model, then characterize the adjustments that should be made to account for the selective reporting of non-preregistered results.

\section{Model}\label{d2:sec:model}
We begin by describing the environment and introducing terms in Section \ref{d2:sec:environment}, formalizing the researcher's timeline discussed in the previous section. To ground motivation, in Section \ref{d2:sec:decision.problem} we contextualize the researcher's choice of specifications in terms of a Bayesian decision problem and show that even a fully Bayesian decision maker may wish to deviate if constrained to choose from a limited set of specifications. We emphasize, however, that our subsequent inferential framework in Section \ref{d2:sec:inference} is valid without this Bayesian structure. 

\subsection{Environment}\label{d2:sec:environment}
The researcher observes data $X \in \X$, for $\X$ a sample space. The researcher considers specifications $S: X \to \R^{d}$ from some set $\cS$; that is, each specification $S \in \cS$ maps data $X$ to estimates $S(X)$. An example of $\cS$ could be the set of all linear regressions.\footnote{This setup accommodates cases where the researcher considers regression specifications with differing numbers of regressors, since vectors in $\R^{\Tilde{d}}$ for $\Tilde{d} < d$ can also be represented as vectors in $ \R^{d}$.} In advance of observing the data, the researcher preregisters a set of specifications $\cSpre \subseteq \cS$, where $\cSpre$ represents the pre-analysis plan (PAP). We denote specifications from the PAP by $\Spre \in \cSpre$. For example, $\cSpre$ might be the specific set of linear regressions the researcher intends to run, and $\Spre$ one such regression.\footnote{Observe that one could alternately define $\cS$ to be the set of all possible PAPs, from which the researcher chooses a single element $\Spre\in\cS$. Our choice of notation, while mathematically equivalent, is more practical for when we discuss inference in Section \ref{d2:sec:inference}, and seems better aligned with how researchers view PAPs.} After observing $X$, the researcher announces a set of non-preregistered specifications $\cSpost \subseteq \cS \setminus \cSpre$. We refer to $\cSpost$ as \textit{deviations from the PAP}, which may or may not depend on the observed data $X$. The researcher then reports (i) PAP estimates $\hbpre = \Spre(X)$ for all PAP specifications $\Spre \in \cSpre$ and (ii) post estimates $\hbpost = \Spost(X)$ for all deviations $\Spost \in \cSpost$. We depict this timeline in Figure \ref{d2:fig:timeline}. To ground this notation, we revisit the example of \cite{bessone2021economic}.

\begin{figure}
    \centering
    \doubleline
    \begin{enumerate}
    \item Register PAP $\cSpre$ consisting of specifications $\Spre$
    \item Observe data $X$
    \item Announce deviations $\cSpost$ consisting of non-preregistered specifications $\Spost$
    \item Report PAP estimates $\hbpre = \Spre(X)$ and post estimates $\hbpost = \Spost(X)$
    \end{enumerate}
    \doubleline
    \caption{Timeline of Researcher's Process}
    \label{d2:fig:timeline}
\end{figure}

\examplecontinues{d2:ex:sleep} \citet{bessone2021economic} cross-randomize workers to receive (i) one of two treatments designed to improve nighttime sleep (``Incentives'' and ``Encouragement'') and (ii) a treatment that provides nap breaks during workdays (``Nap''). They preregister a set of regression specifications $\cSpre$ regressing indices of work, well-being, and cognition outcomes on interactions of the two night sleep and nap treatments.

After conducting the study and observing data $X$, the authors compute the treatment effects of the five fully disaggregated treatment arms on preregistered outcomes, $\hbpre = \Spre(X)$. However, after observing that ``those who received a night sleep treatment in addition to naps had very similar effects to those with naps only,'' the authors reason that ``naps have an overall positive effect on outcomes, whereas increases in night sleep do not.'' To increase statistical power in discussion of effects on outcomes, the authors estimate an additional set of specifications $\Spost\in\cSpost$ which (i) pools the two night sleep treatments and (ii) removes the interaction term for night sleep and nap treatments, yielding $\hbpost = \Spost(X)$. The authors supplement their report of the PAP estimates $\hbpre$ (disaggregated) with post estimates $\hbpost$ (pooled), found in Tables III and IV of their paper, respectively.
\vspace{1em}

\subsection{Decision Problem and Deviations}\label{d2:sec:decision.problem}
There are many reasons that deviations arise. In this section, we show that even a fully Bayesian researcher in this framework will sometimes want to deviate when they are limited to a constrained set of specifications for constructing their PAP.\footnote{This model is for motivation alone: The inference procedures in Section \ref{d2:sec:inference} are valid without this Bayesian structure.} We offer discussion on where such constraints arise in practice, such as through labor and communication costs, tractability, cognitive limitations, and norms in the profession.

The distribution of $X$ is governed by the parameter $\theta \in \Theta$, $X\sim P_{\theta} \in \Delta(\X)$, for $\Theta$ a parameter space and $\Delta(\X)$ the set of probability distributions on $\X$. The researcher has a loss function  $L(S(X), \theta) \geq 0$ that quantifies the consequences of reporting estimates $S(X) \in \R^{d}$ when the true parameter is $\theta \in \Theta$. For example, if the goal is to estimate a treatment effect $\tau(\theta) \in \R$, then a standard loss function is squared error:
\begin{equation}\label{d2:eq:squared.error}
    L(S(X), \theta) = (S(X) - \tau(\theta))^{2}.
\end{equation}
The \textit{risk} of specification $S$ is the expected loss from reporting $S(x)$ across possible realizations $x \in \X$ of the data:
\begin{equation*}
    R(S, \theta) =  \int_{\X} L(S(x), \theta) dP_{\theta}(x).
\end{equation*}
Note that the risk depends on the unknown parameter $\theta \in \Theta$.

\paragraph{PAP Construction.} We consider the researcher's problem of choosing the set of specifications $\cSpre, \cSpost \subseteq \cS$ to estimate, focusing on the case where $|\cSpost| \leq |\cSpre|=1$ for simplicity. We suppose the researcher has prior beliefs about the unknown parameter $\theta$, represented by a prior distribution $\pi \in \Delta(\Theta)$. The researcher determines which specification $\Spre$ to preregister in their PAP by minimizing the average risk of $S$ under their prior $\pi$. Stated formally, given a set of specifications $\cS$, loss function $L$, and prior $\pi$, the researcher solves:
\begin{equation}\label{d2:eq:pre.problem}
    \Spre \in \argmin_{S \in \cS} \Bar{R}(S, \pi), \quad \Bar{R}(S, \pi) = \int_{\Theta} R(S, \theta) d\pi(\theta),
\end{equation}
where $\Bar{R}(S, \pi)$ is the average risk of $S$ under $\pi$. After solving this optimization problem, the researcher preregisters $\Spre$ in their PAP.    

\paragraph{Deviations from PAP.} The researcher observes $X$ and reports PAP estimates $\hbpre = \Spre(X)$, as planned. However, the data also lead the researcher to update their prior beliefs $\pi$ to form posterior beliefs, represented by the posterior distribution $\pi_{X} \in \Delta(\Theta)$ of $\theta$ conditional on $X$. Having gleaned new information about $\theta$, the researcher may now find their preregistered $\Spre$ to be suboptimal, leading the researcher to also report post estimates $\hbpost = \Spost(X)$, where $\Spost \neq \Spre$. We formalize this decision to deviate below.

The analogous problem to minimizing prior average risk (\ref{d2:eq:pre.problem}) is to minimize posterior average loss:
\begin{equation}\label{d2:eq:post.Bayes}
    \Spost \in \argmin_{S \in \cS} \Bar{L}(S(X), \pi_{X}), \quad \Bar{L}(S(X), \pi_{X}) = \int_{\Theta} L(S(X), \theta) d\pi_{X}(\theta),
\end{equation}
where $\Bar{L}(S(X), \pi_{X})$ is the average loss of $S(X)$ under $\pi_{X}$.  The researcher deviates from their PAP when $\Spost \neq \Spre$. This decision to deviate depends on the coarseness of $\cS$.

In a conventional Bayesian decision problem, $\cS$ is taken to be the set of all functions from $\X$ to $\R^{d}$. In this case, a specification that minimizes prior average risk (\ref{d2:eq:pre.problem}) also minimizes posterior average loss (\ref{d2:eq:post.Bayes}) almost surely, and vice versa \citep[Theorem 1.1]{lehmann2006theory}. Intuitively, when $\cS$ is unrestricted, a solution to (\ref{d2:eq:pre.problem}) perfectly prepares for every possible data realization $x \in \X$. In such cases, there is never a need to deviate.\footnote{
In Appendix \ref{d2:app:sec:counterfactual.PAP}, we provide a variant of the decision problem where, instead of considering posterior average loss, the researcher solves \eqref{d2:eq:pre.problem} with $\pi_{X}$ in the place of $\pi$. In this setup, $\Spost$ is the \textit{counterfactual PAP} that the researcher would have constructed had they entered the experiment with their posterior beliefs. If one takes $\cS$ to be the relevant action space, this counterfactual PAP approach aligns with the conditional Bayes principle, which says to choose an action that minimizes average loss under one's current beliefs \citep[Section 1.5.1]{berger2013statistical}. Interestingly, deviations arise in this setup even when $\cS$ is unconstrained.}

We depart from the conventional Bayesian setup by allowing for a restricted $\cS$. That $\cS$ is constrained is not a modeling convenience, but a description of practice: At the preregistration stage, researchers face explicit and implicit constraints that narrow the permissible set of specifications. In general, researchers cannot write PAPs sophisticated and lengthy enough to prepare for every $x\in\X$. For one, researchers face \textit{time, cognitive, and communication} constraints: For most experimental designs, it is not feasible to anticipate, formulate, and communicate optimal responses to every contingency \citep{olken2015promises, banerjee2020praise}. As an example, the PAP for \citet{amy2012oregon}—a 50 page paper—is 159 pages long. And, despite being exceptionally detailed, the initial PAP still fell short of exhausting all relevant analyses: From 2012 to 2020, ten additional PAPs (totaling an additional 368 pages) were pre-registered to explore new outcome measures in follow-up papers after collecting additional data \citep{amy_PAP_docs}. 

Norms in the profession also constrain $\cS$, favoring parsimonious, interpretable specifications tied to the design and estimand. By default, this includes conventional functional forms and links (e.g., linear or low-order polynomials, logit/probit for binary outcomes, etc.), standard control sets motivated by identification, and prespecified clustering and variance estimators. Idiosyncratic specifications—such as cube-root transformations, high-order polynomials, post-hoc cutpoints, data-driven bins or knots, bespoke weighting schemes, stepwise selection, or open-ended machine learning screens—are typically viewed with skepticism unless accompanied by a compelling ex-ante justification grounded in theory, measurement scale, or experimental design (and appropriately powered).

The constraints on $\cS$ lead to a classic dynamic inconsistency problem \citep{Strotz1955,KydlandPrescott1977}, wherein a plan that is optimal ex ante under prior beliefs $\pi$ may no longer be optimal after updating to the posterior $\pi_X$.\footnote{In behavioral terms, present-bias and self-control generate the same logic \citep{Laibson1997,ODonoghueRabin1999,FrederickLoewensteinODonoghue2002}. Related commitment perspectives appear in rules-vs-discretion and menu/temptation models \citep{BarroGordon1983,Schelling1960,Kreps1979,GulPesendorfer2001}.} In a standard Bayesian decision problem with a rich $\cS$, the researcher would be allowed to preregister a contingent rule (e.g., ``after $X$ is observed, take whichever action meets criterion $C$''), so the plan chosen ex ante already pins down the ex-post action. When $\cS$ excludes such contingent rules, however, this flexibility is lost. Concretely, if a researcher is restricted to naming one primary outcome in advance (employment or consumption), a contingent rule like ``highlight whichever outcome best meets our welfare/precision criterion after we see $X$'' is not in $\cS$. Ex ante the researcher may pick employment under $\pi$, but ex post the data may favor consumption. Because $\cS$ was insufficiently rich to accommodate switching, the ex ante choice remains binding even when the posterior $\pi_X$ would prefer the other outcome—hence the dynamic inconsistency.

\section{Inference}\label{d2:sec:inference}
This section proposes conditional inference procedures for post estimates. We begin in Section \ref{d2:sec:exact.inference} by introducing the problem of statistical inference in a model with normally distributed estimates. With the relevant notation and concepts established, we next motivate in Section \ref{d2:subsec:motivatingconditional} the use of conditional inference for post estimates and provide examples in Section \ref{d2:sec:deviaton.types}. We then develop the conditional inference procedures in Sections \ref{d2:sec:inference.procedures}--\ref{d2:sec:polyhedral.conditioning} and provide an explicit example in Section \ref{d2:sec:simple.example}. 
The normality assumption is motivated by standard large sample approximations, such as the central limit theorem and the law of large numbers. In Section \ref{d2:sec:feasible.inference}, we describe how to implement analogues of the proposed conditional inference procedures using sample estimates and establish their uniform asymptotic validity.

\subsection{Inference Setup}\label{d2:sec:exact.inference}
Let $\hball$ be a vector that contains (i) the PAP estimates $\hbpre = \Spre(X)$ for each PAP specification $\Spre \in \cSpre$ and (ii) the post estimates $\hbpost = \Spost(X)$ for each deviation $\Spost \in \cSpost$ made under $X$. We assume that $\hball$ is normally distributed with unknown mean $\ball$ and known positive definite covariance matrix $\Sigmaall$:
\begin{equation}\label{d2:eq:exact.normality}
\begin{split}
    \hball \sim
    N(\ball, \Sigmaall).
\end{split}
\end{equation}
Let $\P[\ball]{\cdot}$ and $\EP[\ball]{\cdot}$ denote probabilities and expectations under this distribution.\footnote{We implicitly assume that estimators and confidence intervals of interest depend on $X$ through $\hball$.} Each $\hbpre$ corresponds to a PAP estimand $\bpre = \EP[\ball]{\hbpre}$, while each $\hbpost$ corresponds to a post estimand $\bpost = \EP[\ball]{\hbpost}$. By construction, $\ball$ contains all the PAP and post estimands. We consider the problem of conducting inference on parameters that can be expressed as linear combinations of these estimands: $v'\bpre$ and $l'\bpost$, where $v, l \neq 0$.\footnote{This setup is quite general, since differentiable nonlinear functions of asymptotically normal estimates will also be asymptotically normal by the delta method. These nonlinear functions can be elements of $\hball$.} For instance, in Example \ref{d2:ex:sleep}, the effects of the various treatments correspond to different contrasts of the regression coefficients.

\paragraph{Inference for PAP Estimates.} Given PAP estimates $\hbpre$, let $\Sigmapre$ denote the covariance matrix induced by $\Sigmaall$. By equation (\ref{d2:eq:exact.normality}), any linear combination of PAP estimates is normally distributed:
\begin{equation*}
     v'\hbpre \sim N(v'\bpre,  \sigma_{pre}^{2}), \quad \sigma_{pre}^{2} = v'\Sigmapre v.
\end{equation*}
Given a desired significance level $\alpha \in (0,1)$, a confidence interval $CI_{\alpha}(X)$ has correct coverage for $v'\bpre$ across data realizations $X \in \X$ if it satisfies
\begin{equation}\label{d2:eq:PAP.coverage}
    \P[\ball]{v'\bpre \in CI_{\alpha}(X)} = 1 - \alpha, \quad \forall \ball.
\end{equation}
In words, a confidence interval $CI_{\alpha}(X)$ that satisfies criteria \eqref{d2:eq:PAP.coverage} contains the parameter of interest $v'\bpre$ with probability $1-\alpha$ for any possible value of the unknown $\ball$. This is a standard criteria for valid inference \citep{lehmann2024testing}. One example is the conventional two-sided interval $[v'\hbpre \pm z_{1-\alpha/2}\sigma_{pre}]$, where $z_{\alpha}$ denotes the $\alpha$-quantile of the $N(0,1)$ distribution. This interval satisfies
\begin{equation*}
    \P[\ball]{v'\bpre \in [v'\hbpre \pm z_{1-\alpha/2}\sigma_{pre}]} = 1 - \alpha, \quad \forall \ball,
\end{equation*}
and therefore yields valid inference for PAP estimates.

\paragraph{Inference for Post Estimates.}
Unlike the preregistered specifications $\cSpre$, the researcher's deviations $\cSpost$ may depend on the observed data $X$. We say that a data realization $X$ \textit{induces} the deviation $\Spost$ if $\Spost \in \cSpost$ in that data realization. For a given $\Spost$, we may then collect the set of all data realizations that induce the deviation:
\begin{align*}
    \Xpost = \{X \in \X: \Spost \in \cSpost\}.
\end{align*}
We refer to $\Xpost$ as the \textit{deviation set} for $\Spost$. By definition, the researcher is only interested in post estimates $\hbpost$ when $X\in\Xpost$. In view of this, we propose confidence intervals that satisfy the following criteria: Given significance level $\alpha \in (0,1)$, a confidence interval $CI_{\alpha}(X)$ has correct \textit{conditional coverage} for $l'\bpost$ given $X \in \Xpost$ if it satisfies
\begin{equation}\label{d2:eq:cond.coverage}
    \P[\ball]{l'\bpost \in CI_{\alpha}(X) \mid X \in \Xpost} = 1 - \alpha, \quad \forall \ball.
\end{equation}
Intuitively, such intervals ensure correct coverage for $l'\bpost$ across the set of data realizations $\Xpost \subseteq \X$ where the researcher is interested in $\hbpost$. This conditional coverage criteria generalizes the \textit{unconditional} criteria in \eqref{d2:eq:PAP.coverage}. Indeed, since PAP estimates are always reported, the analogous ``deviation set'' for a PAP estimate yields $\{X \in \X: \Spre \in \cSpre \} = \X$, in which case conditional and unconditional coverage coincide.

\paragraph{Known Deviation Set.} To obtain $CI_{\alpha}(X)$ with correct conditional coverage, we assume the researcher can (and does) articulate the deviation sets $\Xpost$ for each $\Spost \in \cSpost$ at hand. This is a natural first step for deriving conditional inference procedures. We relax this assumption in Section \ref{d2:sec:robustness}, where we discuss robustness of conditional inference procedures to various forms of misspecification.
To make the above concepts concrete, we return to the example deviation in \cite{bessone2021economic}. In the following example, let $\text{se}(\cdot)$ denote the standard deviation of an estimator under normal distribution \eqref{d2:eq:exact.normality}.

\examplecontinues{d2:ex:sleep} 
\cite{bessone2021economic} initially preregistered the fully interacted specification $\hbpre = (\htau_{N}, \htau_{NE}, \htau_{NI}, \htau_{E}, \htau_{I})'$ , where \begin{itemize}
\itemsep=0.04em
    \item $\htau_{N}$: effect of naps only
    \item $\htau_{NE}$: effect of  naps $\times$ encouragement 
    \item $\htau_{NI}$: effect of naps $\times$ incentives
    \item $\htau_{E}$: effect of encouragement only
    \item $\htau_{I}$: effect of incentives only
\end{itemize}
After seeing the data and estimating $\hbpre$, the authors report three additional post estimates, $\hbpost = (\htau_{E+I}, \htau_{NE+NI}, \hat{\tau}_{N+NE+NI})'$, corresponding to the estimates which (i) pool incentives and encouragement into a single ``night sleep'' treatment ($\htau_{E+I}$ and $\htau_{NE+NI}$) and (ii) collapse interaction effects ($\hat{\tau}_{N+NE+NI}$).\footnote{The estimate $\htau_{NE+NI}$ is not reported in the main paper, but rather in Online Appendix Table A.VIII (where the authors pool the two night sleep treatments but include a separate indicator for individuals who received a combination of either night sleep treatment along with the nap treatment).} The authors find that, ``these results ($\hbpre$) provide evidence that naps have an overall positive effect on outcomes, while increases in night sleep do not. However... this analysis has limited statistical power.'' The authors justify turning to a ``simplified but higher-powered version of this analysis'' with the following three reasons:
\begin{enumerate}\itemsep=0em
    \item[(A)] ``Each of the night sleep treatments alone had no significant effect on this overall index,'' i.e.,
    $$\frac{\abs{\hat{\tau}_{E}}}{\text{se}(\hat{\tau}_{E})} \leq z_{1-\eta/2} \,\text{ and }\, \frac{\abs{\hat{\tau}_{I}}}{\text{se}(\hat{\tau}_{I})} \leq z_{1-\eta/2}.$$
    \item[(B)] ``In contrast, naps alone had a positive, marginally significant effect,'' i.e., 
    $$\frac{\abs{\hat{\tau}_{N}}}{\text{se}(\hat{\tau}_{N})} \geq z_{1-\eta/2}.$$
    \item[(C)] ``Those who received a night sleep treatment in addition to naps had very similar effects to those with naps only,'' i.e.,
    $$\frac{\abs{\hat{\tau}_{N}-\hat{\tau}_{NE}}}{\text{se}(\hat{\tau}_{N}-\hat{\tau}_{NE})} \leq z_{1-\eta/2} \,\text{ and }\,  \frac{\abs{\hat{\tau}_{N}-\hat{\tau}_{NI}}}{\text{se}(\hat{\tau}_{N}-\hat{\tau}_{NI})} \leq z_{1-\eta/2}.$$
\end{enumerate}
Taking these points at face value, a plausible characterization of the authors' deviation set for $\hbpost=(\htau_{E+I}, \htau_{NE+NI}, \hat{\tau}_{N+NE+NI})'$ is the set of data realizations satisfying the intersection of these events, i.e., 
$$\mathcal{X}_{post} = \left\{X \in\mathcal{X}\,:\,\, \text{(A) and (B) and (C)} \right\}.$$

\subsection{Motivating Conditional Inference}\label{d2:subsec:motivatingconditional}

We argue that \textit{conditional} coverage of $l'\bpost$ should be the preferred inferential objective---not unconditional coverage, as is convention---since post estimates are only reported or acted upon in select states of the world. We see two main reasons for this. First, selective reporting can lead to distortions in the density of observed reports of $\hbpost$, as we will demonstrate. Second, conditional validity is, in a certain sense, necessary and sufficient for confidence sets to be valid on average across studies, which is the very property underlying existing frequentist coverage objectives.

\paragraph{Conditional Coverage for Given Study.}
Selective reporting can distort the density of observed reports $\hbpost$ relative to the unconditional density assumed by conventional inference procedures. This is demonstrated in the following toy example.

\begin{example}\label{d2:ex:toy} Let 
\begin{align*}
    \begin{pmatrix}
    \hbpre \\
    \hbpost
    \end{pmatrix}
    \sim
    N(0, \Sigmaall), \quad
    \Sigmaall = 
    \begin{pmatrix} 
    1 & \rho \\ 
    \rho & 1 
    \end{pmatrix},
\end{align*}
where $\rho=0.5$ is the correlation between $\hbpre$ and $\hbpost$. The leftmost panel in Figure~\ref{d2:inference:distortion:fig} shows the joint unconditional density of $(\hbpre, \hbpost)$, i.e., when $\Xpost = \R^{2}.$ The center panel shows the joint density of $(\hbpre, \hbpost)$ conditional on $\hbpre \geq 0$, i.e., when $\hbpost$ is selectively reported based on $\Xpost = \{X \in \R^{2}: \hbpre \geq 0\}$. The rightmost panel superimposes the associated marginal densities of $\hbpost$ under selective (blue) and unconditional (red) reporting regimes. To give a reference of magnitude, the red dashed lines show the conventional 5\% significance test cutoffs. As can be seen, the missing mass in the post-selection joint density causes the density of \textit{observed reports} of $\hbpost$ to look quite different from the unconditional density of $\hbpost$ assumed by conventional inference procedures. By ignoring the correlation between the conditioning event and estimate $\hbpost$, conventional procedures yield inflated type I error, undercoverage, and miscalibrated conclusions.
\begin{figure}
    \begin{minipage}[t]{1.3\textwidth}
    \includegraphics[width=0.30\textwidth]{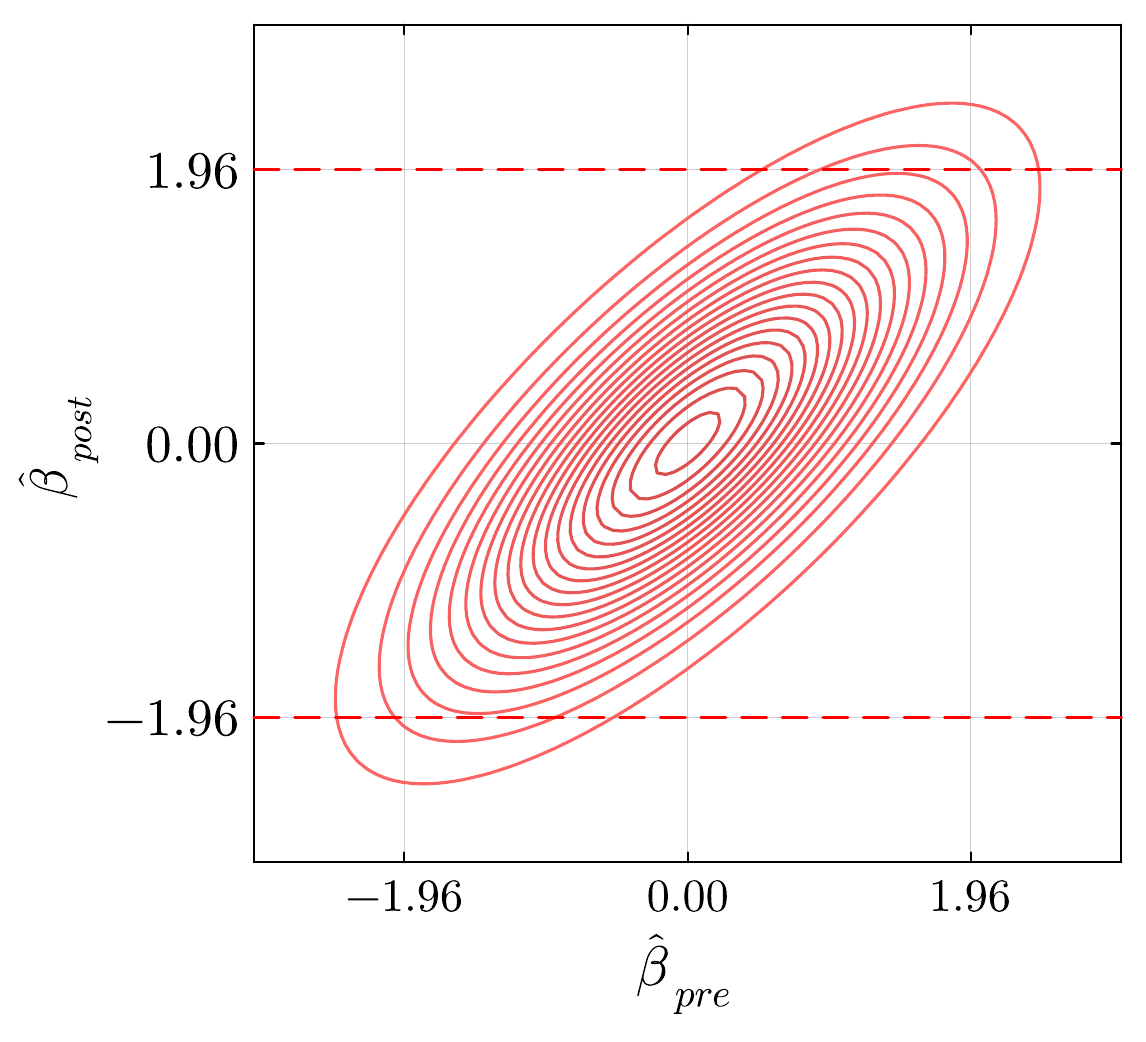}
    \includegraphics[width=0.275\textwidth]{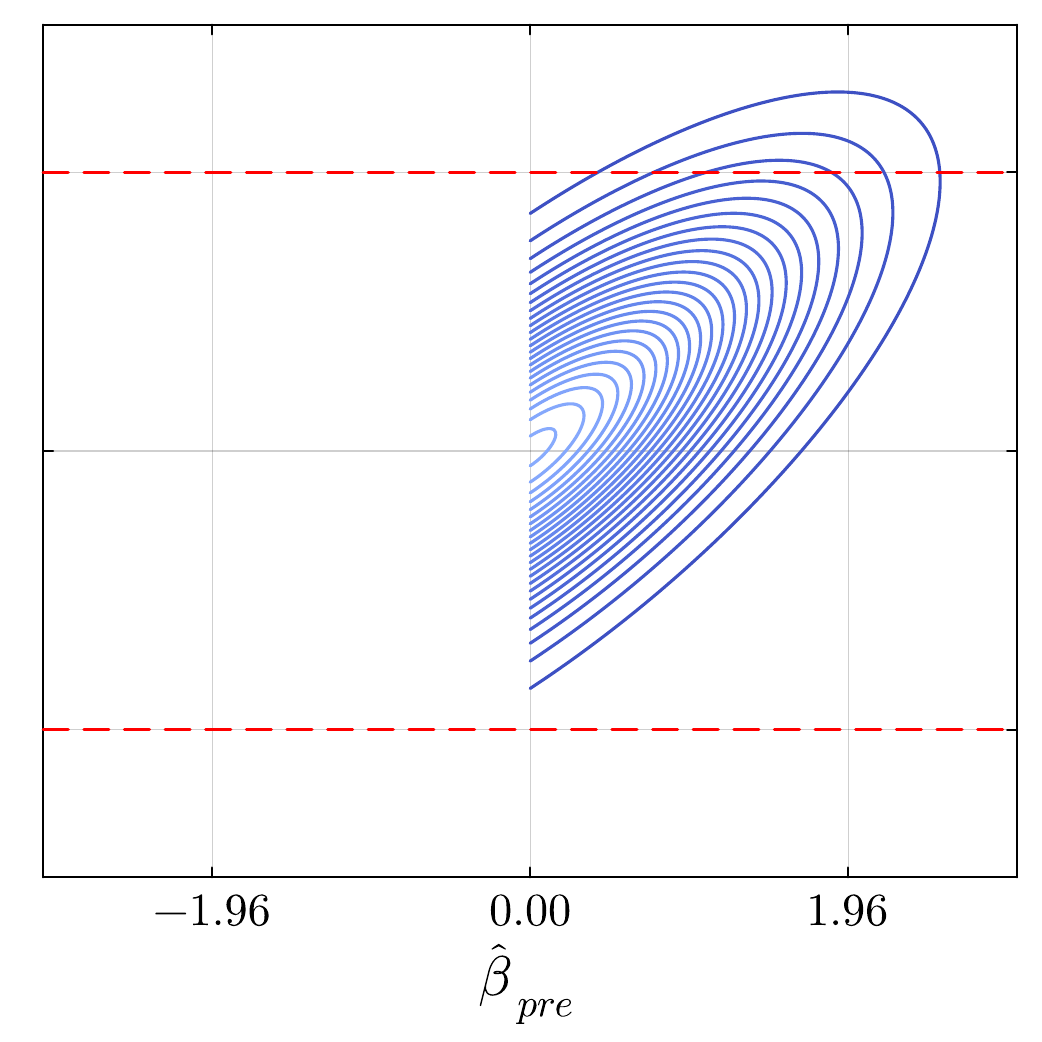}
    \reflectbox{\includegraphics[angle=90, height=0.275\textwidth]{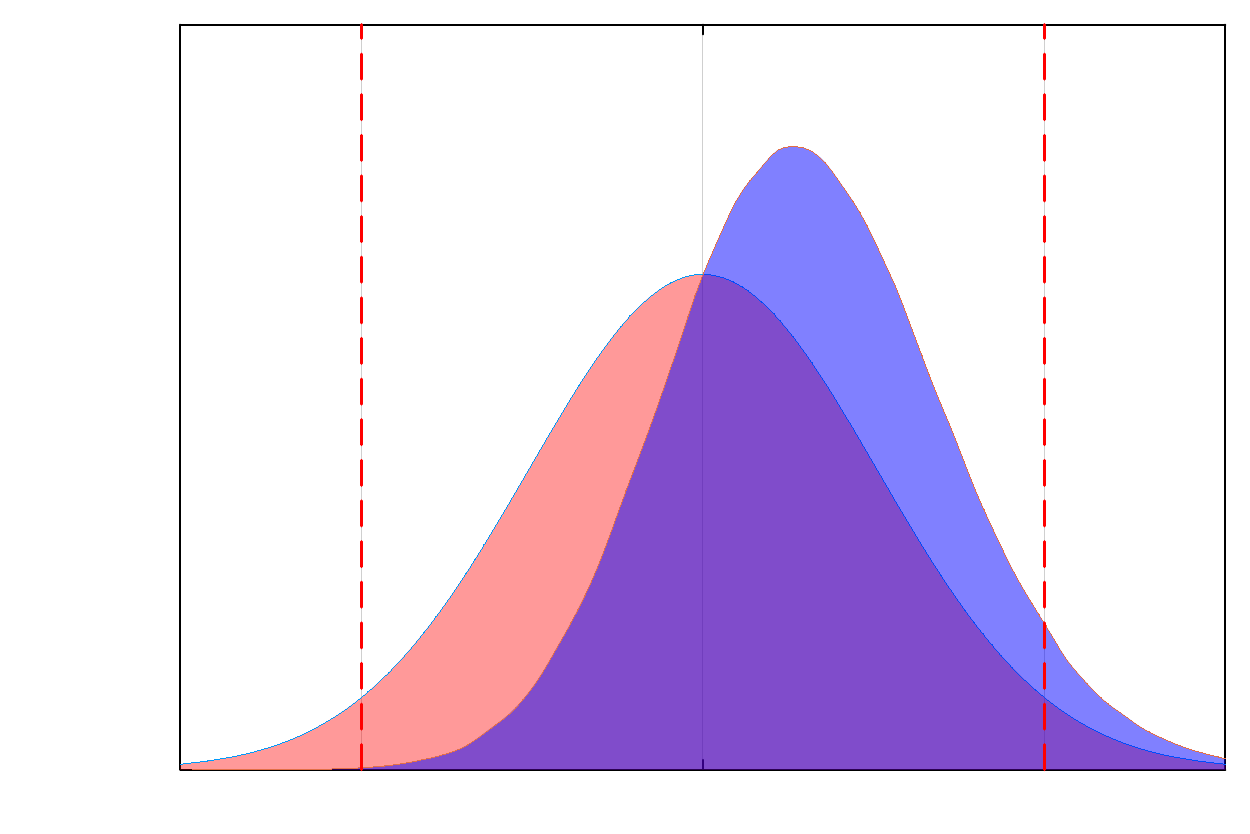}}
    \end{minipage}
    \caption{Unconditional Density of $(\hbpre,\hbpost)$ vs. Conditional $(\hbpre,\hbpost) | \hbpre \geq 0$.}
    \caption*{\small 
    $X = (\hbpre, \hbpost) \sim N(0, \Sigmaall)$, $\spre^{2} = \spost^{2} =1$, $\rho = 0.5$. \textit{Left:} Unconditional joint density of $(\hbpre, \hbpost)$.
    \textit{Center:} Joint density of $(\hbpre, \hbpost)$ conditional on $\hbpre \geq 0$. 
    \textit{Right:} Marginal density of $\hbpost$ under selective reporting (blue) compared to unconditional density (red). Conventional 5\% significance cutoffs shown with red dashed lines.}
\label{d2:inference:distortion:fig}
\end{figure}
\end{example}

\paragraph{Average Coverage Across Studies.} The conditional coverage criterion requires that $CI_{\alpha}(X)$ contain $l'\bpost$ with high probability across data realizations $X \in \Xpost$, where each data realization can be viewed as a particular draw of the same experiment in a given study. Thus, conditional coverage provides a natural inference criterion for a \textit{given} study. However, one might ponder the relevance of conditional coverage when there are \textit{multiple} studies, each with their own experiments and potential PAP deviations. In Appendix \ref{d2:app:sec:average.coverage}, we establish a general sense in which conditional coverage for \textit{each} study is necessary and sufficient for controlling average coverage across \textit{all} studies. 

Formally, we consider a sampling model where deviation sets are drawn from a rich class of distributions reflecting the unanticipated nature of PAP deviations. Under this structure, we show that to ensure average coverage across studies, it is both \textit{necessary and sufficient} to ensure conditional coverage in each study. This equivalence result provides further motivation for the conditional inference approach to PAP deviations.

\subsection{Types of Deviation Sets}\label{d2:sec:deviaton.types}
We now discuss leading examples of deviation sets. Section \ref{d2:sec:deviations.significance} outlines deviations based on the significance of PAP estimates, which we discuss extensively in future sections. Section \ref{d2:sec:deviations.independent} highlights classes of deviation sets that yield no inference distortions in our framework.

\subsubsection{Deviations Based on Significance}\label{d2:sec:deviations.significance}
Researchers often consider post estimates $\hbpost$ for which reasoning about $\Xpost$ is predicated on values of PAP estimates $\hbpre$. In particular, $\hbpost$ is often of interest when linear combinations of PAP estimates $v'\hbpre$ cross significance cutoffs $\kappa \geq 0$. Here we focus on $\kappa = z_{1-\eta/2}\spre$, where $\spre = \sqrt{v'\Sigmapre v},$ which corresponds to a conventional two-sided \textit{statistical} significance test with significance level $\eta \in (0,1)$. However, $\kappa$ can also be based on economic significance.

\paragraph{Statistical Significance.} The researcher deviates when $v'\hbpre$ is significantly large: 
\begin{align*}
    \Xpost = \curly{X \in \X: |v'\hbpre| \geq z_{1-\eta/2}\spre}.
\end{align*}
For example, consider \citet{dube2025cognitive}, who study the effects of a cognitive-skills training program (Sit-D) for police officers on discretionary arrests. In their Section IV.B, they say
\begin{displayquote}
``In this section we check the robustness of our main results to alternate specifications and variable definitions. Our discretionary arrests category constitutes a small subset of arrests \ldots that we thought were most likely to fall in response to the training. To gauge if Sit-D also reduces other types of arrests over which officers might have discretion, we use a different measure of discretionary arrests \ldots which consists of arrests for ``nonindex'' crimes under the FBI's classification. We did not specify analyzing this outcome in our PAP but examine it for robustness.'' 
\end{displayquote}
One can interpret the above quote as suggesting a deviation set based on statistical significance: $\hbpre$ is the vector of coefficients on the regression with the preregistered measure, $v$ is the unit vector that picks the coefficient on SitD, and $\hbpost$ is the vector of coefficients for the regression with the non-preregistered measure, with $v'\hbpre$ significant at level $\eta = 0.05$ \citep[Table IV]{dube2025cognitive}.

\paragraph{Statistical Insignificance.}
The researcher deviates when $v'\hbpre$ is significantly small: 
\begin{align*}
    \Xpost = \curly{X \in \X: |v'\hbpre| \leq z_{1-\eta/2}\spre}.
\end{align*}
For example, in the quote from \citet{bessone2021economic} presented in Example \ref{d2:ex:sleep}, two of the reasons the authors give for pooling, (A) and (B), are based on statistical insignificance. We consider various representations of $\Xpost$ for \citet{bessone2021economic} in the empirical application in Section \ref{d2:sec:application}.

\subsubsection{Non-Distortionary Deviations}\label{d2:sec:deviations.independent}
An important and immediate implication of our framework is that when the deviation event $\{X \in \Xpost\}$ is \textit{independent} of the post estimates $\hbpost$, one does \textit{not} have to correct conventional inferences. In particular, since the conventional interval $[l'\hbpost \pm z_{1-\alpha/2}\sigma_{post}]$ is a function of $\hbpost$, independence yields
\begin{align*}
    \P[\ball]{l'\bpost \in [l'\hbpost \pm z_{1-\alpha/2}\sigma_{post}]|X \in \Xpost} = \P[\ball]{l'\bpost \in [l'\hbpost \pm z_{1-\alpha/2}\sigma_{post}]} = 1-\alpha, \quad \forall \ball.
\end{align*}
Such \textit{non-distortionary} deviations can broadly be grouped into two cases, depicted in Figure~\ref{d2:inference:nodistortion:fig} and discussed below.

\begin{figure}
    \begin{minipage}[t]{1.3\textwidth}
    \includegraphics[width=0.30\textwidth]{img/toy_contour_unconstrained.pdf}
    \includegraphics[width=0.275\textwidth]{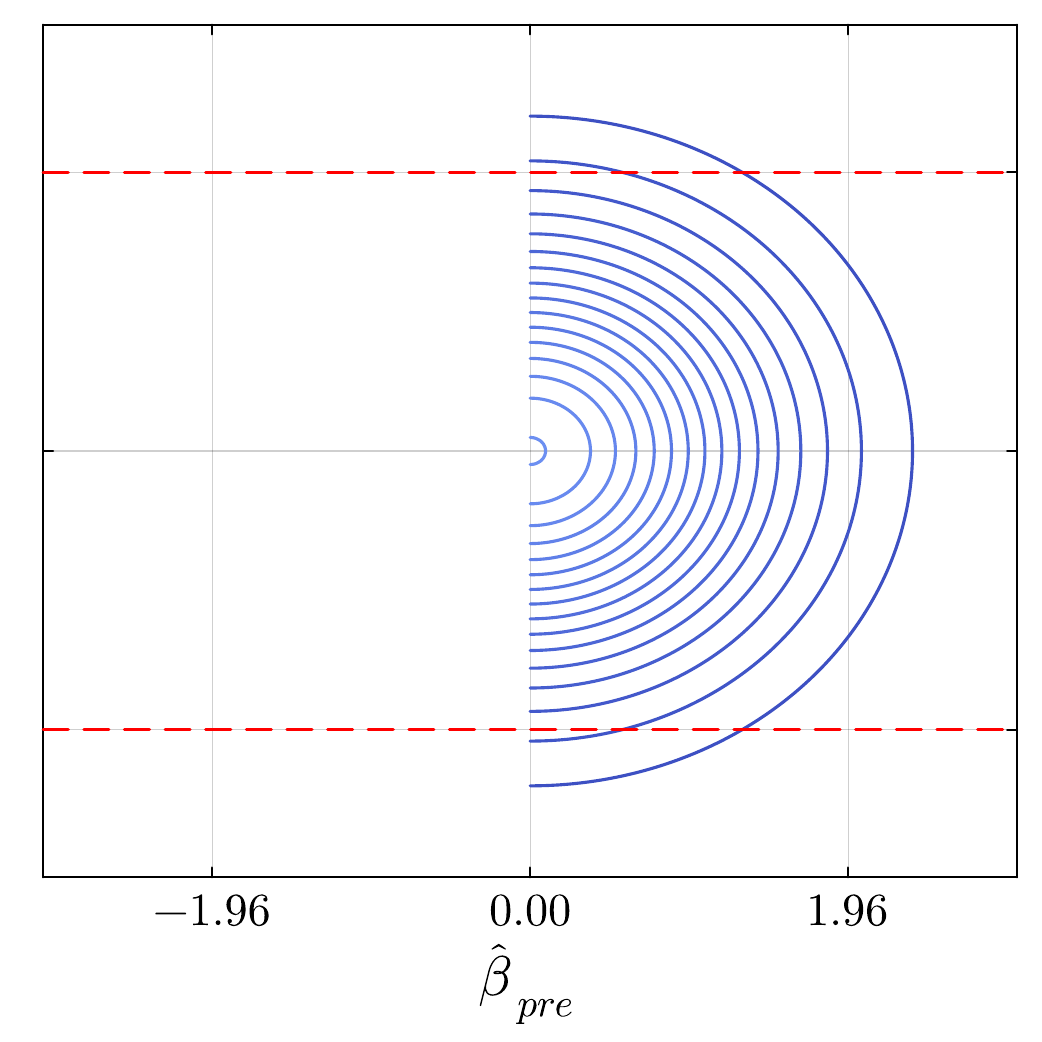}
    \reflectbox{\includegraphics[angle=90, height=0.275\textwidth]{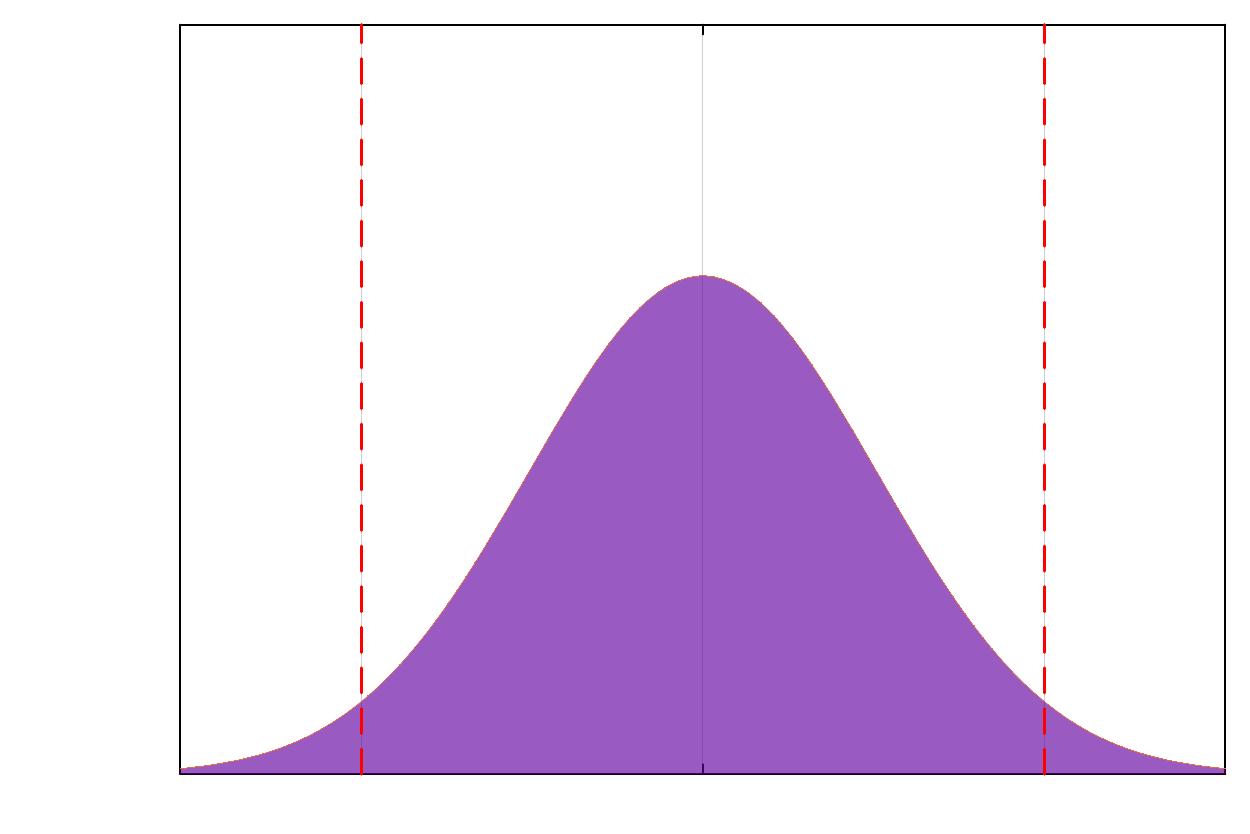}}
    \end{minipage}
    \caption{Non-Distortionary Deviations}
    \caption*{\small 
    $X = (\hbpre, \hbpost) \sim N(0, \Sigmaall)$, $\spre^{2} = \spost^{2} =1$, $\rho \in \{0.5, 0\}$. \textit{Left:} $\hbpre$ and $\hbpost$ are correlated ($\rho = 0.5$), but $\hbpost$ is reported unconditionally ($\Xpost=\R^{2}$).
    \textit{Center:} $\hbpost$ is reported conditional on $\hbpre \geq 0$ ($\Xpost=\{X \in \R^{2}: \hbpre \geq 0\}$), but $\hbpost$ and $\hbpre$ are uncorrelated ($\rho = 0$). 
    \textit{Right:} Overlapped marginal density of $\hbpost$ in either case. Conventional 5\% significance cutoffs shown with red dashed lines.}
\label{d2:inference:nodistortion:fig}
\end{figure}

\paragraph{Unconditional Reporting.} The first case (left panel of Figure~\ref{d2:inference:nodistortion:fig}) is \textit{unconditional reporting}, where even though $\hbpost$ was not registered in the PAP, it would have been reported \textit{no matter the data observed}. In other words, $\Xpost=\X$. While seemingly a trivial result, there are nonetheless many examples of such deviations that make researchers uneasy in the status quo. Instances of $\Xpost=\X$ might include: (i) honest mistakes, oversights in PAP specification, or adjustments for feasibility; (ii) unanticipated events (e.g., change in policy context, global pandemic); or (iii) newly acquired knowledge about the setting/context, econometric methods, or theoretical insights in the literature.\footnote{Examples of (i) include \citet[Appendix D.5]{rafkin2021guidance}; \citet[pages 9-10]{bhat2022long}; \citet[footnote 24]{alsan2024representation}; \citet[footnotes 18 and 24]{jacobson2024price}; \citet[Table A5]{agte2024investing}; \citet[footnote 9]{finkelstein2019take}; \citet[footnote 9]{kaur2024_financialconcerns}; \citet[page 3140]{kelley2024monitoring}. Examples of (ii) include \citet[Section III.B]{kremer2009incentives}; \citet[footnote 7]{evsyukova2025linkedout}. Examples of (iii) include \citet[footnote 26]{bandiera2021allocation}; \citet[page 2345]{field2021her}; \citet[pages 1914-15]{bessone2021economic}; \citet[footnote 18]{giacobinoschoolgirls2024}.} For concreteness, we now consider an explicit example of (iii).
\begin{example}
\citet{bandiera2021allocation} study how the allocation of authority between procurement offices (POs) and their supervising accountant generals (AGs) affects performance, as measured by the price spent on procured goods. The primary result---based on PAP estimates---is that greater PO autonomy reduces prices. The secondary result---based on post estimates---is that the benefits of autonomy are higher when AG incentives are misaligned, as measured by the extent to which the AG delays approvals for PO purchases. The exact language surrounding the deviation is as follows. 
\begin{displayquote}
``In our preanalysis plan we did not prespecify that we would study heterogeneity by the AG’s type. As the experiment rolled out and we discussed its effects with our study participants we came to realize the importance of the type of the AG in determining how the treatments, particularly the autonomy treatment, affected the way POs were able to change the way they carried out procurement.'' 
\end{displayquote}
This above language suggests a deviation set equal to the entire sample space: $\Xpost = \X$.
\end{example}

\paragraph{Independent Estimates.} The second case (center panel of Figure~\ref{d2:inference:nodistortion:fig}) is conditional reporting of $\hbpost$ based on estimates $\hbpre$ that are \textit{independent} of $\hbpost$. There are mechanical ways the researcher can ensure this: A key example is pilot studies where researchers (i) register a PAP and conduct a pilot of their study on one population and (ii) make adjustments based on results from the pilot before conducting the main study on a larger population.\footnote{An example that broadly falls into this class is \citet[Appendix D]{dean2024noise}.} When the population of the pilot is sampled independently from that of the main study, any sort of deviation from the PAP based on findings in the pilot yields no inference distortions. Another example in this vein is sample splitting, where one computes estimates on one half of the data, uses those results to determine their deviations, and computes the corresponding $\hbpost$ using the other half of the data---however, such procedures come at the cost of statistical precision.

Of course, independence need not be a mechanical feature of the experimental design. Consider an intervention disrupted by a lightning strike which changes the underlying sample, as in \citet{kremer2009incentives}.\footnote{See also \citet[footnote 2]{olken2015promises} for a discussion.} This changes the estimand’s interpretation, as the conditioning event may index a different state/population $s$. However, conventional inference will be valid for this estimand. That is, we now have $\hball|\{S=s\} \overset{d}{=} \hball(s)|\{S=s\} \overset{d}{=} \hball(s) \sim N(\ball(s), \Sigmaall(s))$, where the second equality follows from independence of selection variable $S$ (e.g., random lightning strike) and estimates $\hball(s)$.

\subsection{Quantile Conditionally Unbiased Estimation}\label{d2:sec:inference.procedures}
We now present general conditional inference procedures. In Section \ref{d2:sec:QUE.objective}, we state the formal inference objectives for confidence intervals and point estimators. In Sections \ref{d2:sec:QUE.derivation}-\ref{d2:sec:QUE.proposed.procedures}, we detail the construction of optimal procedures that meet these objectives.

\subsubsection{Objective}\label{d2:sec:QUE.objective}
To perform valid conditional inference, it suffices to derive estimators $\quant_{\alpha}$ that are $\alpha$-quantile conditionally unbiased in the sense that their overestimation probability for $l'\bpost$ conditional on $X \in \Xpost$ is equal to $\alpha$:
\begin{equation}\label{d2:eq:QCU}
    \P[\ball]{ l'\bpost \leq \quant_{\alpha} \mid X \in \Xpost} = \alpha, \quad \forall \ball.
\end{equation}
Given such $\quant_{\alpha}$, the confidence interval $CI_{\alpha}(X) = [\quant_{\alpha/2}, \quant_{1-\alpha/2}]$ provides correct conditional coverage:
\begin{equation}\label{d2:eq:QUE.CI}
    \P[\ball]{\quant_{\alpha/2} \leq l'\bpost \leq \quant_{1-\alpha/2} \mid X \in \Xpost} = 1-\alpha, \quad \forall \ball.
\end{equation}
Alternatively, if we want a point estimator of $l'\bpost$, we can use $\quant_{1/2}$ as a median conditionally unbiased estimator:
\begin{equation}\label{d2:eq:QUE.median}
    \P[\ball]{ l'\bpost \leq \quant_{1/2} \mid X \in \Xpost} = \frac{1}{2}, \quad \forall \ball.
\end{equation}
In words, the conditional median of $\quant_{1/2}$ is equal to $l'\bpost$. 

\subsubsection{Derivation}\label{d2:sec:QUE.derivation}
To derive $\quant_{\alpha}$, we follow arguments from \citet{fithian2014optimal}. Assume there exists a set $\Bpost$ of positive measure such that
\begin{equation}\label{d2:eq:X.to.B}
    \Xpost = \curly{X \in \X: \hball \in \Bpost},
\end{equation}
where the set $\Bpost$ is implicitly allowed to depend on $\Sigmaall$. In other words, we focus on deviations that arise due to values of $\hball$. This broadly accommodates the deviations that occur in practice, such as those based solely on PAP estimates $\hbpre$. For deviation sets satisfying 
\eqref{d2:eq:X.to.B}, we have
\begin{align*}
    \hbpost \mid X \in \Xpost \overset{d}{=} \hbpost \mid \hball \in \Bpost.
\end{align*}
Thus, our goal is to obtain valid inference for $l'\bpost$ conditional on $\hball \in \Bpost$. 

\paragraph{Notation.} Let $l_{post}$ denote the vector induced by (i) matrix multiplication to select $\hbpost$ from $\hball$ and (ii) vector multiplication of $\hbpost$ by $l$, so that $l'\hbpost = l_{post}'\hball$. Let $\Sigmapost$ denote the covariance matrix for $\hbpost$, so that the variance of $l'\hbpost$ is $\sigma_{post}^{2} = l_{post}'\Sigmaall l_{post} = l'\Sigmapost l$.

The main challenge is that $l'\hbpost$ is no longer normally distributed after conditioning on $\hball \in \Bpost$, owing to correlation between $l'\hbpost$ and $\hball$: 
\begin{equation}\label{d2:lpost.distribution}
    \begin{pmatrix}
        l'\hbpost \\
        \hball
    \end{pmatrix} \sim
    N\lp 
    \begin{pmatrix}
        l'\bpost \\
        \ball
    \end{pmatrix},
    \begin{pmatrix}
        \sigma_{post}^{2} & l_{post}'\Sigmaall \\
        \Sigmaall l_{post} & \Sigmaall
    \end{pmatrix}
    \rp.
\end{equation}
To account for this correlation, consider the residual $\hat{r}$ from the regression of $\hball$ on $l'\hbpost$ under their joint unconditional distribution (\ref{d2:lpost.distribution}):
\begin{equation*}
    \hat{r} = \hball - \coef l'\hbpost, \quad \coef = \frac{\Cov_{\ball}(\hball, l'\hbpost)}{\Var_{\ball}(l'\hbpost)} = \frac{\Sigmaall l_{post}}{\sigma_{post}^{2}}.
\end{equation*}
The conditional distribution of $l'\hbpost$ given $\{\hball \in \Bpost, \hat{r} = r\}$ is a $N(l'\bpost, \sigma_{post}^{2})$ distribution truncated to the set $\trunc_{post}(r) = \{ z \in \R: r + \gamma z \in \Bpost \}$:
\begin{align}\label{d2:eq:conditioning.residual}
\begin{split}
    l'\hbpost | \{\hball \in \Bpost, \hat{r} = r\} \overset{d}{=} l'\hbpost | \{l'\hbpost \in \trunc_{post}(r), \hat{r} = r\} \overset{d}{=} l'\hbpost | \{l'\hbpost \in \trunc_{post}(r)\},
\end{split}
\end{align}
where the second equality follows from independence of $\hat{r}$ and $l'\hbpost$. To proceed, let
\begin{align}\label{d2:eq:trunc.CDF.formula}
    F_{TN}(z;\mu,\sigma^{2},\trunc) = \frac{\displaystyle \int_{-\infty}^{z} \phi\paren{\frac{t-\mu}{\sigma}}\1\curly{t \in \trunc} dt}{\displaystyle \int_{-\infty}^{\infty}\phi\paren{\frac{t-\mu}{\sigma}}\1\curly{t \in \trunc} dt}
\end{align}
denote the cumulative distribution function (CDF) of the $N(\mu, \sigma^{2})$ distribution truncated to a set $\trunc$, where $\phi(z)$ denotes the probability density function (PDF) of the $N(0, 1)$ distribution. Letting 
$\widehat{\trunc}_{post} = \trunc_{post}(\hat{r})$, equation \eqref{d2:eq:conditioning.residual} and the probability integral transform yields
\begin{equation*}
    F_{TN}(l'\hbpost; l'\bpost, \sigma_{post}^{2}, \widehat{\trunc}_{post})|\{\hball \in \Bpost, \hat{r} = r\} \sim U(0,1), \quad \forall r, \quad \forall \ball.
\end{equation*}
The above holds for all $r$, so we can integrate out the residual $\hat{r}$ to obtain 
\begin{equation}\label{d2:eq:pivotal.quantity}
    F_{TN}(l'\hbpost; l'\bpost, \sigma_{post}^{2}, \widehat{\trunc}_{post})|\hball \in \Bpost \sim U(0,1), \quad \forall \ball.
\end{equation}
Thus, conditional on $\hball \in \Bpost$, the probability of observing $F_{TN}(l'\hbpost; l'\bpost, \sigma_{post}^{2}, \widehat{\trunc}_{post}) \geq 1-\alpha$ is equal to $\alpha$. The truncated normal distribution has strict monotone likelihood ratio in $\mu$, and hence its CDF is strictly decreasing in the potential values $\mu \in \R$ of $l'\bpost$. Thus, given $\hball \in \Bpost$, there exists unique $\quant_{\alpha}^{*}$ for which
\begin{align}\label{d2:eq:defining.QUE}
    F_{TN}(l'\hbpost; \quant_{\alpha}^{*} , \sigma_{post}^{2}, \widehat{\trunc}_{post}) = 1-\alpha.
\end{align}
The estimator $\quant_{\alpha}^{*}$ is quantile conditionally unbiased in the sense of criteria (\ref{d2:eq:QCU}):
\begin{align*}
    \P[\ball]{l'\bpost \leq \quant_{\alpha}^{*} \mid \hball \in \Bpost} = \P[\ball]{F_{TN}(l'\hbpost; l'\bpost, \sigma_{post}^{2}, \widehat{\trunc}_{post}) \geq 1-\alpha \mid \hball \in \Bpost} = \alpha, \quad \forall \beta.
\end{align*}

\subsubsection{Proposed Procedures}\label{d2:sec:QUE.proposed.procedures}
Based on estimator $\quant_{\alpha}^{*}$, the proposed conditional inference procedures are as follows.

\paragraph{Confidence Intervals.} We propose $CI_{\alpha}^{*}(X) =[\quant_{\alpha/2}^{*}, \quant_{1-\alpha/2}^{*}]$ as confidence intervals for $l'\bpost$, which yields
\begin{align*}
    \P[\ball]{l'\bpost \in CI_{\alpha}^{*}(X)|\hball \in \Bpost} = 1-\alpha, \quad \forall \beta.
\end{align*}
That is, $CI_{\alpha}^{*}(X)$ has correct conditional coverage for $l'\bpost$.

\paragraph{Point Estimators.} We propose $\quant_{1/2}^{*}$ as point estimators for $l'\bpost$, which yields
\begin{align*}
    \P[\ball]{ l'\bpost \leq \quant_{1/2}^{*}|\hball \in \Bpost} = \frac{1}{2}, \quad \forall \beta.
\end{align*}
That is, the conditional median of $\quant_{1/2}^{*}$ is equal to $l'\bpost$.

\subsection{Statistical Properties}
We now highlight two statistical properties of $\quant_{\alpha}^{*}$. First, the estimator $\quant_{\alpha}^{*}$ is optimal in the class of quantile conditionally unbiased estimators, in a broad sense formalized by \citet{Pfanzagl1994}. A statement of this result in our setup and notation is as follows. 
\begin{proposition}[\citet{Pfanzagl1994}]\label{d2:prop:pfanzagl}
Given any $\quant_{\alpha}$ that satisfies criteria \eqref{d2:eq:QCU}, and any quasiconvex loss function $d \mapsto L(d, l'\bpost)$ that attains its minimum at $d = l'\bpost$, we have
\begin{equation*}
    \E_{\ball}[L(\quant_{\alpha}^{*}, l'\bpost)|\hball \in \Bpost] \leq \E_{\ball}[L(\quant_{\alpha}, l'\bpost)|\hball \in \Bpost], \quad \forall \ball.
\end{equation*}
\end{proposition}
\begin{proof}
See Appendix \ref{d2:app:proof:pfanzagl}.
\end{proof}

Thus, for a broad class of loss functions, the estimator $\quant_{\alpha}^{*}$ yields lower expected loss than any other quantile conditionally unbiased estimator $\quant_{\alpha}$. This means there is limited scope for improving upon $\quant_{\alpha}^{*}$ for conditional inference: $\quant_{1/2}^{*}$ is an optimal point estimator for $l'\bpost$ and $CI_{\alpha}^{*}(X) =[\quant_{\alpha/2}^{*}, \quant_{1-\alpha/2}^{*}]$ is an optimal (equal-tailed) confidence interval for $l'\bpost$.

The second property is that conditional inferences with $\quant_{\alpha}^{*}$ will agree with conventional \textit{unconditional} inferences when the conditioning event $\{\hball \in \Bpost\}$ occurs with high probability. That is, for deviations where the use of conventional point estimators $l'\hbpost$ and confidence intervals $[l'\hbpost \pm z_{1-\alpha/2}\sigma_{post}]$ would yield minimal distortions, the researcher does not pay a price when using the conditional analogues $(\quant_{1/2}^{*}$, $CI_{\alpha}^{*}(X))$ instead.

To gain intuition for this property, note $\{\hball \in \Bpost\} = \{l'\hbpost \in \widehat{\trunc}_{post}\}$ and consider the case where $\Bpost = \R^{\dim(\hball)}$, so that $\widehat{\trunc}_{post} = \R$. In this case, formula \eqref{d2:eq:trunc.CDF.formula} yields
\begin{align}\label{d2:eq:unconditional.QUE}
    1-\alpha = F_{TN}(l'\hbpost;\quant_{\alpha}^{*},\sigma^{2}_{post},\widehat{\trunc}_{post}) = \Phi\paren{\frac{l'\hbpost - \quant_{\alpha}^{*}}{\sigma_{post}}},
\end{align}
where $\Phi(z)$ is the CDF of the $N(0,1)$ distribution. This is solved by $\quant_{\alpha}^{*} = l'\hbpost + z_{\alpha}\sigma_{post}$. Thus, when $\P[\ball]{\hball \in \Bpost} \approx 1$ so that $\Bpost \approx \R^{\dim(\hball)}$, we expect $\quant_{\alpha}^{*} \approx l'\hbpost + z_{\alpha}\sigma_{post}$. We formalize this intuition in the following result, based on \citet[Appendix C.3]{andrews2024inference}.

\begin{proposition}[\citet{andrews2024inference}]\label{d2:prop:highprob.deviations}
Hold $\Sigmaall$ fixed and consider a sequence of sets $\Bpost^{m}$ and estimands $\ball_{m}$. If $\P[\ball_{m}]{\hball \in \Bpost^{m}} \to 1$, then $\quant_{\alpha,m}^{*}$ converges in probability to $l'\hbpost + z_{\alpha}\sigma_{post}$ conditional on $\hball \in \Bpost^{m}$. That is, for each $\e > 0$, we have
\begin{align*}
    \P[\ball_{m}]{|\quant_{\alpha,m}^{*} - (l'\hbpost + z_{\alpha}\sigma_{post})| > \e \mid \hball \in \Bpost^{m}} \to 0.
\end{align*}
Thus, conditional on $\hball \in \Bpost^{m}$, $\quant_{1/2,m}^{*}$ converges in probability to $l'\hbpost$, and the endpoints of $CI_{\alpha,m}^{*}(X) = [\quant_{\alpha/2,m}^{*}, \quant_{1-\alpha/2,m}^{*}]$ converge in probability to the endpoints of $[l'\hbpost \pm z_{1-\alpha/2}\sigma_{post}]$. 
\end{proposition}
\begin{proof}
See Appendix \ref{d2:app:proof:highprob.deviations}.
\end{proof}

\subsection{Polyhedral Conditioning}\label{d2:sec:polyhedral.conditioning}
To compute $\quant_{\alpha}^{*}$ we focus on conditioning events $\hball \in \Bpost$ such that, for a finite set of matrices $A_{post, k}$ and cutoff vectors $c_{post, k}$ yielding polyhedra $\{\hball: A_{post, k} \hball \leq c_{post, k}\}$, $k=1,\ldots, K$, we have
\begin{equation}\label{d2:eq:polyhedral.conditioning}
    \Bpost = \bigcup_{k=1}^{K} \{\hball: A_{post, k} \hball \leq c_{post, k}\}. 
\end{equation}
That is, we restrict attention to deviations that arise from the estimates $\hball$ falling into a union of polyhedra, where the cutoff vectors $c_{post, k}$ are implicitly allowed to depend on $\Sigmaall$. This setup yields computationally tractable $\quant_{\alpha}^{*}$, and broadly accommodates the types of deviations that occur in practice, such as those based on statistical significance cutoffs.

\paragraph{Statistical Significance Cutoffs.} Suppose $\hbpost$ is of interest when the linear combination $v'\hbpre$ of PAP estimates passes some significance threshold governed by $\eta$. Let $v'\hbpre = v_{pre}'\hball$, where $v_{pre}$ denotes the vector induced by (i) matrix multiplication to select $\hbpre$ from $\hball$ and (ii) vector multiplication of $\hbpre$ by $v$. A two-sided statistical significance cutoff yields
\begin{align}\label{d2:eq:Bpost.significance}
    \Bpost &= \{\hball: |v'\hbpre| \geq z_{1-\eta/2}\spre\} = \{\hball: A_{post,1}\hball \leq c_{post,1}\} \bigcup \{\hball: A_{post,2}\hball \leq c_{post,2}\},
\end{align}
where $A_{post,1} = v_{pre}'$, $A_{post,2} = -v_{pre}'$, and $c_{post,1} = c_{post,2} = -z_{1-\eta/2}\spre$. On the other hand, a two-sided statistical insignificance cutoff yields
\begin{align*}
    \Bpost &= \{\hball: |v'\hbpre| \leq z_{1-\eta/2}\spre\} = \{\hball: A_{post}\hball \leq c_{post}\}, 
\end{align*}
where $A_{post} = (-v_{pre}, v_{pre})'$ and $c_{post} = (z_{1-\eta/2}\spre, z_{1-\eta/2}\spre)'$.

Under polyhedral conditioning, we obtain a convenient representation for the truncation set $\widehat{\trunc}_{post}$, based on \citet[Lemma 5.1]{lee2016exact}. In what follows, we define the maximum over the empty set as $-\infty$ and the minimum over the empty set as $+\infty$. 
\begin{proposition}[\citet{lee2016exact}]\label{d2:prop:lee.representation}
Define
\begin{align*}
    \hZminuspostk &= \max_{j \in J_{k}^{-}} \frac{(\cpostk)_{j} - (\Apostk \hat{r})_{j}}{(\Apostk\coef)_{j}}, & J_{k}^{-} &= \curly{j: (\Apostk\coef)_{j} < 0}, \\
    \hZpluspostk &= \min_{j \in J_{k}^{+}} \frac{(\cpostk)_{j} - (\Apostk \hat{r})_{j}}{(\Apostk \coef)_{j}}, & J_{k}^{+} &= \curly{j: (\Apostk\coef)_{j} > 0}, \\ 
    \hZzeropostk &= \min_{j \in J_{k}^{0}} (\cpostk)_{j} - (\Apostk \hat{r})_{j}, & J_{k}^{0} &= \curly{j: (\Apostk\coef)_{j} = 0}.
\end{align*}
If the set $\Bpost$ takes the form \eqref{d2:eq:polyhedral.conditioning}, then
\begin{equation*}
    \widehat{\trunc}_{post} = \bigcup_{k=1}^{K} \widehat{\trunc}_{post,k}, \quad  
    \widehat{\trunc}_{post,k} =
    \begin{cases}
        [\hZminuspostk, \hZpluspostk], & \hZzeropostk \geq 0, \\
        \hfil \varnothing, & \hZzeropostk < 0.
    \end{cases}
\end{equation*}
Moreover, for each $k$, the event $\{A_{post, k} \hball \leq c_{post, k}\}$ is equivalent to $\{l'\hbpost \in \widehat{\trunc}_{post,k}\}$.
\end{proposition}
\begin{proof}
See Appendix \ref{d2:app:proof:lee.representation}.
\end{proof}

In addition to yielding computationally tractable estimators $\quant_{\alpha}^{*}$, the polyhedral structure in \eqref{d2:eq:polyhedral.conditioning} allow us to better interpret the truncation set $\widehat{\trunc}_{post}$, and hence the behavior of $\quant_{\alpha}^{*}$ along relevant dimensions of the inference problem. The following section illustrates this behavior when the deviation event is characterized by a one-sided significance test. 

\subsection{Example of Explicit Derivation}\label{d2:sec:simple.example}
Consider inference on post estimand $\bpost \in \R$ conditional on PAP estimate $\hbpre \in \R$ crossing a one-sided $\eta = 0.05$ significance cutoff:
\begin{equation*}
    \hbpost \mid \hbpre \geq 1.96, \quad 
    \Sigmaall = 
    \begin{pmatrix}
    1 & \rho  \\
    \rho  & 1
    \end{pmatrix}, \quad \rho \geq 0.
\end{equation*}
Here $\rho$ is the correlation between $\hbpre$ and $\hbpost$, and we have assumed $\spre = \spost = 1$. Let
\begin{align*}
    CI_{0.95}^{*}\left(\,\hbpost\mid\{\hbpre \geq 1.96\}\,\right) = [\quant_{0.025}^{*}, \quant_{0.975}^{*}]
\end{align*}
denote the corrected (i.e., conditional) 95\%  interval. Figure~\ref{fig:cutoff_distortions} depicts how this interval (in green) varies as the realized $\hbpre$ moves further from the reporting threshold of $1.96$. To give a sense of magnitudes, we (i) mark significance cutoffs $z_{1-\eta}$ on the horizontal axis for different values of $\eta$ and (ii) depict the conventional 95\% interval (in orange). Each panel depicts this behavior over different values of the correlation $\rho$ between $\hbpre$ and $\hbpost$.

\begin{figure}
    \centering
    \includegraphics[width=0.49\linewidth]{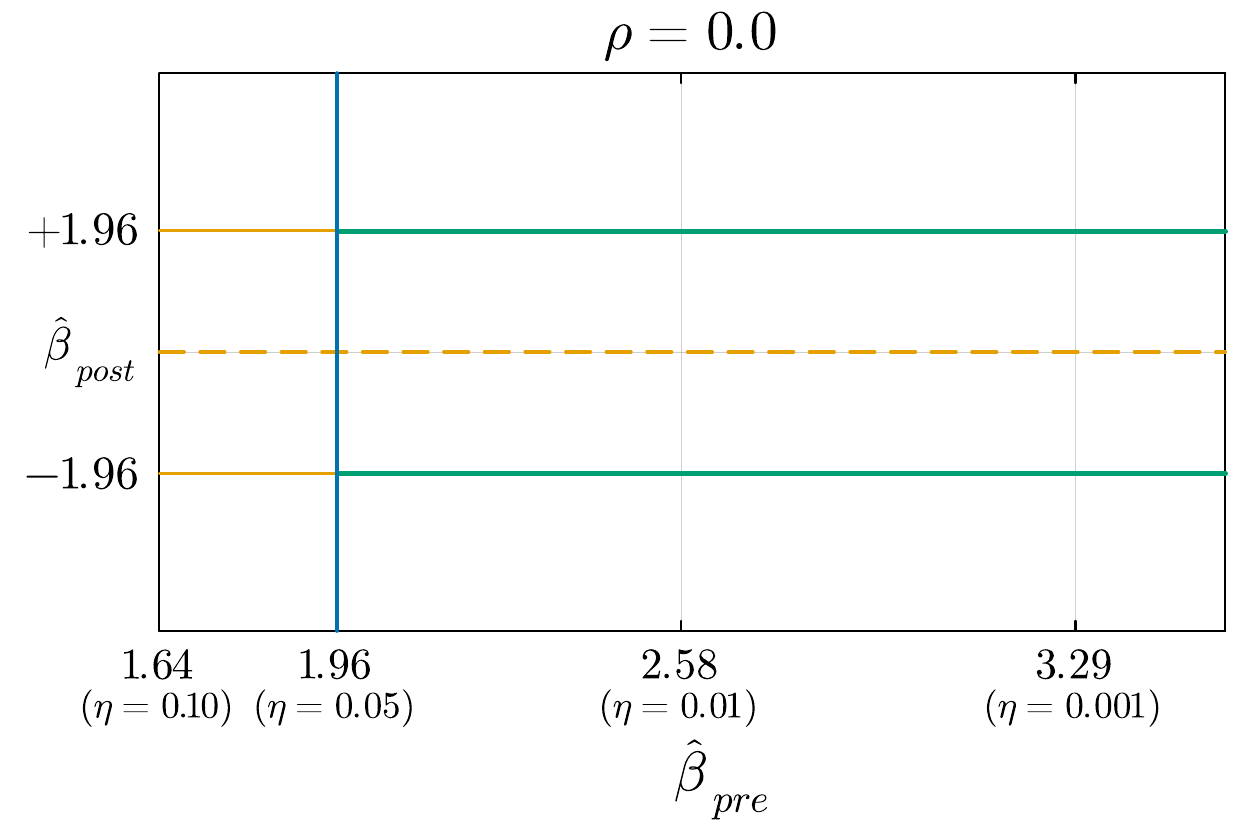}
    \includegraphics[width=0.49\linewidth]{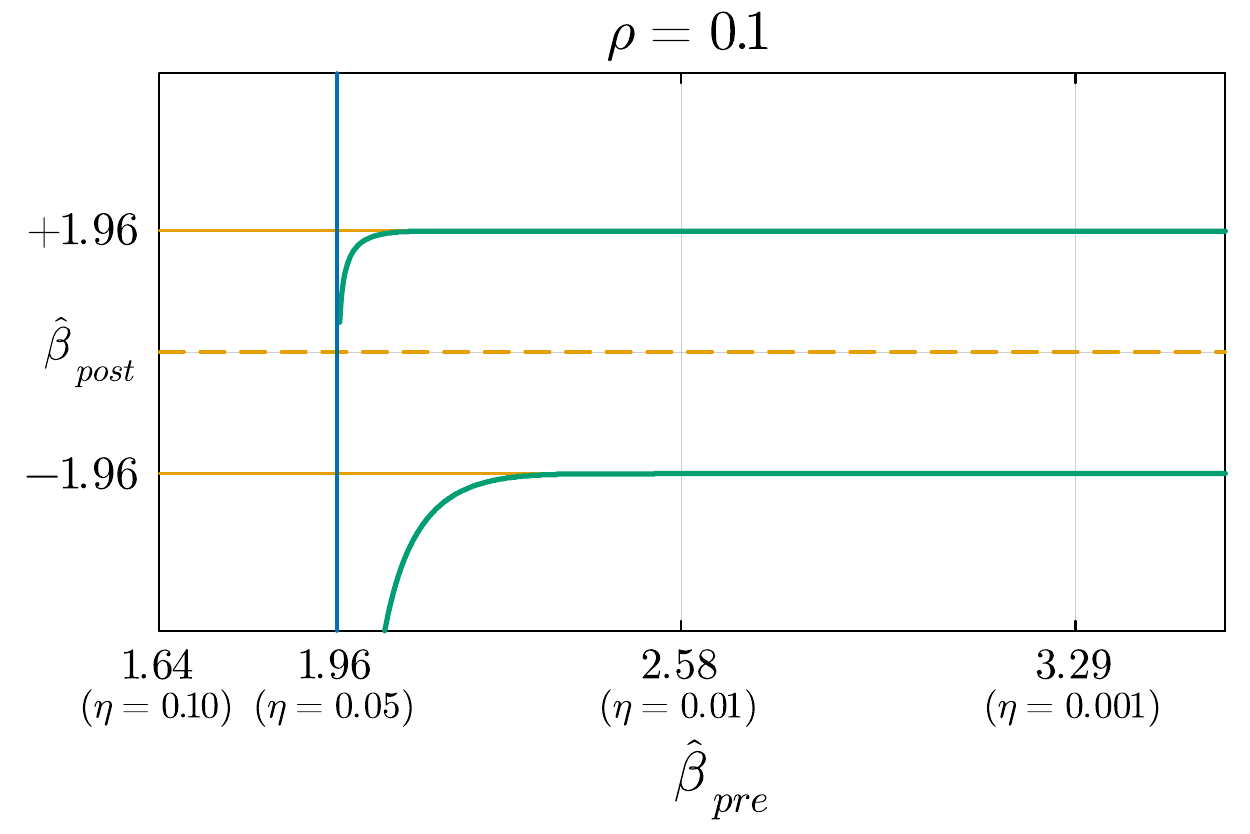}\\[1em]
    \includegraphics[width=0.49\linewidth]{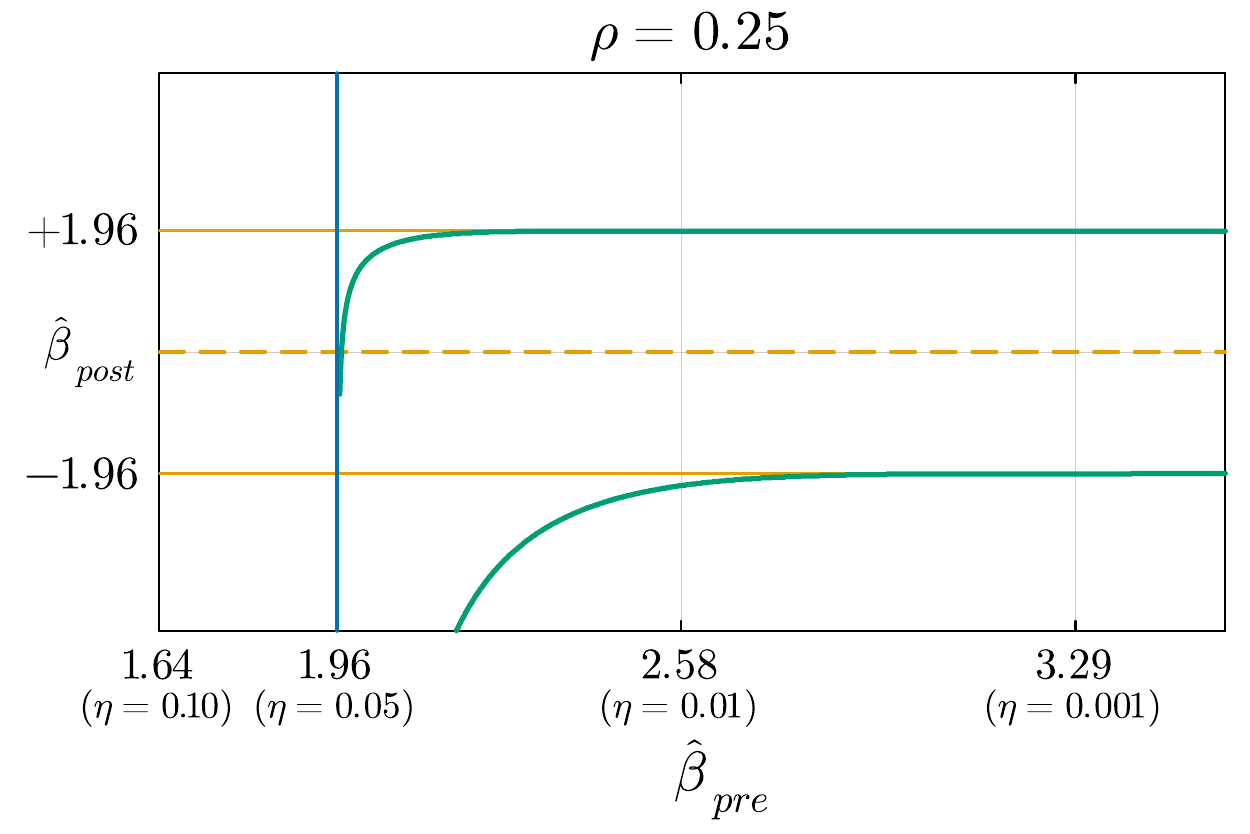}
    \includegraphics[width=0.49\linewidth]{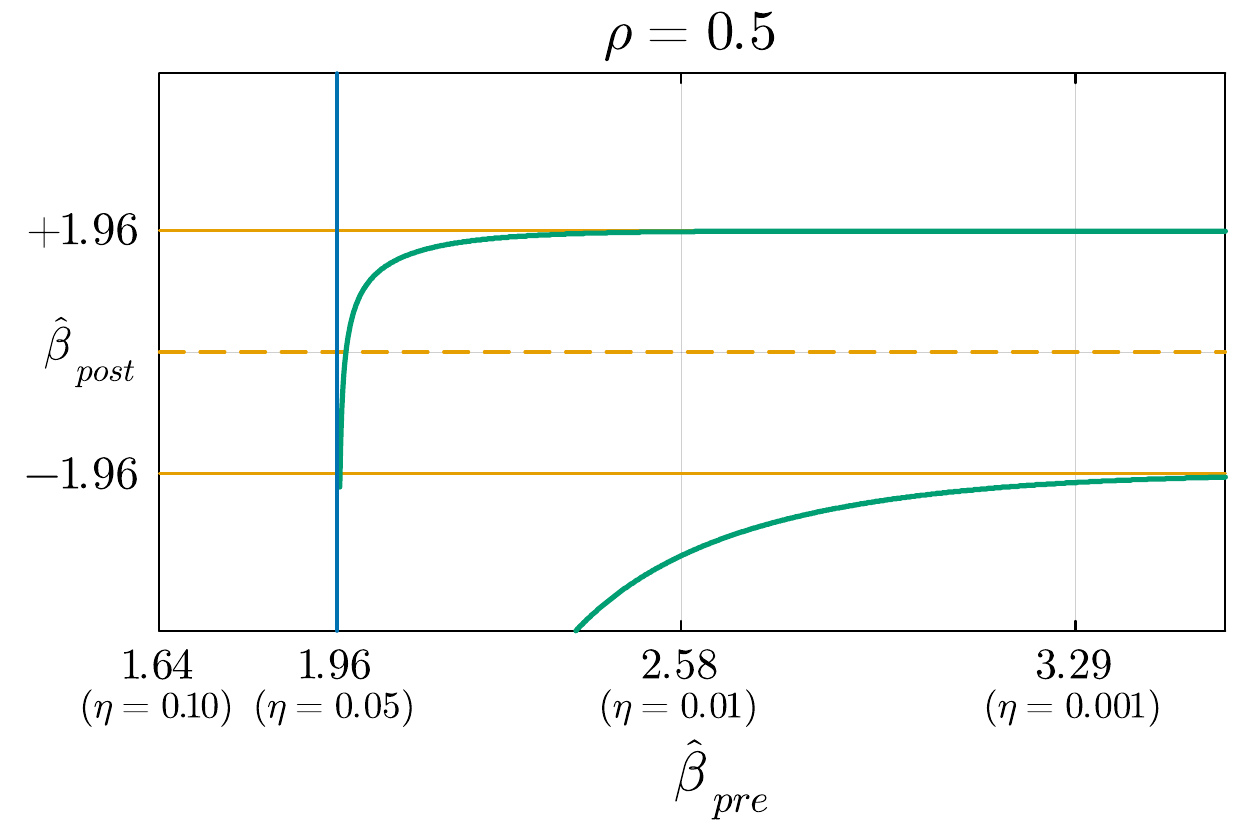} 
    \caption{Corrected $CI_{0.95}^{*}\left(\,\hbpost\mid\{\hbpre \geq 1.96\}\,\right)$ vs. Conventional $CI_{0.95}\left(\,\hbpost\right)$}
    \label{fig:cutoff_distortions}
\end{figure}

There are two major takeaways from these plots of what drives distortions. (i) \textit{Distance:} The closer $\hbpre$ is to the cutoff of the conditioning event (e.g., the distance between $\hbpre$ and 1.96), the greater the distortion from the conventional intervals. By contrast, when the realized $\hbpre$ is far above the cutoff, accounting for conditional reporting has a negligible effect. (ii) \textit{Strength of Correlation:} When the dependence (governed by $\rho$) between $\hbpost$ and the conditioning event $\{\hbpre \geq 1.96\}$ grows, the distortion becomes more pronounced. As seen in the first panel, the corrected and conventional intervals are equal when $\hbpost$ and $\hbpre$ are independent ($\rho=0$).

\begin{remark}
The above example highlights that the corrections produced by $\quant_{\alpha}^{*}$ depend on the covariance between the estimator  $\hbpost$ and the conditioning event $\{\hbpre \geq 1.96\}$. In Appendix \ref{d2:app:sec:covariance.structure}, we show that this holds more broadly for the general case of estimator $l'\hbpost$ and conditioning event $\bigcup_{k} \{A_{post, k} \hball \leq c_{post, k}\}$.
\end{remark}

\section{Feasible Inference}\label{d2:sec:feasible.inference}
We have thus far assumed that the vector $\hball$ of PAP and post estimates is normally distributed with mean $\ball$ and a known covariance matrix $\Sigmaall$. This assumption is motivated by large-sample asymptotic results that yield asymptotic normality of estimates (e.g., under the central limit theorem) and consistency of covariance matrix estimators (e.g., under the law of large numbers). These asymptotic results underlie conventional procedures used for unconditional inference, which replace $(\hball, \Sigmaall)$ with analogues $(\hballn, \hSigmaalln)$ constructed from a sample of size $n$. Following this same logic, Section \ref{d2:sec:implementation} shows how to implement plug-in versions of the conditional inference procedures proposed in Section \ref{d2:sec:inference}. In Section \ref{d2:sec:asymptotic.validity}, we establish the uniform asymptotic validity of these plug-in procedures over a broad class of probability distributions as $n \to \infty$. We prove this asymptotic validity in Appendix \ref{d2:app:sec:uniformity}.

\subsection{Implementation}\label{d2:sec:implementation}

Given a sample of size $n$ from some unknown distribution $P_{n}$, the researcher constructs PAP and post estimates, which we collect into a vector $\hballn$. The researcher also constructs a corresponding covariance matrix estimator $\hSigmaalln$. The researcher is interested in a linear combination $l'\hat{\beta}_{post,n}$ of some vector of post estimates $\hat{\beta}_{post,n}$. As before, we let $l_{post}$ denote the vector induced by (i) matrix multiplication to select $\hat{\beta}_{post,n}$ from $\hballn$ and (ii) vector multiplication of $\hat{\beta}_{post,n}$ by $l$, so that $l'\hat{\beta}_{post,n} = l_{post}'\hballn$. The corresponding variance estimator is $\hsigmapostn^{2} = l_{post}'\hSigmaalln l_{post} = l'\hat{\Sigma}_{post,n} l$. As we formalize below in Section \ref{d2:sec:asymptotic.validity}, we can think of $\hballn$ as corresponding to some underlying estimand $\ball(P_{n})$. Thus, the goal is to conduct inference on parameter $l'\ball_{post}(P_{n}) = l_{post}'\ball(P_{n})$. 

The feasible conditional inference procedures are plug-in versions of the procedures from Section \ref{d2:sec:inference}. That is, the feasible plug-in procedures replace all expressions that depend on $(\hball, \Sigmaall)$ in Section \ref{d2:sec:inference} with analogous expressions that depend on $(\hballn, \hSigmaalln)$. We explicitly describe these plug-in procedures in Section \ref{d2:sec:implementation.procedure}, and provide a concrete example in Section \ref{d2:sec:implementation.example}.

\subsubsection{Explicit Procedures}\label{d2:sec:implementation.procedure}

\noindent \fbox{\textbf{Step 1.} Derive matrices and vectors $(\Apostk, \hcpostkn)$, $k = 1, \ldots, K$.} 

\vspace{0.2cm}

Let $X_{1:n} \in \mathcal{X}_{n}$ represent the sample data. The researcher deviates to post estimates $\hat{\beta}_{post,n}$ when $X_{1:n} \in \mathcal{X}_{post,n}$, where
\begin{align*}
    \mathcal{X}_{post,n} = \{X_{1:n} \in \mathcal{X}_{n}: \hballn \in \hBpostn\}, \quad \hBpostn = \bigcup_{k=1}^{K} \{\hballn: \Apostk \hballn \leq \hcpostkn\}.
\end{align*}
The polyhedra $\{\hballn: \Apostk \hballn \leq \hcpostkn\}$, $k = 1, \ldots, K$, are based on nonrandom matrices $\Apostk$ and random cutoff vectors $\hcpostkn$ that may depend on $\hSigmaalln$. 

\vspace{0.2cm}

\noindent \fbox{\textbf{Step 2.} Compute truncation set $\htruncpostn$.} 

\vspace{0.2cm}

To account for conditioning, compute the residual
\begin{align*}
    \hat{r}_{n} = \hballn - \hcoefn l'\hat{\beta}_{post,n}, \quad \hcoefn = \frac{\hSigmaalln l_{post}}{\hsigmapostn^{2}},
\end{align*}
and compute the truncation quantities 
\begin{align*}
    \hZminuspostkn &= \max_{j \in \hat{J}_{k,n}^{-}} \frac{(\hcpostkn)_{j} - (\Apostk \hat{r}_{n})_{j}}{(\Apostk\hcoefn)_{j}}, & \hat{J}_{k,n}^{-} &= \curly{j: (\Apostk\hcoefn)_{j} < 0}, \\
    \hZpluspostkn &= \min_{j \in \hat{J}_{k,n}^{+}} \frac{(\hcpostkn)_{j} - (\Apostk \hat{r}_{n})_{j}}{(\Apostk \hcoefn)_{j}}, & \hat{J}_{k,n}^{+} &= \curly{j: (\Apostk\hcoefn)_{j} > 0}, \\ 
    \hZzeropostkn &= \min_{j \in \hat{J}_{k,n}^{0}} (\hcpostkn)_{j} - (\Apostk \hat{r}_{n})_{j}, & \hat{J}_{k,n}^{0} &= \curly{j: (\Apostk\hcoefn)_{j} = 0}.
\end{align*}
As in Proposition \ref{d2:prop:lee.representation}, the event $\{\hballn \in \hBpostn\}$ is equivalent to $\{l'\hat{\beta}_{post,n} \in \htruncpostn\}$, where
\begin{equation*}
    \htruncpostn = \bigcup_{k=1}^{K} \htruncpostkn, \quad  
    \htruncpostkn =
    \begin{cases}
        [\hZminuspostkn, \hZpluspostkn], & \hZzeropostkn \geq 0, \\
        \hfil \varnothing, & \hZzeropostkn < 0.
    \end{cases}
\end{equation*}

\vspace{0.2cm}

\noindent \fbox{\textbf{Step 3.} Solve for the plug-in quantile unbiased estimator $\hquantn$.} 

\vspace{0.2cm}

Given $\hballn \in \hBpostn$, the plug-in quantile unbiased estimator is the unique $\hquantn$ that solves
\begin{equation*}
     F_{TN}(l'\hat{\beta}_{post,n};\hquantn,\hsigmapostn^{2},\htruncpostn ) = 1-\alpha.
\end{equation*}
Computation is fast, since this amounts to finding the root of strictly monotone function. The corresponding conditional inference procedures are
\begin{enumerate}[label=(\roman*)]
    \item confidence intervals $CI_{\alpha,n}^{*}(X_{1:n}) =[\quant_{\alpha/2, n}^{*}, \quant_{1-\alpha/2, n}^{*}]$; and
    \item point estimators $\quant_{1/2,n}^{*}$.
\end{enumerate}
Section \ref{d2:sec:asymptotic.validity} establishes the asymptotic validity of these feasible plug-in procedures for inference on $l'\ball_{post}(P_{n})$ as $n \to \infty$.

\subsubsection{Concrete Example}\label{d2:sec:implementation.example}
Consider $\hballn = (\hbpren, \hbpostn)' \in \R^{2}$ and $\ball(P_{n}) = (\bpre(P_{n}), \bpost(P_{n}))' \in \R^{2}$. We want inference on post estimand $\bpost(P_{n})$ conditional on the PAP estimate $\hbpren$ crossing a two-sided statistical significance cutoff governed by $\eta \in (0,1)$:
\begin{equation*}
    \hBpostn = \curly{
    \begin{pmatrix}
    \hbpren \\
    \hbpostn
    \end{pmatrix}: |\hbpren| \geq z_{1-\eta/2}\hspren}, \quad \hSigmaalln = 
    \begin{pmatrix}
    \hspren^{2} & \hcovn \\
    \hcovn  & \hspostn^{2}
    \end{pmatrix}, \quad \hcovn > 0,
\end{equation*}
where $\hcovn$ is the covariance estimator for $\hbpren$ and $\hbpostn$. In this case, $l = l_{post} = (0,1)'$.

\vspace{0.2cm}

\noindent \fbox{\textbf{Step 1.} Derive matrices and vectors $(\Apostk, \hcpostkn)$, $k = 1, \ldots, K$.} 

\vspace{0.2cm}

In this case, the set $\hBpostn$ takes the form of \eqref{d2:eq:Bpost.significance}:
\begin{align*}
    \hBpostn = \curly{
    \begin{pmatrix}
    \hbpren \\
    \hbpostn
    \end{pmatrix}
    : \hbpren
    \leq -z_{1-\eta/2}\hspren} \bigcup \curly{
    \begin{pmatrix}
    \hbpren \\
    \hbpostn
    \end{pmatrix}: 
    -\hbpren \leq -z_{1-\eta/2}\hspren}.
\end{align*}
Thus, $K=2$ and
\begin{align*}
    A_{post,1} = 
    (1, 0), \quad 
    A_{post,2} = 
    (-1, 0), \quad
    \hat{c}_{post,1,n} = \hat{c}_{post,2,n} = -z_{1-\eta/2}\hspren.
\end{align*}

\vspace{0.2cm}

\noindent \fbox{\textbf{Step 2.} Compute truncation set $\htruncpostn$.} 

\vspace{0.2cm}

The regression step yields
\begin{align*} 
    \hcoefn = \frac{\hSigmaalln l_{post}}{\hsigmapostn^{2}}
    =
    \begin{pmatrix}
    \displaystyle\frac{\hcovn}{\hspostn^{2}} \\
    1
    \end{pmatrix}, \quad
    \hat{r}_{n} = \hballn - \hcoefn \hbpostn = 
    \begin{pmatrix}
    \hbpren - \displaystyle \frac{\hcovn}{\hspostn^{2}}\hbpostn  \\
    0
    \end{pmatrix}.
\end{align*}
For $k=1$, we obtain $\hZzero_{post,1,n} = 0$ and
\begin{align*}
    \widehat{\trunc}_{post,1,n} = [-\infty, \hZplus_{post,1,n}], \quad 
    \hZplus_{post,1,n} = \frac{-z_{1-\eta/2}\hspren - \paren{\hbpren - \displaystyle\frac{\hcovn}{\hspostn^{2}}\hbpostn}}{\displaystyle\frac{\hcovn}{\hspostn^{2}}}.
\end{align*}
For $k=2$ we obtain $\hZzero_{post,2,n} = 0$ and
\begin{align*}
    \widehat{\trunc}_{post,2,n} = [\hZminus_{post,2,n}, +\infty], \quad 
    \hZminus_{post,2,n} = \frac{z_{1-\eta/2}\hspren - \paren{\hbpren - \displaystyle \frac{\hcovn}{\hspostn^{2}}\hbpostn}}{ \displaystyle\frac{\hcovn}{\hspostn^{2}}}. 
\end{align*}
Thus, $\{|\hbpren| \geq z_{1-\eta/2}\hspren\} = \{\hbpostn \in \widehat{\trunc}_{post,n}\}$, where 
\begin{align*}
    \widehat{\trunc}_{post,n} = [-\infty, \hZplus_{post,1,n}] \cup [\hZminus_{post,2,n}, +\infty].
\end{align*}

\vspace{0.2cm}

\noindent \fbox{\textbf{Step 3.} Solve for the plug-in quantile unbiased estimator $\hquantn$.} 

\vspace{0.2cm}

Given $|\hbpren| \geq z_{1-\eta/2}\hspren$, formula \eqref{d2:eq:trunc.CDF.formula} yields
\begin{equation*}
    F_{TN}(z; \mu, \hspostn^{2}, \widehat{\trunc}_{post,n}) =
    \begin{cases}
        \hfil \dfrac{\displaystyle \Phi \lp \frac{z - \mu}{\hspostn} \rp}{\displaystyle \Phi \paren{ \frac{\hZplus_{post,1,n} - \mu}{\hspostn}} + 1 - \Phi \lp \frac{\hZminus_{post,2,n} - \mu}{\hspostn} \rp }, & z < \hZplus_{post,1,n}, \\[10pt]
        \hfil \dfrac{\displaystyle \Phi \lp \frac{z - \mu}{\hspostn} \rp - \Phi \lp \frac{\hZminus_{post,2,n} - \mu}{\hspostn} \rp + \Phi \paren{ \frac{\hZplus_{post,1,n} - \mu}{\hspostn}}}{\displaystyle \Phi \paren{ \frac{\hZplus_{post,1,n} - \mu}{\hspostn}} + 1 - \Phi \lp \frac{\hZminus_{post,2,n} - \mu}{\hspostn} \rp}, & z \geq \hZminus_{post,2,n}.
    \end{cases}
\end{equation*}
The estimator $\quant_{\alpha,n}^{*}$ is the unique solution to $F_{TN}(\hbpostn; \quant_{\alpha,n}^{*}, \hspostn^{2}, \widehat{\trunc}_{post,n}) = 1-\alpha$.

\subsection{Uniform Asymptotic Validity}\label{d2:sec:asymptotic.validity}
We now show that inference based on $\hquantn$ is uniformly asymptotically valid, formally stated in equation \eqref{d2:eq:QUE.asymptotic.validity} below. Existing results by \cite{andrews2024inference} and \cite{mccloskey2024hybrid} can be used to prove uniformity the case of $K=1$. Here, we prove the general case for $K \geq 1$. In the proof we assume the rows of $\Apostk$ are nonzero, i.e., $(\Apostk)_{j} \neq 0$ for all $(j,k)$, which rules out, for example, a deviation event of the form $\{0 \leq \hat{\sigma}_{pre,n} - \hat{\sigma}_{post,n}\}$.

\paragraph{Environment.} We suppose that a sample of size $n$ is drawn from some unknown distribution $P_{n} \in \cPn$, where $\cPn$ is a class of probability distributions corresponding to a sample of size $n$. For example, given a class of distributions $\mathcal{P}_{0}$ with bounded moments, if we have i.i.d. draws of size $n$ from some fixed $P_{0} \in \mathcal{P}_{0}$, then $P_{n} = (P_{0})^{n}$ is the product distribution and $\mathcal{P}_{n}$ is the set of all products of a fixed distribution with bounded moments. As we take $n \to \infty$, there is a corresponding sequence of unknown distributions $\{P_{n}\} \in \times_{n=1}^{\infty}\cPn$, where $\times_{n=1}^{\infty}\cPn$ is the sequence of distribution classes, and the notation $\{P_{n}\} \in \times_{n=1}^{\infty}\cPn$ means that $P_{n} \in \cPn$ for all $n$. Below we let $P \in \cup_{n=1}^{\infty}\cPn$ index elements belonging to $\mathcal{P}_{n}$ for some $n$.
 
\paragraph{Asymptotic Normality.} We first assume that the scaled estimates $\tballn = \sqrt{n}\hballn$ are uniformly asymptotically normal. Let $BL_{1}$ denote the set of real-valued functions that are bounded above in absolute value by one and have Lipschitz constant bounded above by one \citep[Section 1.12]{van1996weak}. Furthermore, given a candidate covariance matrix $\Sigmaall$, let $\lambda_{min}(\Sigmaall)$ and $\lambda_{max}(\Sigmaall)$ denote its minimum and maximum eigenvalues.

\begin{assumption}\label{d2:ass:BL}
There exist $(\ballP, \SigmaallP)$ such that for $\tballPn = \sqrt{n}\ballP$ and $\xi_{P} \sim N(0, \SigmaallP)$,
\begin{align*}
    \lim_{n \to \infty} \sup_{P \in \cPn} \sup_{f \in BL_{1}} \abs{\E_{P}[
    f(\tballn - \tballPn)] - \E[f(\xi_{P})]}.
\end{align*}
Moreover, there exists finite $\blam > 0$ such that $1/\blam \leq \lambda_{min}(\SigmaallP) \leq \lambda_{max}(\SigmaallP) \leq \blam$ for all $P\in\cPn$.
\end{assumption}

Assumption \ref{d2:ass:BL} requires that $\sqrt{n}(\hballn - \ballP)$ converges to $N(0, \SigmaallP)$ in bounded Lipschitz metric, uniformly in $P$. This is a standard way to define uniform convergence in distribution.\footnote{Examples include \citet[Assumption 6]{andrews2024inference} and \citet[Assumption 2]{mccloskey2024hybrid}.} For example, if the components of $\hballn$ are sample averages of unit-level observations, then uniform convergence follows from bounds on the moments of the observations and bounds on dependence across observations. Assumption \ref{d2:ass:BL} also requires the eigenvalues of $\SigmaallP$ to be uniformly bounded above and away from zero, which ensures that $\sqrt{n}(\hballn - \ballP)$ is stochastically bounded with nonzero asymptotic variance. 

\paragraph{Consistent Covariance Matrix Estimation.} We next assume that the scaled covariance matrix estimator $\tSigmaalln = n\hSigmaalln$ is consistent for $\SigmaallP$, uniformly over $P$.
\begin{assumption}\label{d2:ass:consistent.variance}
For each $\e > 0$,
\begin{align*}
    \lim_{n \to \infty} \sup_{P \in \cPn} \P[P]{\norm{\tSigmaalln - \SigmaallP} > \e} = 0.
\end{align*}
\end{assumption}
If $\hSigmaalln$ is appropriate for the setting at hand (e.g., sample covariance for iid data, long-run covariance for time series data), then Assumption \ref{d2:ass:consistent.variance} follows from the same kind of sufficient conditions that justify Assumption \ref{d2:ass:BL}.

\paragraph{Consistent Cutoff Vector Estimation.} Finally, we assume that the scaled cutoff vectors $\tcpostkn = \sqrt{n}\hcpostkn$ are uniformly consistent, in similar fashion to $\tSigmaalln$.
\begin{assumption}\label{d2:ass:consistent.cutoff}
For each $k=1,\ldots, K$, there exist $\cpostkP$ such that for each $\e > 0$,
\begin{align*}
    \lim_{n \to \infty} \sup_{P \in \cPn} \P[P]{\norm{\tcpostkn - \cpostkP} > \e} = 0.
\end{align*}
Moreover, there exists finite $\blam_{c} > 0$ such that $\norm{\cpostkP} \leq \blam_{c}$ for all $P\in\cPn$ and $k$.\footnote{This assumption is analogous to \citet[Assumption 1]{mccloskey2024hybrid}.}
\end{assumption}

For example, often the cutoff vectors take the form
$\hcpostkn = q_{post,k}\sqrt{C_{post,k}'\hSigmaalln C_{post,k}}$, where $C_{post,k} \neq 0$ and $q_{post,k} \neq 0$ are nonrandom vectors. Assumption \ref{d2:ass:consistent.variance} then implies 
\begin{align*}
    &\sup_{P \in \cPn} \P[P]{\abs{\sqrt{C_{post,k}'\tSigmaalln C_{post,k}} - \sqrt{C_{post,k}'\SigmaallP C_{post,k}}} > \e} \\
    &\leq \sup_{P \in \cPn} \P[P]{\sqrt{\abs{C_{post,k}'\tSigmaalln C_{post,k} - C_{post,k}'\SigmaallP C_{post,k}}} > \e} \\
    &\leq \sup_{P \in \cPn} \P[P]{\abs{C_{post,k}'(\tSigmaalln - \SigmaallP)C_{post,k}} > \e^{2}} \\
    &\leq \sup_{P \in \cPn} \P[P]{\bignorm{\tSigmaalln - \SigmaallP} > \frac{\e^{2}}{\norm{C_{post,k}}^{2}}} \to 0,
\end{align*}
where the first inequality follows from $|\sqrt{x} - \sqrt{y}| \leq \sqrt{|x - y|}$ for $x, y \geq 0$, and the last equality follows from the matrix operator norm bound. Thus, under Assumption \ref{d2:ass:consistent.variance}, the above $\tcpostkn = q_{post,k}\sqrt{C_{post,k}'\tSigmaalln C_{post,k}}$ satisfies Assumption \ref{d2:ass:consistent.cutoff} with $\cpostkP = q_{post,k}\sqrt{C_{post,k}'\SigmaallP C_{post,k}}$.

\paragraph{Uniform Asymptotic Validity.} Let $\ball_{post}(P)$ denote the components of the centering vector $\ball(P)$ corresponding to $\hat{\beta}_{post,n}$. Under Assumptions \ref{d2:ass:BL}-\ref{d2:ass:consistent.cutoff}, we show in Appendix \ref{d2:app:sec:uniformity} that inference based on $\hquantn$ is uniformly asymptotically valid in the sense that
\begin{align}\label{d2:eq:QUE.asymptotic.validity}
    \lim_{n \to \infty} \sup_{P \in \cPn} \bigabs{\P[P]{l'\bpost(P) \leq \hquantn \mid \hballn \in \hBpostn} - \alpha}\P[P]{\hballn \in \hBpostn} = 0.
\end{align}
That is, $\hquantn$ satisfies an asymptotic analogue of the quantile conditional unbiasedness criteria \eqref{d2:eq:QCU} for the normal model $\hball \sim N(\ball, \Sigmaall)$.\footnote{The multiplication by $\P[P]{\hballn \in \hBpostn}$ accounts for sequences where the probability of the conditioning event converges to zero. This a standard way to account for such cases \citep{andrews2024inference, mccloskey2024hybrid}.} The required uniformity in the class of distributions $\cPn$ is the asymptotic analogue of requiring that procedures in the normal model be valid regardless of the unknown mean $\ball$.\footnote{The impossibility results of \citet{leeb2006can} do not apply here, since the above approach does not attempt to consistently estimate the conditional distribution of $l'\hbpostn$ given deviation. Rather, it accounts for conditioning by using the plug-in analogue of pivotal quantity \eqref{d2:eq:pivotal.quantity} for $l'\bpost$ from the normal model.}

\paragraph{Confidence Intervals.} The above convergence implies uniformly valid conditional coverage for our proposed confidence intervals:
\begin{align*}
    \lim_{n \to \infty} \sup_{P \in \cPn} \bigabs{\P[P]{l'\bpost(P) \in CI_{\alpha,n}^{*}(X_{1:n}) \mid \hballn \in \hBpostn} - (1-\alpha)}\P[P]{\hballn \in \hBpostn} = 0,
\end{align*}
which is the asymptotic analogue of criteria \eqref{d2:eq:QUE.CI} from the normal model. 

\paragraph{Point Estimators.} The above convergence also implies uniformly valid conditional median-unbiasedness for our proposed point estimators:
\begin{align*}
    \lim_{n \to \infty} \sup_{P \in \cPn} \bigabs{\P[P]{l'\bpost(P) \leq \hat{\mu}_{1/2,n}^{*} \mid \hballn \in \hBpostn} - \frac{1}{2}}\P[P]{\hballn \in \hBpostn} = 0,
\end{align*}
which is the asymptotic analogue of criteria \eqref{d2:eq:QUE.median} from the normal model. 

\section{Application}\label{d2:sec:application}
We now return to Example \ref{d2:ex:sleep} and apply our approach to the empirical results of \cite{bessone2021economic}, using the authors' stated reasons for deviating to consider different possible deviation sets. Recall that the authors cite three reasons for reporting the non-preregistered pooled coefficient $\htau_{N+NE+NI}$:
\begin{enumerate}\itemsep=0em
    \item[(A)] \textbf{Statistical Power:} ``Those who received a night sleep treatment in addition to naps had very similar effects to those with naps only,'' i.e., 

    $$\Xpost^A = \lc X\in\X \,:\, \frac{\abs{\htau_{N}-\htau_{NE}}}{\text{se}(\htau_{N}-\htau_{NE})} \leq z_{1-\eta/2} \,\text{ and }\,  \frac{\abs{\htau_{N}-\htau_{NI}}}{\text{se}(\htau_{N}-\htau_{NI})} \leq z_{1-\eta/2}\rc .$$
    
    \item[(B)] \textbf{Interpretability:} ``Each of the night sleep treatments alone had no significant effect on this overall index,'' i.e.,
    $$\Xpost^B = \lc X\in\X\,:\, \frac{\abs{\htau_{E}}}{\text{se}(\htau_{E})} \leq z_{1-\eta/2} \,\text{ and }\, \frac{\abs{\htau_{I}}}{\text{se}(\htau_{I})} \leq z_{1-\eta/2}\rc.$$
    
     \item[(C)] \textbf{Economic Significance:} ``In contrast, naps alone had a positive, marginally significant effect,'' i.e., 
    $$\Xpost^C = \lc X\in\X \,:\, \frac{\abs{\htau_{N}}}{\text{se}(\htau_{N})} \geq z_{1-\eta/2}\rc .$$
\end{enumerate}

With these three plausible deviation sets in mind, we can consider the empirical consequences of correcting for the conditional reporting of $\htau_{N+NE+NI}$ by conditioning on sequential intersections of these events. Table \ref{tab:corrected_CIs} shows the corrected 95\% confidence intervals for the three discussed deviation sets alongside the original CI reported in the paper. The same significance cutoff of $\eta=0.05$ is used for all constraints in (A), (B), and (C). We show in subsequent tables how the choice of significance cutoff $z_{1-\eta/2}$ used for (A), (B), and (C) impacts these CIs. 

\begin{table}[ht]
\centering
\caption{Corrected Confidence Intervals Under Selective Reporting}
\label{tab:corrected_CIs}
\begin{tabular}{lll}
\toprule
\textbf{Conditioning Event(s) in $\Xpost$} & $CI^*_{0.95}$\\
\midrule
No conditioning (original) & [0.0313, 0.1820] \\[0.4em]
(A) & [0.0313, 0.1820] \\[0.4em]
(A) and (B) & [-0.0001, 0.1837] \\[0.4em]
(A) and (B) and (C)  & [-0.0001, 0.1837] \\[0.4em]
\bottomrule
\end{tabular}
\end{table}

\paragraph{Conditioning on (A) vs. No Conditioning.}
We first observe that conditioning on (A) alone yields a $CI^*_{0.95}$ almost exactly equal to the conventional $CI_{0.95}$. 
Observe that (A) is made up of four inequalities: 
\begin{align*}
    -z_{1-\eta/2} &\leq \frac{\htau_{N}-\htau_{NE}}{\text{se}(\htau_{N}-\htau_{NE})} \leq z_{1-\eta/2}, \\
    -z_{1-\eta/2} &\leq \frac{\htau_{N}-\htau_{NI}}{\text{se}(\htau_{N}-\htau_{NI})} \leq z_{1-\eta/2}.
\end{align*}
Recall from the toy example of a single constraint from Figure~\ref{fig:cutoff_distortions} that the corrected confidence intervals $CI^*_{0.95}$ converge to conventional confidence intervals at a sufficient distance from the cutoff for reporting. To contextualize this distance, consider the constraint closest to binding in the data,
$$\htau_{N} - \htau_{NI} \leq 1.96\cdot \text{se}(\htau_{N}-\htau_{NI}).$$ 
Loosely speaking, the unchanged confidence intervals in this case reflect that the realized estimates in \cite{bessone2021economic} are sufficiently ``far'' from the boundary of $\Xpost^A$.\footnote{The corrected intervals $CI^*_{0.95}$ also do not change for different choices of significance cutoffs, $z_{1-\eta}$, for $\alpha \in\{0.01,0.05,0.1\}$ (not shown given redundancy).} To provide a sense of the \textit{scale} of this distance, we can consider an exercise of holding estimates $\htau_{N}, \htau_{NI}$ fixed and varying $\text{se}(\htau_{N}-\htau_{NI})$ through $\hat{\sigma}^2_{NI}$ alone. To see any appreciable difference between $CI_{0.95}^*$ and the conventional $CI_{0.95}$, one would need $\hat{\sigma}^2_{NI}$ to be around 5 times smaller than its observed value. 

\paragraph{Conditioning on (A) and (B).}
Adding the condition that both $\htau_E$ and $\htau_I$ are near-zero imposes a constraint which is much closer to binding, and has an appreciable impact on the result of conditioning. For further intuition, Table~\ref{d2:tab:corrected_varied_CIs_B} shows how the confidence intervals change when one varies the the threshold for significance in (B), i.e., $z_{1-\eta/2}$.

\begin{table}[ht]
\centering
\caption{Conditional Confidence Intervals for $\Xpost^{AB}$}
\label{d2:tab:corrected_varied_CIs_B}
\begin{tabular}{lll}
\toprule
\textbf{Threshold for significance $z_{1-\eta/2}$ in (B)} & $CI^*_{0.95}$\\
\midrule
No conditioning (original) & [0.0313, 0.1820] \\[0.4em]
$\eta=0.01$ & [0.02492, 0.18216] \\[0.4em]
$\eta=0.05$ & [-0.00012, 0.18365] \\[0.4em]
$\eta=0.1$  & [-0.04088, 0.18578] \\[0.4em]
\bottomrule
\end{tabular}
\end{table}
In the data, this change is driven by the insignificance constraint on $\htau_E$, which grows closer to binding as $\eta=0.1$ increases, i.e., the standard for significance becomes more lax.

\paragraph{Conditioning on (A), (B), and (C).}
As with the addition of constraint (A) relative to no conditioning, the addition of constraint (C) to $\{$(A) and (B)$\}$ does not have a meaningful impact. This is because $\htau_N$ is sufficiently far from the boundary of significance for conditioning to have a meaningful impact---the same holds true for conditioning either on $\{$(A) and (C)$\}$ or on (C) alone. Overall, the above empirical results demonstrate that, depending on the reason for deviating, the adjustments from our procedures can range from having no difference to an economically significant difference relative to conventional practice.

\section{Robustness to Misspecification of Deviation Set}\label{d2:sec:robustness}
We have so far assumed that for any $\hbpost = \Spost(X)$, the reported deviation set $\Xpost = \{X \in \X: \Spost \in \cSpost\}$ is correct. This assumption is justified when the researcher is honest and capable of articulating $\Xpost$. In some settings, however, either of these assumptions may seem implausible. For example, an earnest researcher may face cognitive or communication costs when reporting $\Xpost$ (much resembling the costs faced during initial PAP specification). Alternately, a nefarious researcher may know their true $\Xpost$, but strategically report a different deviation set (e.g., to obtain shorter confidence intervals).

In this section, we consider the possibility that the researcher reports some deviation set different from the truth, $\tXpost\neq\Xpost$, leading to potentially biased point estimators $\quant_{1/2}^{*}(X,\tXpost)$ and confidence intervals $CI_{\alpha}^{*}(X,\tXpost)$ with incorrect coverage. We refer to $\tXpost$ as the \textit{reported deviation set}. In cases where the reported deviation set varies with the realized data $x\in\X$, $\tXpost(x) \subseteq \X$ is defined more generally as a correspondence $\tXpost : \X \rightrightarrows \X$.\footnote{Note that defining results in terms of this correspondence abstracts from underlying assumptions about the researcher (e.g., incentives, utility, accuracy) generating the reporting behavior. Therefore, while we will point out possible assumptions about the researcher that \textit{generate} certain reporting behavior (i.e., properties of correspondence $\tXpost(x)$), our results are agnostic to these assumptions.}

\subsection{Impossibility Result}\label{d2:sec:impossibility}
To frame upcoming discussion, we open with an impossibility result: Proposition~\ref{d2:prop:impossibility} states that no non-trivial conditional inference procedure can guarantee valid conditional coverage when the researcher is fully unrestricted in reporting $\tXpost$. For instance, no non-trivial procedure can insure against a nefarious researcher strategically choosing a particular $\tXpost$ to exclude some $\beta_{0}$.

\begin{proposition}\label{d2:prop:impossibility}
For $X \in \R^{\dim(X)}$ continuous with full support, let $CS_{\alpha}(X, \tX)$ be any confidence set procedure for which there exists $(x_0, \tX_{0})$ with $x_{0} \in \tX_{0}$ such that $\beta_{0}\notin CS_{\alpha}(x_0, \tX_{0})$ (i.e., the procedure is ``capable of rejecting'' $\beta_{0}$) and $\beta_{0}$ is not contained in $CS_{\alpha}(x,\tX_{0})$ for all $x$ in some open neighborhood $U \subseteq \tX_{0}$ of $x_0$.\footnote{We restrict to procedures ``capable of rejecting'' $\beta_{0}$ so as to disregard trivial procedures for coverage (such as the constant set $CS_{\alpha}(\cdot, \cdot) \equiv \R^{\dim(\beta)}$).} Then there exists $(\Xpost, \tXpost)$ such that (i) the researcher deviates with positive probability for $X \in \Xpost$ and (ii) conditional coverage of $CS_{\alpha}(X, \tXpost)$ is zero.
\end{proposition}

\begin{proof}
See Appendix \ref{d2:app:proof:impossibility}.
\end{proof}

Reframed, Proposition~\ref{d2:prop:impossibility} implies that \textit{some} form of additional structure on reporting behavior must be imposed to guarantee valid coverage. Therefore, while the assumptions in previous sections of the researcher being honest and accurate in specifying their deviation sets may feel unpalatably strong, \textit{some} type of assumption on $\tXpost(x)$ must be made. Observe also that a reinterpretation of this result can be understood as implying that an arbitrarily skeptical reviewer can always find an $\tXpost'$ which invalidates the researcher's reported results, should they wish to. In this sense, structure on allowable deviation sets also allows the \textit{researcher} to defend their results against arbitrarily unfavorable counter-assertions of their ``true'' deviation set, which we demonstrate more concretely in Section~\ref{d2:sec:robustness2}.

We now discuss two possible sources of such structure: Section~\ref{d2:sec:robustness1} considers \textit{partial} reports $\tXpost \subseteq \Xpost$, and Section~\ref{d2:sec:robustness2} presents sensitivity analysis for assessing the robustness of results to $\tXpost\neq\Xpost$, more generally.

\subsection{Local Reporting}\label{d2:sec:robustness1}

It may happen that an honest researcher can describe their deviation behavior \textit{local} to the realized $x\in\X$, but has difficulty articulating deviation behavior across all of $\Xpost$. For instance, when $\Xpost$ consists of many disjoint components, the researcher may have a clear understanding of components ``local'' to the realized data draw $x\in\X$, but not for those ``distant'' from $x$. This disjoint structure is particularly likely when there are \textit{many potential motives} for reporting $\hbpost$, but only certain motives are \textit{relevant} at any given draw $X\in\X$. In such a setting, articulating the full $\Xpost$ requires the researcher to enumerate all hypothetical motives for reporting $\hbpost$, otherwise reporting an incomplete description of $\tXpost \subseteq \Xpost.$

We can show that under certain ``local consistency'' conditions in reporting behavior for some partial component $\tXpost\subseteq\Xpost$ containing the data realization (i.e., $X\in\tXpost$), estimators $\quant_{\alpha}^{*}(X,\tXpost)$ and confidence intervals $CI_{\alpha}^{*}(X,\tXpost)$ that condition on said partial component yield valid inferences. Note, however, that these inferences will be less precise than those based on procedures that condition on the full $\Xpost$. For concreteness, we discuss this result in terms of a stylized example.

\begin{example}[Same Specification, Different Motive]\label{d2:ex:schools}
A researcher is studying the effect of peer quality on student performance. Depending on preliminary findings $\hbpre$ from the PAP, the researcher may be prompted to consider the following ex-post analysis $\Spost$: 
\begin{equation}\label{d2:ex:peereq}
\text{TestScore}_i = \alpha + \bpost \cdot \text{PeerScore}_i + \gamma \cdot \text{Controls}_i + \varepsilon_i,
\end{equation}
where $\bpost$ measures the relationship between a student’s outcome and the average test scores of their classroom peers. Depending on PAP outcomes, the researcher may pursue specification $\Spost$ under different interpretive lenses, leading to distinct \textit{motives} for reporting $\hbpost$. For instance, suppose the researcher employs a research design that implies quasi-random variation in peer scores across classrooms (e.g., due to institutional features). After collecting the data, the researcher performs a series of balance checks to verify successful randomization. If such checks pass, the researcher interprets \eqref{d2:ex:peereq} as a strategy for estimating a \textit{causal peer effect} ($m_1$). Should the balance checks fail, however, (e.g., peer scores are found to be tightly correlated with students' prior achievement or classroom assignment), \eqref{d2:ex:peereq} is still of interest, but might now serve as a way to diagnose \textit{sorting mechanisms} in administrative placement ($m_2$). 

These two motives for reporting ($m_1$: causal spillover; $m_2$: diagnostic for sorting) correspond to mutually exclusive interpretations of the data: $m_1$ relies on peer scores being as-good-as-random, while $m_2$ assumes peer scores are non-random and structured. These interpretations then also correspond to disjoint deviation sets $\X_{1}\cap\X_{2}=\varnothing$ with $\X_{1}, \X_{2} \subseteq \Xpost$, defined informally as:
\begin{align*}
    \X_{1} &= \{X \in \X \mid \text{peer-score variance \textit{large}, peer scores \textit{plausibly exogenous}} \}\\
    \X_{2} &= \{X \in \X \mid \text{peer-score assignment \textit{correlated with prior performance}} \}
\end{align*}
When the realized data $x \in \X_{1}$ causes the researcher to view $\bpost$ as a ``causal spillover'' ($m = 1$), they might not be primed to consider the case $x\in\X_{2}$ where $\bpost$ is a ``diagnostic for sorting'' ($m = 2$). As a result, the researcher may report an incomplete deviation set, $\tXpost = \X_{1} \subseteq \Xpost$.
\end{example}

The following proposition gives us that under certain reporting conditions, it is valid to condition on just $\X_1$ when $X\in\X_1$ and $\X_2$ when $X\in\X_2$.

\begin{proposition}\label{d2:prop:local.reporting}
Let $\Xpost = \{X\in\X:\Spost\in\cSpost\}$ be the deviation set for $\hbpost=\Spost(X)$, and let $CI_{\alpha}(X,\tilde{\X})$ be a confidence interval procedure which delivers valid conditional coverage as in \eqref{d2:eq:cond.coverage} for $X\in\tilde{\X}$, i.e.,
$$\P[\ball]{l'\bpost \in CI_{\alpha}(X,\tilde{\X}) \mid X \in \tilde{\X}} \geq 1 - \alpha, \quad \forall \ball.$$ 
For $\{\X_m\}_{m=1}^M$ a partition of $\Xpost$, i.e., $\bigcup_m\X_m = \Xpost$ and $\forall\,m\neq m'$, $\X_{m} \cap \X_{m'}=\varnothing$, define ``local report function''  $\tilde{\mathcal{X}}_{post}(X) : \mathcal{X} \to \{\mathcal{X}_1, \ldots, \mathcal{X}_M\}$ as
\[
\tilde{\mathcal{X}}_{post}(X) = \mathcal{X}_m \quad \text{if } X \in \mathcal{X}_m.
\]
Then $CI_{\alpha}(X,\tXpost(X))$ using local report function $\tilde{\mathcal{X}}_{post}(X)$ delivers valid conditional coverage $\forall X\in\Xpost$, i.e.,
$$\P[\ball]{l'\bpost \in CI_{\alpha}(X,\tXpost(X)) \mid X \in \Xpost} \geq 1 - \alpha, \quad \forall \ball.$$
\end{proposition}

\begin{proof}
See Appendix \ref{d2:app:proof:local.reporting}.
\end{proof}

The crucial requirement here is that the researcher behaves ``locally coherently'' in reporting the incomplete deviation set $\mathcal{X}_m$ for all $X\in\mathcal{X}_m$. To be concrete, consider again the partial deviation sets from Example~\ref{d2:ex:schools}. Define $\mathcal{X}_{1}\subseteq \mathcal{X}_{\mathrm{post}}$ by the following two sample-based conditions:
\begin{enumerate}
    \item[(i)] \textbf{Sufficient dispersion of peer scores across classrooms:} for a prespecified threshold $\delta>0$,
    \[ \widehat{\Var}_{j}\!\left(\,\overline{\text{PeerScore}}_{j}\,\right) > \delta,
    \]
    where $\widehat{\Var}_{j}$ denotes the \emph{sample} variance across classrooms of the classroom mean peer score $\overline{\text{PeerScore}}_{j}$.
    \item[(ii)] \textbf{No baseline correlation with prior achievement:} estimating
    \[
    \text{PreScore}_i \;=\; \eta \,+\, \rho\,\text{PeerScore}_i \,+\, \nu_i,
    \]
    we fail to reject $H_0\!:\rho=0$ at the 5\% level (i.e., $p(\hat\rho)>0.05$).
\end{enumerate}
Together, these define
\[
\mathcal{X}_{1}
=\Big\{\,X\in\mathcal{X}:\;
\widehat{\Var}_{j}\!\left(\overline{\text{PeerScore}}_{j}\right)>\delta
\ \text{and}\ 
p(\hat\rho)>0.05
\Big\}.
\]

This form of ``coherence'' can conceptually be broken into two conditions. The first condition is that for every data draw $X\in\X_{1}$, the researcher reports the same $\X_{1}$ as their deviation set. For this example, the first condition would be violated if there existed some counterfactual draw $X\in\X_{1}$ for which the researcher reported a different threshold $\delta'\neq\delta$ or different significance cutoff $\eta'\neq 0.05$ for $\hat{\rho}$. 

The second condition requires that the partial reports $\{\X_m\}_m$ be \emph{disjoint} subsets of the sample space $\X$, i.e., $\X_m \cap \X_{m'} = \varnothing$ for all $m \neq m'$. Disjointness fits settings where the researcher has multiple, distinct interpretive frames for reporting $\hbpost$, each tied to a separate region of the data, so that at any realization $X$ only one interpretation applies. This requirement is naturally satisfied when interpretations are mutually exclusive, as in Example~\ref{d2:ex:schools}.

\subsection{Sensitivity Analysis}\label{d2:sec:robustness2}

In settings where the deviation set takes the form of a cutoff rule, the researcher may not be able to discern the exact value of the cutoff. In this sense, the true cutoff is fuzzy. For example, suppose $\hbpost$ is of interest because a PAP estimate $\hbpre$ was observed to be small. That is, there exists some cutoff $\kappa_{0} > 0$ for which
\begin{align*}
    \X_{post} = \X_{0} = \{X \in \X: |\hbpre| \leq \kappa_{0}\}.
\end{align*}
However, the exact value of $\kappa_{0}$ may not be obvious to the researcher. After some introspection, the researcher reports $\tilde{\kappa}$ as an approximation to $\kappa_{0}$. This yields reported deviation set
\begin{align*}
    \tilde{\X}_{post} &= \{X \in \X: |\hbpre| \leq \tilde{\kappa}\}. 
\end{align*}
In such settings, a natural robustness exercise is to plot conditionally valid intervals for a range of $\kappa$ above and below the reported $\tilde{\kappa}$. Formally, given $\e > 0$, let $\Tilde{\mathcal{K}}_{\e}$ be a set of $\kappa$ such that $\tilde{\kappa} - \e \leq \kappa \leq \tilde{\kappa} + \e$ for each $\kappa \in \Tilde{\mathcal{K}}_{\e}$. One can plot
\begin{align}\label{d2:eq:CI.range}
    CI_{\alpha}^{*}(X,\X_{\kappa}), \quad \X_{\kappa} = \{X \in \X: |\hbpre| \leq \kappa\}, \quad \forall \kappa \in \Tilde{\mathcal{K}}_{\e}.
\end{align}
This allows one to assess the sensitivity of conclusions to different potential values of $\kappa_{0}$.

This robustness exercise is formally justified under the condition that $|\tilde{\kappa} -\kappa_{0}| \leq \e$ for known $\e$. Under this condition, if the conclusions that the researcher reaches with $CI_{\alpha}^{*}(X, \tXpost)$ can also reached with $CI_{\alpha}^{*}(X,\X_{\kappa})$ for each $\kappa \in \Tilde{\mathcal{K}}_{\e}$, then we know the researcher would reach those conclusions with the true $CI_{\alpha}^{*}(X,\X_{0})$.
While knowledge of $\e$ may be a strong assumption, one can in practice vary $\e$ to determine the largest $\e$ for which the researcher's conclusions persist. Intuitively, if we believe that $\tilde{\kappa} \approx \kappa_{0}$, then we expect $\kappa_{0}$ to fall into $\Tilde{\mathcal{\kappa}}_{\e}$ for some $\e$ that is not too large. Plausible departures from the reported $\tilde{\kappa}$ ought to not lead to massive changes in results.

As intuition for this exercise, see Figure~\ref{d2:fig:cutoffs}, which considers $(\hbpre,\hbpost)$ from Example~\ref{d2:ex:toy}, where $\rho=0.25$. Recall the plots from Figure~\ref{fig:cutoff_distortions}. We consider the exercise of fixing a particular data realization  $(\hbpre,\hbpost)$ and plotting the estimates and confidence intervals as a function of cutoff $\tilde{\kappa}$. 
The green lines show the corrected $CI^*_{0.95}$ for $\hbpost$, with the first panel corresponding to $\Xpost = \{\hbpre \geq \tilde{\kappa}\}$. As one can see, the reported $\tilde{\kappa}$ can yield very different intervals: as one reports $\tilde{\kappa}$ further and further from $\hbpre$, the confidence intervals move from including to excluding zero.
\begin{figure}
\centering
    \includegraphics[width=0.6\textwidth]{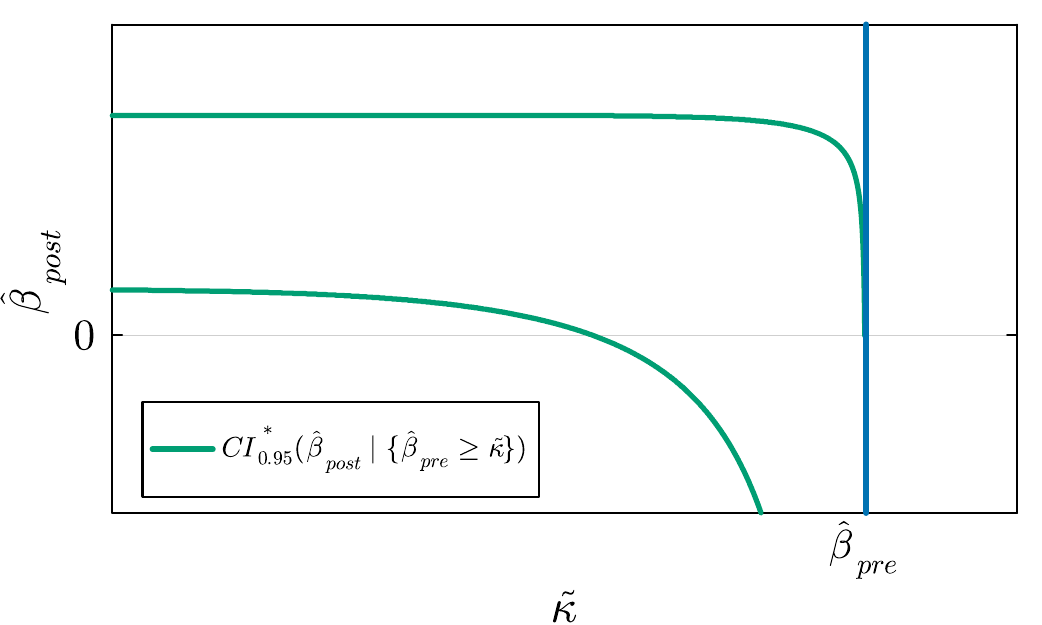} \\ 
    \caption{Effect of moving cutoff $\tilde{\kappa}$ on $CI^*_{0.95}$}
\label{d2:fig:cutoffs}
\end{figure}

\paragraph{Implications for Reporting Conventions.} For deviations based on statistical significance cutoffs $\kappa_{\eta} = z_{1-\eta/2}\sigma_{pre}$, a conventional choice of significance level is $\eta = 0.05$. By sticking to convention, there is less need for the scrutiny in \eqref{d2:eq:CI.range}. But for cutoffs with potentially no obvious conventions, such as those based on economic significance, the robustness exercise in \eqref{d2:eq:CI.range} is useful. To avoid such scrutiny, one can preregister definitions of economic significance for their primary outcomes. Note that the above analysis is also valid for any deviation set $\X_{0}$ known up to a finite set of fuzzy cutoffs $\kappa_{0}$. Moreover, while we focused on confidence intervals $CI_{\alpha}^{*}(X,\X_{\kappa})$, the same sensitivity analysis can be applied to point estimators $\quant_{1/2}^{*}(\X_{\kappa})$.

\section{Conclusion}\label{d2:sec:conclusion}
This paper considers the statistical consequences of deviating from prespecified analysis. We first develop a general model of preregistration in the research process, which we use to demonstrate that PAP deviations can be viewed as a form of conditional reporting that, if left unacknowledged, yields invalid inference. Our framework yields two recommendations. First, researchers should adopt conditional inference as the relevant criteria for non-prespecified analysis. Second, given this conditional criteria, researchers should articulate their corresponding reasons for deviating, then leverage $\Xpost$ to report corrected (i.e., conditional) inferences alongside conventional (i.e., unconditional) inferences. To this end, we provide general and tractable inference procedures for obtaining confidence intervals with correct conditional coverage and point estimators that are conditionally unbiased. We formalize our conditional inference objectives in a model with normally distributed estimates, and present general procedures for constructing confidence intervals and point estimators that are conditionally valid. We then show how to implement these procedures in practice, providing uniformity guarantees for non-normal estimates. Using data from \cite{bessone2021economic}, we demonstrate that, depending on the deviation set and data realization, accounting for conditional reporting can range from having no difference to an economically significant difference relative to conventional practice. In particular, when the data puts one close to the boundary of reporting, the impact of conditioning can be large, whereas when the data realization is sufficiently far from the boundary, there is no change from unconditional inference. This framework has direct implications for considering the validity of past non-preregistered results reported in papers.

We conclude with a discussion of the robustness of our procedures to misspecification of the reported deviation set. Our results suggest possible directions for future work. In particular, this framework may have implications for the optimal length and detail of PAP specification, yielding possible rules of thumb for researchers and journals. This framework may also suggest new paradigms for adaptive experiments with data-driven selection.

\newpage

\bibliography{lit}
\newpage

\appendix

\section*{Appendix}
The outline of the Appendix is as follows.
\begin{itemize}
    \item Appendix \ref{d2:app:sec:counterfactual.PAP}: Bayesian decision problem based on counterfactual PAPs.
    \item Appendix \ref{d2:app:sec:average.coverage}: Conditional coverage for average coverage across studies.
    \item Appendix \ref{d2:app:sec:covariance.structure}: Covariance structure of conditioning events based on multiple polyhedra.
    \item Appendix \ref{d2:app:proofs}: Proofs of results from the main text.
    \item Appendix \ref{d2:app:sec:uniformity}: Proof of uniform asymptotic validity.
\end{itemize}

\section{Deviations Based on Counterfactual PAP}\label{d2:app:sec:counterfactual.PAP}
In Section \ref{d2:sec:decision.problem}, we illustrated how deviations can arise in a standard Bayesian decision problem when $\cS$ is constrained. In this section, we provide a variant of the Bayesian decision problem where deviations can arise even when $\cS$ is \textit{unconstrained}. We give the setup in Section \ref{d2:app:sec:counterfactual.PAP.setup}, and provide a stylized example in \ref{d2:app:sec:stylized.examples} that yields a deviation condition analogous to the one from \cite{bessone2021economic}, discussed in Example \ref{d2:ex:sleep}. 

\subsection{Counterfactual PAP Approach}\label{d2:app:sec:counterfactual.PAP.setup}
If the researcher counterfactually entered the experiment with their posterior beliefs $\pi_{X}$, they would have preregistered 
\begin{equation}\label{d2:app:eq:post.counterfactual}
    \Spost \in \argmin_{S \in \cS} \Bar{R}(S, \pi_{X}), \quad \Bar{R}(S, \pi_{X}) = \int_{\Theta} R(S, \theta) d\pi_{X}(\theta),
\end{equation}
where $\Bar{R}(S, \pi_{X})$ is the average risk of $S$ under $\pi_{X}$. We can think of $\Spost$ as the counterfactual PAP from solving problem \eqref{d2:eq:pre.problem} with $\pi_{X}$ in the place of $\pi$. If one views $\cS$ as the relevant action space, then the counterfactual PAP approach (i.e., solving \eqref{d2:eq:pre.problem} and \eqref{d2:app:eq:post.counterfactual}) aligns with the conditional Bayes principle, which says to choose an action that minimizes average loss under one's current beliefs \citep[Section 1.5.1]{berger2013statistical}.

By construction, the counterfactual PAP minimizes average risk over \textit{future} data draws that are independent of $X$, making it a reasonable point of departure for future researchers studying similar topics. Indeed, the process of solving problems \eqref{d2:eq:pre.problem} and \eqref{d2:app:eq:post.counterfactual} is analogous to conducting a pilot study with initial PAP $\Spre$, seeing results $X_{pre}$, then registering updated PAP $\Spost$ before implementing the experiment at scale on an independent sample $X_{post}$. The difference here is that the same data is used for both steps, $X = X_{pre} = X_{post}$, which must be accounted for when conducting inferences; see Section \ref{d2:sec:inference}. Notably, deviations can arise in \eqref{d2:app:eq:post.counterfactual} even when $\cS$ is unrestricted. In this sense, the dynamic inconsistency that arises in the counterfactual PAP approach is more attributable to the shift in beliefs from $\pi$ to $\pi_{X}$, relative to the posterior average loss approach in \eqref{d2:eq:post.Bayes}.

\subsection{Stylized Example}\label{d2:app:sec:stylized.examples}
To build further intuition for the counterfactual PAP approach, we solve problems (\ref{d2:eq:pre.problem}) and (\ref{d2:app:eq:post.counterfactual}) in a stylized example patterned after the deviation from \cite{bessone2021economic}, discussed in Example \ref{d2:ex:sleep}.  

\begin{itemize}
    \item The data $X = (\htau_{1}, \htau_{2}, \sigma)$ consists of treatment effect estimates $\htau_{k}$ with common standard deviation $\sigma$. Given treatment effect estimands $\tau_{k}$, the estimates are normally distributed $\htau_{k} \sim N(\tau_{k}, \sigma^{2})$, independently across groups $k$.
    \item The parameter $\theta =  (\tau_{1}, \tau_{2}, \sigma)$ contains the treatment effect estimands $\tau_{k}$.
    \item Under the prior distribution $\pi$, estimands are drawn $\tau_{k} \overset{iid}{\sim} N(m, v^{2})$, independently of $\sigma$. In particular, the researcher believes that the estimands $\tau_{k}$ capture some latent causal effect $m$, but maintains a level of uncertainty, captured by $v^{2}$. 
\end{itemize}
Mapping to Example \ref{d2:ex:sleep}, the estimands $(\tau_{1}, \tau_{2})$ represent treatment effects for (i) workers that receive the incentives treatment and (ii) workers that receive the encouragement treatment. In this case, $m$ represents the latent effectiveness of night sleep treatments. 

The researcher wishes to estimate the treatment effect $\tau_{1}$ for the first group, and chooses between estimators $w\htau_{1} + (1-w)\htau_{2}$, for weights $w \in \{1, 1/2\}$. Choosing $w=1$ yields $\htau_{1}$, which is the unbiased plug-in estimator for $\tau_{1}$. Choosing $w=1/2$ yields $(\htau_{1} + \htau_{2})/2$, which is the pooled estimator that may be biased for $\tau_{1}$, but has smaller variance than the plug-in estimator. The loss function is squared error, as in equation \eqref{d2:eq:squared.error}: $L(w\htau_{1} + (1-w)\htau_{2}, \theta) = (w\htau_{1} + (1-w)\htau_{2} - \tau_{1})^{2}$. The risk is given by the mean squared error (MSE):
\begin{equation*}
    R(w, \theta) =  \underbrace{(1-w)^{2}(\tau_{2} - \tau_{1})^{2}}_{\text{squared bias}} + \underbrace{\sigma^{2}(w^{2} + (1-w)^{2})}_{\text{variance}}.
\end{equation*}
MSE reflects the trade-off between bias and variance associated with the choice of weights $w$. Note that, in this setup, the set of specifications $\cS$ under consideration is one-to-one with the weighting schemes $w \in \{1, 1/2\}$. The first scheme puts all the weight on the first group, while the second scheme pools the estimates according to equal shares. 
 
Let $\E_{\pi}[\cdot]$ denote expectations with respect to the prior distribution. To solve problem (\ref{d2:eq:pre.problem}), the researcher compares the average risk of $w = 1$ and $w = 1/2$ under their prior.
\begin{align*}
    \Bar{R}(1, \pi) = \E_{\pi}[\sigma^{2}], \quad \Bar{R}(1/2, \pi) = \frac{v^{2} + \E_{\pi}[\sigma^{2}]}{2}.
\end{align*}
We consider a researcher who expects the sampling variance of the estimates to be low relative to the dispersion of the treatment effect distribution: $\E_{\pi}[\sigma^{2}] < v^{2}$. In this case, it is optimal for the researcher to preregister the specification that yields an unbiased estimator: $\Spre(X) = \htau_{1}$. Intuitively, since $\E_{\pi}[\sigma^{2}] < v^{2}$, the precision gains from pooling are perceived to be low, leading the researcher to choose $w = 1$.

Upon observing $X = (\htau_{1}, \htau_{2}, \sigma)$, the researcher learns the true $\sigma$, and updates their prior beliefs over $\tau_{k}$ to posterior beliefs where the treatment effect estimands are distributed
\begin{align*}
    \tau_{k}|X \sim N\left(\frac{v^{2}}{v^{2} + \sigma^{2}}\htau_{k} + \left(1-\frac{v^{2}}{v^{2} + \sigma^{2}}\right)m, \left(1-\frac{v^{2}}{v^{2} + \sigma^{2}}\right)v^{2}\right),
\end{align*}
independently across $k$. That is, observing $X$ leads the researcher to update their prior mean $m$ towards the treatment effect estimates $\htau_{k}$, while reducing their prior uncertainty. The extent of belief updating is governed by $v^{2}/(v^{2} + \sigma^{2})$, which measures the researcher's prior uncertainty relative to the sampling variability of the treatment effect estimates.  

To solve problem (\ref{d2:app:eq:post.counterfactual}), the researcher compares the average risk of $w = 1$ and $w = 1/2$ under their posterior:
\begin{align*}
    \Bar{R}(1, \pi_{X}) = \sigma^{2}, \quad \Bar{R}(1/2, \pi_{X}) = \left(1-\frac{v^{2}}{v^{2} + \sigma^{2}}\right)\frac{v^{2}}{2} + \left(\frac{v^{2}}{v^{2} + \sigma^{2}}\right)^{2}\frac{(\htau_{2} - \htau_{1})^{2}}{4} + \frac{\sigma^{2}}{2}.
\end{align*}
Thus, it is optimal for the researcher to deviate from their PAP if and only if the  $t$-statistic $\widehat{T} = |\htau_{2} - \htau_{1}|/\sqrt{2\sigma^{2}}$ satisfies
\begin{equation*}
    \widehat{T} \leq \sqrt{\frac{\sigma^{2}}{v^{2}}\left(1 + \frac{\sigma^{2}}{v^{2}}\right)}.
\end{equation*}
In particular, the deviation $\Spost(X) = (\htau_{1} + \htau_{2})/2$ occurs when the treatment effect estimates $\htau_{k}$ are sufficiently homogeneous relative to the researcher's prior beliefs. In such data realizations, the scope for bias is lower than expected, so the researcher pools the estimates. 

In summary, the researcher preregisters $\Spre(X) = \htau_{1}$ and deviates to $\Spost(X) = (\htau_{1} + \htau_{2})/2$ when
\begin{equation*}
    \E_{\pi}[\sigma^{2}] \leq v^{2} \leq \frac{\sigma^{2}}{\widehat{T}^{2}}\left(\frac{1}{2} + \sqrt{\frac{1}{4} + \widehat{T}^{2}}\right).
\end{equation*}
These inequalities reflect the scenario (i) the researcher had decided not to pool from the start and (ii) it became optimal to pool thereafter. If either of these conditions fail, then one would not observe the above switch from not pooling to pooling.

\section{Average Coverage}\label{d2:app:sec:average.coverage}
Consider a sequence of studies $t = 1, \ldots, T$. For each study, there is a baseline vector of estimates $X_{t} \in \R^{d}$ with mean $\mu_{t} \in \R^{d}$ and positive definite covariance matrix $\Sigma_{t} \in \R^{d \times d}_{++}$. The researcher in study $t$ reports $l_{t}'X_{t}$ if and only if $X_{t} \in \X_{t}$, where $l_{t} \in \R^{d}$ is a linear combination vector and $\X_{t} \in \mathcal{L}(\R^{d})$ is a Lebesgue measurable deviation set. In data realizations $X_{t} \in \X_{t}$, researchers use intervals $CI_{\alpha}(X_{t}; \Sigma_{t}, \X_{t}, l_{t})$ to conduct inference on $l_{t}'\mu_{t}$. We suppress the dependence on $(\Sigma_{t}, \X_{t}, l_{t})$, and simply denote the confidence intervals by $CI_{\alpha, t}(X_{t}) = CI_{\alpha}(X_{t}; \Sigma_{t}, \X_{t}, l_{t})$

The average coverage rate (ACR) across studies is
\begin{align*}
    ACR_{T} = \frac{\sum_{t=1}^{T} \1\{l_{t}'\mu_{t} \in CI_{\alpha, t}(X_{t})\} \1\{X_{t} \in \X_{t}\}}{\sum_{t=1}^{T} \1\{X_{t} \in \X_{t}\}}.
\end{align*}
We want an average coverage rate that stays above $1-\alpha$ as $T \to \infty$, in analogy to criteria \eqref{d2:eq:cond.coverage}. Otherwise, the research discipline spanned by the above sequence of studies may systematically produce misleading results. Below we establish a general sense in which conditional coverage for each study is \textit{necessary and sufficient} for controlling average coverage across all studies.

To place structure on the above problem, we suppose there exists an unknown probability distribution $F \in \mathcal{F}$ on the sample space $\mathcal{V} = \R^{d} \times \R^{d \times d}_{++} \times \mathcal{L}(\R^{d}) \times \R^{d}$ such that
\begin{align*}
     V_{t} = (\mu_{t}, \Sigma_{t}, \X_{t}, l_{t}) \overset{iid}{\sim} F, \quad t = 1,\ldots,T,
\end{align*}
where $\mathcal{F}$ is a class of probability distributions on $\mathcal{V}$. This assumption approximates a research environment where similarly qualified researchers pursue independent prespecified analysis, obtain their data, and then potentially deviate from their PAPs.\footnote{Similar sampling assumptions are made in models of publication bias \citep{andrews2019identification}.} Denoting $X = (X_{1}', \ldots, X_{T}')'$, $\mu = (\mu_{1}', \ldots, \mu_{T}')'$, $\Sigma = \text{diag}(\Sigma_{1}, \ldots, \Sigma_{T})$, and $V = (V_{1}, \ldots, V_{T})$, we assume that
\begin{align*}
    X|V \overset{d}{=} X|\mu, \Sigma \sim N(\mu, \Sigma).
\end{align*}
In particular, conditional on $V$, the baseline estimates are independently normally distributed across $t$, where $X_{t}|V \overset{d}{=} X_{t}|V_{t} \sim N(\mu_{t}, \Sigma_{t})$. Thus, since $V_{t}$ is independent across $t$, we have that $(X_{t}, V_{t})$ is independent across $t$. Define $\P[v]{\cdot} = \P[F]{\cdot|V_{t} =v}$, which does not depend on $F$ under our sampling assumptions. Denote $v = (\mu_{v}, \Sigma_{v}, \X_{v}, l_{v})$ for a particular realization $V_{t} = v$, so that $CI_{\alpha,v}(X_{t}) = CI_{\alpha}(X_{t}; \Sigma_{v}, \X_{v}, l_{v})$.

Let $w(v) = \P[v]{X_{t} \in \X_{v}}$ denote the probability of deviation under $N(\mu_{t},\Sigma_{t})|V_{t}=v$, which does not depend on $t$ under our sampling assumptions. The law of iterated expectations yields
\begin{align*}
    \EP[F]{\1\{l_{t}'\mu_{t} \in CI_{\alpha,t}(X_{t})\} \1\{X_{t} \in \X_{t}\}} = \int_{v} w(v)\P[v]{l_{v}'\mu_{v} \in CI_{\alpha,v}(X_{t})|X_{t} \in \X_{v}} dF(v),
\end{align*}
which does not depend on $t$. Thus, law of large numbers implies
\begin{align*}
    \plim_{T \to \infty} \frac{1}{T}\sum_{t=1}^{T} \1\{l_{t}'\mu_{t} \in CI_{\alpha,t}(X_{t})\} \1\{X_{t} \in \X_{t}\} = \int_{v} w(v) \P[v]{l_{v}'\mu_{v} \in CI_{\alpha,v}(X_{t})|X_{t} \in \X_{v}} dF(v).
\end{align*}
An analogous argument for the denominator, together with Slutsky's lemma, yields
\begin{align*}
    ACR_{F} = \plim_{T \to \infty}ACR_{T} = \frac{\displaystyle \int_{v} w(v)\P[v]{l_{t}'\mu_{t} \in CI_{\alpha,v}(X_{t})|X_{t} \in \X_{v}} dF(v)}{\displaystyle  \int_{v} w(v) dF(v)},
\end{align*}
provided that the denominator is positive. 

We are now prepared to state the result. In what follows, $\mathcal{V}_{+} = \{v \in \mathcal{V}: w(v) > 0\}$ denotes the set of $v$ for which the probability of deviation $w(v) = \P[v]{X_{t} \in \X_{v}}$ under $N(\mu_{t},\Sigma_{t})|V_{t}=v$ is positive, and $\mathcal{F}_{+}$ denotes the set of distributions on $\mathcal{V}_{+}$.

\begin{proposition}\label{d2:app:prop:necessary.sufficient}
Suppose that $\mathcal{F}$ satisfies $\mathcal{F}_{+} \subseteq \mathcal{F}$ and $\displaystyle \inf_{F \in \mathcal{F}}\int_{v} w(v) dF(v) > 0$. Then, to obtain correct average coverage across studies in the sense that
\begin{align*}
    ACR_{F} \geq 1-\alpha, \quad \forall F \in \mathcal{F},
\end{align*}
it is necessary and sufficient to ensure conditional coverage for each study in the sense that
\begin{align*}
    \P[v]{l_{v}'\mu_{v} \in CI_{\alpha,v}(X_{t})|X_{t} \in \X_{v}} \geq 1-\alpha, \quad \forall v \in \mathcal{V}_{+}.
\end{align*}
\end{proposition}

\begin{proof}[Proof of Proposition \ref{d2:app:prop:necessary.sufficient}]
$ $ \newline
For the sufficiency direction, observe that for all $F \in \mathcal{F}$,
\begin{align*}
    ACR_{F} &= \frac{\displaystyle \int_{v \in \mathcal{V}_{+}} w(v)\P[v]{l_{v}'\mu_{v} \in CI_{\alpha,v}(X_{t})|X_{t} \in \X_{v}} dF(v)}{\displaystyle  \int_{v \in \mathcal{V}_{+} } w(v) dF(v)} \\
    &\geq \frac{\displaystyle \int_{v \in \mathcal{V}_{+}} w(v) (1-\alpha) dF(v)}{\displaystyle \int_{v \in \mathcal{V}_{+}} w(v) dF(v)} = 1-\alpha.
\end{align*}
For the necessity direction, suppose there exists $v \in \mathcal{V}_{+}$ for which 
\begin{align*}
    \P[v]{l_{v}'\mu_{v} \in CI_{\alpha,v}(X_{t})|X_{t} \in \X_{v}} < 1-\alpha.
\end{align*}
But since $\mathcal{F}_{+} \subseteq \mathcal{F}$, the probability distribution $F_{v} \in \mathcal{F}_{+}$ that puts probability one on $v$ must be contained in $\mathcal{F}$, which yields $ACR_{F_{v}} < 1-\alpha$. Thus, by contraposition, the necessity direction holds. 
\end{proof}

Proposition \ref{d2:app:prop:necessary.sufficient} shows that, if the class of distributions $\mathcal{F}$ that generates $V_{t} = (\mu_{t}, \Sigma_{t}, \X_{t}, l_{t})$ is rich enough, then conditional coverage in each study is necessary and sufficient for obtaining correct average coverage across studies.\footnote{\citet[Example 1]{fithian2014optimal} consider an analogous setup, but do not show necessity.} Intuitively, when $\mathcal{F}$ is rich enough, there will exist unfavorable distributions $F \in \mathcal{F}$ that invalidate unconditional inference procedures. Allowing for rich classes of $\mathcal{F}$ in this way seems warranted, given that deviations occur precisely because researchers cannot anticipate them.

\section{Covariance Structure of Conditioning Event}\label{d2:app:sec:covariance.structure}
The truncation quantities $(\hZminuspostk, \hZpluspostk, \hZzeropostk)$ depend on the sign and magnitude of the terms $(\Apostk \coef)_{j}$. These terms are the coefficients from regressing the conditioning event variables $\Apostk \hball$ on the estimator of interest $l'\hbpost$:
\begin{align*}
    (\Apostk \coef)_{j} = \paren{\Apostk \frac{\Sigmaall l_{post}}{\sigma_{post}^{2}}}_{j} = \paren{\Apostk\frac{\Cov_{\ball}(\hball, l'\hbpost)}{\Var_{\ball}(l'\hbpost)}}_{j} = \paren{\frac{\Cov_{\ball}(\Apostk\hball, l'\hbpost)}{\Var_{\ball}(l'\hbpost)}}_{j}.
\end{align*}
Thus, the truncated distribution used to construct $\quant_{\alpha}^{*}$ depends on the covariance structure of the conditioning event:
\begin{align*}
    \Cov_{\ball}(\Apostk\hball, l'\hbpost), \quad k = 1,\ldots,K.
\end{align*}
Of course, $\quant_{\alpha}^{*}$ also depends on the variance term $\sigma_{post}^{2} = \Var_{\ball}(l'\hbpost)$, but since this term does not depend on $(\Apostk)_{k=1}^{K}$, we ignore it in this discussion.

To illustrate, suppose one's experiment consists of pilot data $X_{0}$ and primary data $X_{1}$ from independent samples; that is, $X = (X_{0}, X_{1})$, where $X_{0}$ is independent of $X_{1}$. If the deviation set for $\hbpost = \Spost(X_{1})$ is based on estimates $\hball_{0} \subseteq \hball$ computed on $X_{0}$, then there exist matrices $(A_{post,k,0})_{k=1}^{K}$ such that $\Apostk\hball = A_{post,k,0}\hball_{0}$ for each $k$. In this case,
\begin{align*}
    \Cov_{\ball}(\Apostk\hball, l'\hbpost) =   A_{post,k,0}\Cov_{\ball}(\hball_{0}, l'\hbpost) = 0, \quad \forall k.
\end{align*}
Thus, $[\hZminuspostk, \hZpluspostk] = [-\infty, +\infty]$ for all $k$ so that $\widehat{\trunc}_{post} = [-\infty, +\infty]$. In this case, by equation \eqref{d2:eq:unconditional.QUE}, conditional inference based on $\quant_{\alpha}^{*}$ reduces to conventional unconditional inference.

As highlighted above, the adjustments produced by $\quant_{\alpha}^{*}$ depend on the covariance structure of the conditioning event. When there is zero covariance, our proposed conditional intervals and point estimators reduce to their unconditional analogues:
\begin{align*}
    \Apostk \coef = 0, \quad \forall k
    \quad \implies \quad 
    \begin{pmatrix}
    CI_{\alpha}^{*}(X) \\
    \quant_{1/2}^{*}
    \end{pmatrix} 
    = 
    \begin{pmatrix}
    [l'\hbpost \pm z_{1-\alpha/2}\sigma_{post}] \\
    l'\hbpost
    \end{pmatrix}.
\end{align*}

Moreover, reductions to unconditional inference can occur even without full independence. For example, suppose the estimates $\hball$ fall into a polyhedra $k$ where $\Apostk \coef = 0$. Then, we have $\widehat{\trunc}_{post,k} = [-\infty, +\infty]$, which yields $\widehat{\trunc}_{post} = [-\infty, +\infty]$. Notably, this can occur even if there are other polyhedra $k'\neq k$ with $\Apostk \coef \neq 0$. Thus, even without independent samples, it is possible for inferences based on $\quant_{\alpha}^{*}$ to align with conventional ones, depending on the particular data realization.

\section{Proofs of Results from the Main Text}\label{d2:app:proofs}

\subsection{Proof of Proposition \ref{d2:prop:pfanzagl}}\label{d2:app:proof:pfanzagl}
It suffices to verify that our setup satisfies the conditions of \citet[Theorem 5.5.13]{Pfanzagl1994}. To this end, first note that observing $\hball$ is equivalent to observing $(l'\hbpost, \hat{r})$, since $\hat{r} = \hball - \coef l'\hbpost$ for known $\coef$. Moreover,
\begin{align*}
\begin{pmatrix}
    l'\hbpost \\
    \hat{r}
\end{pmatrix}
\sim
N\paren{ 
\begin{pmatrix}
    l'\bpost \\
    \ball - \coef l'\bpost
\end{pmatrix}
,
\begin{pmatrix}
    \sigma_{post}^{2} & 0 \\
    0 & \Sigmaall_{r}
\end{pmatrix}
}, \quad \Sigmaall_{r} = \Sigmaall - \displaystyle \frac{\Sigmaall l_{post}l_{post}'\Sigmaall}{\sigma_{post}^{2}}, \quad
\forall \ball \in \R^{\dim(\ball)}.
\end{align*}
Following \citet[Section 8a.4]{rao1973linear}, there exists a dominating measure for this class of normal distributions, with corresponding density $f(v; l'\bpost)f(r; \ball - \coef l'\bpost)$, where
\begin{align*}
    f(v; \theta) = \frac{1}{\sqrt{2\pi \sigma_{post}^{2}}} \exp\paren{\frac{-(v - \theta)^{2}}{2\sigma_{post}^{2}}}, \quad \theta \in \Theta = \curly{l_{post}'\ball: \ball \in \R^{\dim(\ball)}} = \R,
\end{align*}
and
\begin{align*}
    f(r; \eta) = \frac{1}{\sqrt{\abs{2\pi\Sigmaall_{r}}_{+}}} \exp\paren{\frac{-(r - \eta)'\Sigmaall_{r}^{+}(r - \eta)}{2}}, \quad \eta \in H = \curly{(I - \coef l_{post}')\ball: \ball \in \R^{\dim(\ball)}},
\end{align*}
where $(\Sigmaall_{r}^{+}, \abs{\Sigmaall_{r}}_{+})$ denote the Moore-Penrose inverse and pseudo-determinant of $\Sigmaall_{r}$, respectively. Note that $\ball \mapsto (l'\bpost, \ball - \coef l'\bpost)$ is an invertible transformation. Thus, it suffices to show that $\quant_{\alpha}^{*}$ dominates quantile conditionally unbiased estimators of the form $\quant_{\alpha} = \mu_{\alpha}(l'\hbpost, \hat{r})$ over the class of conditional densities given by
\begin{align*}
    f(v; \theta)f(r; \eta) \frac{\1\curly{r + \gamma v \in \Bpost}}{\P[\theta, \eta]{\hat{r} + \gamma l'\hbpost \in \Bpost}}, \quad (\theta, \eta) \in \Theta \times H,
\end{align*}
where the aforementioned dominating measure is left implicit.

The above densities and the estimator $\quant_{\alpha}^{*}$ satisfy the conditions of \citet[Theorem 5.5.13]{Pfanzagl1994}. The conclusion of that theorem together with \citet[Proposition 2.5.3]{Pfanzagl1994} yields
\begin{equation*}
    \E_{\theta, \eta}[L(\quant_{\alpha}^{*}, \theta)|\hball \in \Bpost] \leq \E_{\theta, \eta}[L(\quant_{\alpha}, \theta)|\hball \in \Bpost], \quad \forall (\theta, \eta) \in \Theta \times H,
\end{equation*}
for any quasiconvex loss function $d \mapsto L(d, \theta)$ that attains its minimum at $d = \theta$. By definition of $\Theta$ and $H$ and invertibility of $\ball \mapsto (l'\bpost, \ball - \coef l'\bpost)$, the conclusion follows.

\subsection{Proof of Proposition \ref{d2:prop:highprob.deviations}}\label{d2:app:proof:highprob.deviations}
Let $\quant_{\alpha}^{\Phi} = l'\hbpost + z_{\alpha}\sigma_{post}$. To show $\P[\ball_{m}]{|\quant_{\alpha,m}^{*} - \quant_{\alpha}^{\Phi}| > \e|\hball \in \Bpost^{m}} \to 0$, it suffices to show
\begin{align*}
    \P[\ball_{m}]{|\quant_{\alpha,m}^{*} - \quant_{\alpha}^{\Phi}| > \e|\hball \in \Bpost^{m}, \hat{r}_{m}} \to[p] 0
\end{align*}
Since $\P[\ball_{m}]{\hball \in \Bpost^{m}} \to 1$ implies $\P[\ball_{m}]{\hball \in \Bpost^{m}|\hat{r}_{m}} \to[p] 1$, to show that the above converges in probability to zero when $\P[\ball_{m}]{\hball \in \Bpost^{m}} \to 1$, it suffices to show that
\begin{align}\label{d2:app:highprob.condition}
    \P[\ball_{m}]{\hball \in \Bpost^{m}|\hat{r}_{m}} \to[p] 1 \quad \implies \quad \P[\ball_{m}]{|\quant_{\alpha,m}^{*} - \quant_{\alpha}^{\Phi}| > \e|\hball \in \Bpost^{m}, \hat{r}_{m}} \to[p] 0.
\end{align}
Following similar notation to \citet[Proof of Proposition 9]{andrews2024inference}, denote
\begin{align*}
    g(\Bpost, \ball, r) &= \P[\ball]{\hball \in \Bpost|\hat{r} = r}, &
    \mathcal{G}(\kappa) &= \{(\Bpost, \ball, r): 1-g(\Bpost, \ball, r) \leq \kappa \}, \\
    h_{\e}(\Bpost, \ball, r) &= \P[\ball]{|\quant_{\alpha}^{*} - \quant_{\alpha}^{\Phi}| > \e |\hball \in \Bpost, \hat{r} = r}, &
    \mathcal{H}(\e, \delta) &= \{(\Bpost,\ball, r): h_{\e}(\Bpost, \ball, r) \leq \delta\}.
\end{align*}
Suppose we can show that for any sequence $(\Bpost^{m}, \ball_{m}, r_{m})$ where $g(\Bpost^{m}, \ball_{m}, r_{m}) \to 1$, we have $h_{\e}(\Bpost^{m}, \ball_{m}, r_{m}) \to 0$. Then, for each $(\e, \delta)$ there exists $\kappa(\e, \delta)$ such that $\mathcal{G}(\kappa(\e, \delta)) \subseteq \mathcal{H}(\e, \delta)$. In this case, we have 
\begin{align*}
     \P[\ball_{m}]{1 - \P[\ball_{m}]{\hball \in \Bpost^{m}|\hat{r}_{m}} \leq \kappa(\e,\delta)} &= \P[\ball_{m}]{(\Bpost^{m}, \ball_{m}, \hat{r}_{m}) \in \mathcal{G}(\kappa(\e, \delta))} \\
     &\leq \P[\ball_{m}]{(\Bpost^{m}, \ball_{m}, \hat{r}_{m}) \in \mathcal{H}(\e, \delta)} \\
     &= \P[\ball_{m}]{\P[\ball_{m}]{|\quant_{\alpha,m}^{*} - \quant_{\alpha}^{\Phi}| > \e|\hball \in \Bpost^{m}, \hat{r}_{m}} \leq \delta},
\end{align*}
By $\P[\ball_{m}]{\hball \in \Bpost^{m}|\hat{r}_{m}} \to[p] 1$, the LHS converges to one for any $\kappa$, so that the RHS converges to one for any $(\e, \delta)$, which implies \eqref{d2:app:highprob.condition}. Thus, to conclude the proof, it suffices to show that for any sequence $(\Bpost^{m}, \ball_{m}, r_{m})$ where $g(\Bpost^{m}, \ball_{m}, r_{m}) \to 1$, we have $h_{\e}(\Bpost^{m}, \ball_{m}, r_{m}) \to 0$. 

Let $(\Bpost^{m}, \ball_{m}, r_{m})$ be such a sequence. Recall $\trunc_{post}^{m}(r) = \{ z \in \R: r + \gamma z \in \Bpost^{m}\}$ so that
\begin{align*}
    g(\Bpost^{m}, \ball_{m}, r_{m}) = \P[\ball_{m}]{l'\hbpost \in \trunc_{post}^{m}(r_{m})|\hat{r}_{m} = r_{m}} = \P[\ball_{m}]{l'\hbpost \in \trunc_{post}^{m}(r_{m})},
\end{align*}
where the last equality follows from independence of $\hat{r}_{m}$ and $l'\hbpost$. Likewise, letting $\quant_{\alpha,m}^{*}(r)$ denote the unique solution to $F_{TN}(l'\hbpost; \quant_{\alpha,m}^{*}(r), \sigma_{post}^{2}, \trunc_{post}^{m}(r)) = 1-\alpha$ given $l'\hbpost \in \trunc_{post}^{m}(r)$, we have
\begin{align*}
    h_{\e}(\Bpost^{m}, \ball_{m}, r_{m}) &= \P[\ball_{m}]{|\quant_{\alpha,m}^{*}(r_{m}) - \quant_{\alpha}^{\Phi}| > \e|l'\hbpost \in \trunc_{post}^{m}(r_{m}), \hat{r}_{m} = r_{m}} \\
    &= \P[\ball_{m}]{|\quant_{\alpha,m}^{*}(r_{m}) - \quant_{\alpha}^{\Phi}| > \e|l'\hbpost \in \trunc_{post}^{m}(r_{m})}.
\end{align*}
Thus, to show \eqref{d2:app:highprob.condition}, it suffices to show
\begin{align*}
    \P[\ball_{m}]{l'\hbpost \in \trunc_{post}^{m}(r_{m})} \to 1 \quad \implies \quad \P[\ball_{m}]{|\quant_{\alpha,m}^{*}(r_{m}) - \quant_{\alpha}^{\Phi}| > \e|l'\hbpost \in \trunc_{post}^{m}(r_{m})} \to 0.
\end{align*}
As we argue below, this setup is amenable to \citet[Lemma 5.10]{van2000asymptotic}. To set things up, denote $\Bar{\mu}_{\alpha, m}^{*} = (\quant_{\alpha,m}^{*}(r_{m})-l'\hbpost)/\sigma_{post}$ and $\Bar{\trunc}_{post}^{m} = (\trunc_{post}^{m}(r_{m}) - l'\hbpost)/\sigma_{post}$, and note that given $0 \in \Bar{\trunc}_{post}^{m}$,
\begin{align*}
    F_{TN}(l'\hbpost; \quant_{\alpha,m}^{*}(r_{m}), \sigma_{post}^{2}, \trunc_{post}^{m}(r_{m})) &= \frac{\displaystyle \int_{-\infty}^{l'\hbpost} \phi\paren{\frac{t-\quant_{\alpha,m}^{*}(r_{m})}{\sigma_{post}}}\1\curly{t \in \trunc_{post}^{m}(r_{m})} dt}{\displaystyle \int_{-\infty}^{\infty}\phi\paren{\frac{t-\quant_{\alpha,m}^{*}(r_{m})}{\sigma_{post}}}\1\curly{t \in \trunc_{post}^{m}(r_{m})} dt} \\
    &= \frac{\displaystyle \int_{-\infty}^{0} \phi(u-\Bar{\mu}_{\alpha, m}^{*})\1\curly{u \in \Bar{\trunc}_{post}^{m}} du}{\displaystyle \int_{-\infty}^{\infty} \phi(u-\Bar{\mu}_{\alpha, m}^{*}) \1\curly{u \in \Bar{\trunc}_{post}^{m}} du} \\
    &= F_{TN}(0; \Bar{\mu}_{\alpha, m}^{*}, 1, \Bar{\trunc}_{post}^{m}),
\end{align*}
which follows from the substitution $u = (t - l'\hbpost)/\sigma_{post}$. Since $\quant_{\alpha,m}^{*}(r_{m})$ is the unique solution for which the LHS equals $1-\alpha$, then $\Bar{\mu}_{\alpha, m}^{*}$ is the unique zero to
\begin{align*}
    \Psi_{m}(\Bar{\mu}_{\alpha, m}^{*}) = (1-\alpha) - F_{TN}(0; \Bar{\mu}_{\alpha, m}^{*}, 1, \Bar{\trunc}_{post}^{m}) = 0, \quad \forall m.
\end{align*}
Given $0 \in \Bar{\trunc}_{post}^{m}$, the maps $z \mapsto \Psi_{m}(z)$ are strictly increasing over $z \in \R$ and have unique zeros $\Bar{\mu}_{\alpha, m}^{*}$. Thus, we appeal to the arguments of \citet[Lemma 5.10]{van2000asymptotic}. In particular, if there exists a function $z \mapsto \Psi(z)$ such that (i) $\Psi_{m}(z)|\{0 \in \Bar{\trunc}_{post}^{m}\} \to[p] \Psi(z)$ for all $z \in \R$ and (ii) $\Psi(z_{\alpha} - \e) < 0 < \Psi(z_{\alpha} + \e)$ for every $\e > 0$, then we obtain $\Bar{\mu}_{\alpha, m}^{*}|\{0 \in \Bar{\trunc}_{post}^{m}\} \to[p] z_{\alpha}$. Indeed, under the monotonicity property, we have 
\begin{align*}
    \P[\ball_{m}]{\Psi_{m}(z_{\alpha} - \e) < 0 < \Psi_{m}(z_{\alpha} + \e)|0 \in \Bar{\trunc}_{post}^{m}} \leq \P[\ball_{m}]{z_{\alpha} - \e < \Bar{\mu}_{\alpha, m}^{*} < z_{\alpha} + \e|0 \in \Bar{\trunc}_{post}^{m}},
\end{align*}
The LHS converges to one, since under conditions (i) and (ii) we have $(\Psi_{m}(z_{\alpha} - \e), \Psi_{m}(z_{\alpha} + \e))|\{0 \in \Bar{\trunc}_{post}^{m}\} \to[p] (\Psi(z_{\alpha} - \e), \Psi(z_{\alpha} + \e))$ and $\Psi(z_{\alpha} - \e) < 0 < \Psi(z_{\alpha} + \e)$. Thus, the conclusion of the lemma yields
\begin{align*}
    \P[\ball_{m}]{|\quant_{\alpha,m}^{*}(r_{m}) - \quant_{\alpha}^{\Phi}| > \e | l'\hbpost \in \trunc_{post}^{m}(r_{m})} = \P[\ball_{m}]{|\Bar{\mu}_{\alpha, m}^{*} - z_{\alpha}| > \e | 0 \in \Bar{\trunc}_{post}^{m}} \to 0,
\end{align*}
which is the desired convergence. We now show that $\P[\ball_{m}]{l'\hbpost \in \trunc_{post}^{m}(r_{m})} \to 1$ implies the existence of $\Psi(z)$ satisfying conditions (i) and (ii), which will conclude the proof.

Consider $\Psi(z) = (1-\alpha)-\Phi(-z)$, which has unique zero at $z_{\alpha}$, and hence satisfies condition (ii). Now, in considering $\P[\ball_{m}]{\abs{\Psi_{m}(z) - \Psi(z)} > \e|0 \in \Bar{\trunc}_{post}^{m}}$, we bound the inner absolute value as follows. Letting $\xi \overset{d}{=} l'\hbpost - l_{post}'\ball_{m} \overset{\ball_{m}}{\sim} N(0, \sigma_{post}^{2})$, we have
\begin{align*}
    &\abs{\Psi_{m}(z) - \Psi(z)} = \abs{F_{TN}(0; z, 1, \Bar{\trunc}_{post}^{m}) - \Phi(-z)} \\
    &= \abs{F_{TN}(\xi; \xi + z\sigma_{post}, \sigma_{post}^{2}, \trunc_{post}^{m}(r_{m}) - l_{post}'\ball_{m}) - \Phi(-z)} \\
    &\leq \frac{\displaystyle \abs{\int_{-\infty}^{\xi} \phi\paren{\frac{t-(\xi + z\sigma_{post})}{\sigma_{post}}}(\1\{t \in \trunc_{post}^{m}(r_{m}) - l_{post}'\ball_{m}\}-1)dt}}{\displaystyle \int_{-\infty}^{\infty}\phi\paren{\frac{t-(\xi + z\sigma_{post})}{\sigma_{post}}}\1\{t \in \trunc_{post}^{m}(r_{m}) - l_{post}'\ball_{m}\}dt} \\
    &+ \abs{\frac{1}{\displaystyle \int_{-\infty}^{\infty}\frac{1}{\sigma_{post}}\phi\paren{\frac{t-(\xi + z\sigma_{post})}{\sigma_{post}}}\1\{t \in \trunc_{post}^{m}(r_{m}) - l_{post}'\ball_{m}\}dt} - 1}\Phi(-z) \\
    &\leq \frac{\displaystyle \int_{-\infty}^{\infty} \phi\paren{\frac{t-(\xi + z\sigma_{post})}{\sigma_{post}}}\1\{t \notin \trunc_{post}^{m}(r_{m}) - l_{post}'\ball_{m}\}dt}{\displaystyle \int_{-\infty}^{\infty}\phi\paren{\frac{t-(\xi + z\sigma_{post})}{\sigma_{post}}}\1\{t \in \trunc_{post}^{m}(r_{m}) - l_{post}'\ball_{m}\}dt} \\
    &+ \abs{\frac{1}{\displaystyle \int_{-\infty}^{\infty}\frac{1}{\sigma_{post}}\phi\paren{\frac{t-(\xi + z\sigma_{post})}{\sigma_{post}}}\1\{t \in \trunc_{post}^{m}(r_{m}) - l_{post}'\ball_{m}\}dt} - 1}\Phi(-z) \\
    &= f_{m}(\xi, z).
\end{align*}
For a given realization $\xi = v$, and letting $\xi_{z,v} \sim N(v + z\sigma_{post}, \sigma_{post}^{2})$, the above is equal to
\begin{align*}
    f_{m}(v, z) = \frac{\P[]{\xi_{z,v} \notin \trunc_{post}^{m}(r_{m}) - l_{post}'\ball_{m}}}{\P[]{\xi_{z,v} \in \trunc_{post}^{m}(r_{m}) - l_{post}'\ball_{m}}} + \abs{\frac{1}{\P[]{\xi_{z,v} \in \trunc_{post}^{m}(r_{m}) - l_{post}'\ball_{m}}} - 1}\Phi(-z)
\end{align*}
Since $\P[]{\xi \in \trunc_{post}^{m}(r_{m}) - l_{post}'\ball_{m}} \to 1$ and $\xi_{z,v}$ has full support for each $(z,v)$, then $f_{m}(v,z) \to 0$ for each $(z,v)$. In particular, $f_{m}(\xi,z)$ converges to zero almost surely: $\P[]{\lim_{m \to \infty}f_{m}(\xi,z) = 0} = 1$. This implies convergence in probability $\P[]{f_{m}(\xi,z) > \e} \to 0$, and hence
\begin{align*}
    \P[\ball_{m}]{\abs{\Psi_{m}(z) - \Psi(z)} > \e|0 \in \Bar{\trunc}_{post}^{m}} &\leq \P[]{f_{m}(\xi,z) > \e|\xi \in \trunc_{post}^{m}(r_{m}) - l_{post}'\ball_{m}} \\
    &\leq \frac{\P[]{f_{m}(\xi,z) > \e}}{\P[]{\xi \in \trunc_{post}^{m}(r_{m}) - l_{post}'\ball_{m}}} \\
    &\to 0.
\end{align*}
Thus, $\Psi(z) = (1-\alpha)-\Phi(-z)$ satisfies condition (i).

\subsection{Proof of Proposition \ref{d2:prop:lee.representation}}\label{d2:app:proof:lee.representation}
For each $k$, \citet[Lemma 5.1]{lee2016exact} implies
\begin{align*}
    \{A_{post, k} \hball \leq c_{post, k}\} &= \{A_{post, k} (\gamma l'\hbpost + \hat{r}) \leq c_{post, k}\}  \\
    &= \{A_{post, k}\gamma l'\hbpost \leq c_{post, k} - A_{post, k}\hat{r}\} \\
    &= \curly{(A_{post, k}\gamma)_{j} l'\hbpost \leq (c_{post, k})_{j} - (A_{post, k}\hat{r})_{j}, \forall j} \\
    &= \curly{l'\hbpost \geq \frac{(c_{post, k})_{j} - (A_{post, k}\hat{r})_{j}}{(A_{post, k}\gamma)_{j}}, \quad \forall j \in J_{k}^{-}} \\
    &\bigcap \curly{l'\hbpost \leq \frac{(c_{post, k})_{j} - (A_{post, k}\hat{r})_{j}}{(A_{post, k}\gamma)_{j}}, \quad \forall j \in J_{k}^{+}} \\
    &\bigcap \curly{(c_{post, k})_{j} - (A_{post, k}\hat{r})_{j} \geq 0, \quad  \forall j \in J_{k}^{0}}
\end{align*}
Thus, for each $k$, the event $\{A_{post, k} \hball \leq c_{post, k}\}$ is equivalent to $\{l'\hbpost \in \widehat{\trunc}_{post,k}\}$. Following the same above steps, we also have
\begin{align*}
    \widehat{\trunc}_{post} = \{z \in \R: \hat{r} + \gamma z \in \Bpost\} &= \bigcup_{k=1}^{K}\curly{z \in \R: \Apostk(\gamma z + \hat{r}) \leq \cpostk} \\
    &= \bigcup_{k=1}^{K}\curly{z \in \R: z \in \widehat{\trunc}_{post,k}} \\
    &= \bigcup_{k=1}^{K}\widehat{\trunc}_{post,k}.
\end{align*}

\subsection{Proof of Proposition \ref{d2:prop:impossibility}}\label{d2:app:proof:impossibility}
Let $(x_{0}, \tXo)$ be as in the conditions, so we can take open neighborhood  $U \subseteq \tXo$ of $x_0$ such that $\beta_{0} \notin CS_{\alpha}(x, \tXo)$ for all $x \in U$.
By definition of the set
\[
A = \left\{ x \in \R^{\dim(X)}  : \exists\, \tXpost \subseteq \R^{\dim(X)} \text{ such that } x \in \tXpost \text{ and } \beta_{0} \notin CS_{\alpha}(x, \tXpost) \right\},
\]
we have that every $x \in U$ belongs to $A$ by taking $\tXpost = \tXo$. Since $U$ is an open neighborhood, it has positive Lebesgue measure. Thus, $A$ contains a set of positive measure, and so $A$ itself has positive Lebesgue measure. By continuity and full support yielding density $p(x) > 0$ for all $x \in \R^{\dim(X)}$, we therefore obtain 
$$\P{X\in A} \geq \P{X\in U} = \int_{U} p(x)dx > 0,$$
meaning that $A$ has strictly positive probability under the distribution of $X$. Thus, a researcher who deviates for $\Xpost = A$ will do so with positive probability. Moreover, by definition of $A$, one can report $\tXpost$ as a function of $x \in A$ to ensure that $\beta_{0}$ will never be in $CS_{\alpha}(x, \tXpost)$, which yields conditional coverage equal to zero.

\subsection{Proof of Proposition \ref{d2:prop:local.reporting}}\label{d2:app:proof:local.reporting}
By the law of total probability over the partition $\Xpost=\bigcup_{m}\X_m$, for all $\ball$,
    \begin{align*}
    &\P[\ball]{l'
    \bpost \in CI_{\alpha}(X,\tXpost(X)) \mid X \in \Xpost} \\
    &= \sum_{m=1}^M \P[\ball]{X \in \X_m \mid X \in \Xpost} \cdot \P[\ball]{l'
    \bpost \in CI_{\alpha}(X,\tXpost(X)) \mid X \in \X_m} \\
    &= \sum_{m=1}^M \P[\ball]{X \in \X_m \mid X \in \Xpost} \cdot \underbrace{\P[\ball]{l'\bpost \in CI_{\alpha}(X,\X_m) \mid X \in \X_m}}_{\geq 1-\alpha} \\
    &\geq (1-\alpha)\cdot \underbrace{\sum_{m=1}^M \P[\ball]{X \in \X_m \mid X \in \Xpost}}_{=1} \\ 
    &= 1-\alpha.
    \end{align*}

\section{Uniform Asymptotic Validity}\label{d2:app:sec:uniformity}
We establish the uniform asymptotic validity of the feasible conditional inference procedures discussed Section \ref{d2:sec:feasible.inference} of the main text. The following analyses and proofs follow the approaches of \cite{andrews2024inference} and \cite{mccloskey2024hybrid}. To make this section self-contained, we once again summarize the environment and procedure---though with some additional discussion, notation, and remarks.

\subsection{Setup}

\paragraph{Environment.} We suppose that a sample of size $n$ is drawn from some unknown distribution $P_{n} \in \cPn$, where $\cPn$ is a class of probability distributions corresponding to a sample of size $n$. For example, given a class of distributions $\mathcal{P}_{0}$ with bounded moments, if we have i.i.d. draws of size $n$ from some fixed $P_{0} \in \mathcal{P}_{0}$, then $P_{n} = (P_{0})^{n}$ is the product distribution and $\mathcal{P}_{n}$ is the set of all products of a fixed distribution with bounded moments. As we take $n \to \infty$, there is a corresponding sequence of unknown distributions $\{P_{n}\} \in \times_{n=1}^{\infty}\cPn$, where $\times_{n=1}^{\infty}\cPn$ is the sequence of distribution classes, and the notation $\{P_{n}\} \in \times_{n=1}^{\infty}\cPn$ means that $P_{n} \in \cPn$ for all $n$. We use $P \in \cup_{n=1}^{\infty}\cPn$ to index elements belonging to $\mathcal{P}_{n}$ for some $n$.

\paragraph{Estimates.} Given a sample of size $n$ from the unknown $P_{n}$, the researcher constructs PAP and post estimates, which we collect into a vector $\hballn$, alongside a corresponding covariance matrix estimator $\hSigmaalln$. The researcher is interested in a linear combination of a particular vector of post estimates: $l'\hat{\beta}_{post,n}$. Let $l_{post}$ denote the vector induced by (i) matrix multiplication to select $\hat{\beta}_{post,n}$ from $\hballn$ and (ii) multiplication of $\hat{\beta}_{post,n}$ by $l$. Thus, $l'\hat{\beta}_{post,n} = l_{post}'\hballn$. We denote the corresponding variance estimator by $\hsigmapostn^{2} = l_{post}'\hSigmaalln l_{post}$.

\paragraph{Conditioning Event.} The conditioning event corresponding to $\hat{\beta}_{post,n}$ is 
\begin{align*}
    \hballn \in \hBpostn, \quad \hBpostn = \bigcup_{k=1}^{K} \{\hballn: \Apostk \hballn \leq \hcpostkn\},
\end{align*}
where $\{\hballn: \Apostk \hballn \leq \hcpostkn\}$ are disjoint polyhedra based on nonrandom matrices $\Apostk \in \R^{\dim(c_{k}) \times \dim(\ball)}$ and random vectors $\hcpostkn \in \R^{\dim(c_{k})}$. We assume the rows of $\Apostk$ are nonzero, i.e., $(\Apostk)_{j} \neq 0$ for all $(j,k)$. When convenient, we use
\begin{align*}
    A = 
    \begin{pmatrix}
    A_{1} \\
    \vdots \\
    A_{K} \\
    \end{pmatrix} \in \R^{\dim(c) \times \dim(\ball)}, \quad \quad 
    c = 
    \begin{pmatrix}
    c_{1} \\
    \vdots \\
    c_{K}
    \end{pmatrix} \in \R^{\dim(c)}.
\end{align*}
to denote stacked matrices and vectors. To account for conditioning, we consider 
\begin{align*}
    \hat{r}_{n} = \hballn - \hcoefn l_{post}'\hballn, \quad \hcoefn = \frac{\hSigmaalln l_{post}}{\hsigmapostn^{2}}.
\end{align*}
Throughout, we define the maximum over the empty set as $-\infty$ and the minimum over the empty set as $+\infty$. With that in mind, define the sets 
\begin{align*}
    \hZminuspostkn &= \max_{j \in \hat{J}_{k,n}^{-}} \frac{(\hcpostkn)_{j} - (\Apostk \hat{r}_{n})_{j}}{(\Apostk\hcoefn)_{j}}, & \hat{J}_{k,n}^{-} &= \curly{j: (\Apostk\hcoefn)_{j} < 0}, \\
    \hZpluspostkn &= \min_{j \in \hat{J}_{k,n}^{+}} \frac{(\hcpostkn)_{j} - (\Apostk \hat{r}_{n})_{j}}{(\Apostk \hcoefn)_{j}}, & \hat{J}_{k,n}^{+} &= \curly{j: (\Apostk\hcoefn)_{j} > 0}, \\ 
    \hZzeropostkn &= \min_{j \in \hat{J}_{k,n}^{0}} (\hcpostkn)_{j} - (\Apostk \hat{r}_{n})_{j}, & \hat{J}_{k,n}^{0} &= \curly{j: (\Apostk\hcoefn)_{j} = 0}.
\end{align*}
As in the main text, the event $\{\Apostk \hballn \leq \hcpostkn\}$ is equivalent to event $\{l_{post}'\hballn \in \htruncpostkn\}$, and $\{\hballn \in \hBpostn\}$ is equivalent to event $\{l_{post}'\hballn \in \htruncpostn\}$, where
\begin{equation*}
    \htruncpostn = \bigcup_{k=1}^{K} \htruncpostkn, \quad  
    \htruncpostkn =
    \begin{cases}
        [\hZminuspostkn, \hZpluspostkn], & \hZzeropostkn \geq 0, \\
        \hfil \varnothing, & \hZzeropostkn < 0.
    \end{cases}
\end{equation*}

\paragraph{Conditional Inference.} Let $z \mapsto F_{TN}(z;\mu,\sigma^{2},\trunc)$ denote the CDF of the $N(\mu, \sigma^{2})$ distribution truncated to a set $\trunc$. This truncated CDF takes the form
\begin{align*}
    F_{TN}(z;\mu,\sigma^{2},\trunc) =  \frac{\displaystyle \int_{-\infty}^{z} \phi\paren{\frac{t-\mu}{\sigma}}\1\curly{t \in \trunc} dt}{\displaystyle \int_{-\infty}^{\infty}\phi\paren{\frac{t-\mu}{\sigma}}\1\curly{t \in \trunc} dt}.
\end{align*}
The truncated normal distribution with known variance has strict monotone likelihood ratio in $\mu$, and hence $\mu \mapsto F_{TN}(z;\mu,\sigma^{2},\trunc)$ is strictly decreasing in $\mu$. Given $\hballn \in \hBpostn$, the conditionally quantile unbiased estimator (QUE) is the unique $\hquantn$ that solves
\begin{equation*}
     F_{TN}(l_{post}'\hballn;\hquantn,\hsigmapostn^{2},\htruncpostn ) = 1-\alpha.
\end{equation*}
The goal is to show that inference based on $\hquantn$ is uniformly asymptotically valid in the sense of Proposition \ref{d2:app:prop:valid} below.

\paragraph{Scaled Quantities.} In what follows, we consider scaled estimates $\tballn = \sqrt{n}\hballn$ and $\tSigmaalln = n\hSigmaalln$, which induce variance $\tsigmapostn^{2} = n\hsigmapostn^{2}$, deviation set $\tBpostn = \sqrt{n}\hBpostn$ such that
\begin{align*}
    \tballn \in \tBpostn, \quad \tBpostn = \bigcup_{k=1}^{K} \{\tballn: \Apostk \tballn \leq \tcpostkn\}, \quad \tcpostkn = \sqrt{n}\hcpostkn,
\end{align*}
residualization step 
\begin{align*}
    \tilde{r}_{n} = \tballn - \tcoefn l_{post}'\tballn = \sqrt{n}\hat{r}_{n}, \quad \tcoefn = \frac{\tSigmaalln l_{post}}{\tsigmapostn^{2}} = \hcoefn,
\end{align*}
and sets 
\begin{align*}
    \tZminuspostkn &= \max_{j \in \tilde{J}_{k,n}^{-}} \frac{(\tcpostkn)_{j} - (\Apostk \tilde{r}_{n})_{j}}{(\Apostk\tcoefn)_{j}} = \sqrt{n}\hZminuspostkn, & \tilde{J}_{k,n}^{-} &= \curly{j:(\Apostk\tcoefn)_{j} < 0}, \\
    \tZpluspostkn &= \min_{j \in \tilde{J}_{k,n}^{+}} \frac{(\tcpostkn)_{j} - (\Apostk \tilde{r}_{n})_{j}}{(\Apostk \tcoefn)_{j}} = \sqrt{n}\hZpluspostkn, & \tilde{J}_{k,n}^{+} &= \curly{j:(\Apostk\tcoefn)_{j} > 0}, \\ 
    \tZzeropostkn &= \min_{j \in \tilde{J}_{k,n}^{0}} (\tcpostkn)_{j} - (\Apostk \tilde{r}_{n})_{j} = \sqrt{n}\hZzeropostkn, & \tilde{J}_{k,n}^{0} &= \curly{j:(\Apostk\tcoefn)_{j} = 0},
\end{align*}
such that the event $\{\tballn \in \tBpostn\}$ is equivalent to event $\{l_{post}'\tballn \in \ttruncpostn\}$, where
\begin{equation*}
    \ttruncpostn = \bigcup_{k=1}^{K} \ttruncpostkn, \quad  
    \ttruncpostkn =
    \begin{cases}
        [\tZminuspostkn, \tZpluspostkn], & \tZzeropostkn \geq 0, \\
        \hfil \varnothing, & \tZzeropostkn < 0.
    \end{cases}
\end{equation*}
Note that $\ttruncpostkn = \sqrt{n}\htruncpostkn$ so that $\ttruncpostn = \sqrt{n}\htruncpostn$.

\subsection{Assumptions}\label{d2:app:sec:assumptions}
\setcounter{assumption}{0}
For ease of reference, we restate Assumptions \ref{d2:app:ass:BL}-\ref{d2:app:ass:consistent.cutoff}.

Let $BL_{1}$ denote the set of real-valued functions that are bounded above in absolute value by one and have Lipschitz constant bounded above by one. Furthermore, given matrix $\Sigmaall$, let $\lambda_{min}(\Sigmaall)$ and $\lambda_{max}(\Sigmaall)$ denote its minimum and maximum eigenvalues.
\begin{assumption}\label{d2:app:ass:BL}
There exist $(\ballP, \SigmaallP)$ such that for $\tballPn = \sqrt{n}\ballP$ and $\xi_{P} \sim N(0, \SigmaallP)$,
\begin{align*}
    \lim_{n \to \infty} \sup_{P \in \cPn} \sup_{f \in BL_{1}} \abs{\E_{P}[
    f(\tballn - \tballPn)] - \E[f(\xi_{P})]}.
\end{align*}
Moreover, there exists finite $\blam > 0$ such that $1/\blam \leq \lambda_{min}(\SigmaallP) \leq \lambda_{max}(\SigmaallP) \leq \blam$ for all $P$.
\end{assumption}

Consider the scaled covariance estimator $\tSigmaalln = n\hSigmaalln$
\begin{assumption}\label{d2:app:ass:consistent.variance}
For each $\e > 0$,
\begin{align*}
    \lim_{n \to \infty} \sup_{P \in \cPn} \P[P]{\norm{\tSigmaalln - \SigmaallP} > \e} = 0.
\end{align*}
\end{assumption}

Consider $\tcpostkn = \sqrt{n}\hcpostkn$ for $\{\tballn: \Apostk \tballn \leq \tcpostkn\} = \sqrt{n}\{\hballn: \Apostk \hballn \leq \hcpostkn\}$.
\begin{assumption}\label{d2:app:ass:consistent.cutoff}
For each $k=1,\ldots, K$, there exist $\cpostkP$ such that for each $\e > 0$,
\begin{align*}
    \lim_{n \to \infty} \sup_{P \in \cPn} \P[P]{\norm{\tcpostkn - \cpostkP} > \e} = 0.
\end{align*}
Moreover, there exists finite $\blam_{c} > 0$ such that $\norm{\cpostkP} \leq \blam_{c}$ for all $P$ and $k$.
\end{assumption}

\subsection{Results}\label{d2:app:sec:results}
Given $\hballn \in \hBpostn$, the QUE $\hquantn$ uniquely solves
\begin{equation*}
     F_{TN}(l_{post}'\hballn;\hquantn,\hsigmapostn^{2},\htruncpostn) = 1-\alpha.
\end{equation*}
Our goal is to show that conditional inference based on $\hquantn$ is uniformly asymptotically valid, in the following sense. 
\begin{proposition}\label{d2:app:prop:valid}
Under Assumptions \ref{d2:app:ass:BL}-\ref{d2:app:ass:consistent.cutoff}, we have
\begin{align*}
    \lim_{n \to \infty} \sup_{P \in \cPn} \bigabs{\P[P]{\hquantn \geq l_{post}'\ballP \mid \hballn \in \hBpostn} - \alpha}\P[P]{\hballn \in \hBpostn} = 0.
\end{align*}
\end{proposition}

To prove Proposition \ref{d2:app:prop:valid}, we state and prove a series of lemmas. 

\paragraph{Sufficient Conditions Based on Subsequences.} The first lemma shows that, to verify the uniform asymptotic validity in Proposition \ref{d2:app:prop:valid}, it suffices to verify a type of pointwise asymptotic validity for every subsequence that satisfies a set of conditions. Let $\tquantn = \sqrt{n}\hquantn$.

\begin{lemma}\label{d2:app:lemma:sufficient}
Suppose there exists finite $\blam, \blam_{c} > 0$ such that $1/\blam \leq \lambda_{min}(\SigmaallP) \leq \lambda_{max}(\SigmaallP) \leq \blam$ and $\norm{\cpostkP} \leq \blam_{c}$, for all $P$ and $k = 1, \ldots, K$. Then, to show that
\begin{align*}
    \limsup_{n \to \infty} \sup_{P \in \cPn} \bigabs{\P[P]{\tquantn \geq l_{post}'\tballPn \mid \tballn \in \tBpostn} - \alpha}\P[P]{\tballn \in \tBpostn} = 0,
\end{align*}
it suffices to show that for all subsequences $\curly{n_{s}} \subseteq \curly{n}$ and $\curly{P_{n_{s}}} \in \times_{s=1}^{\infty}\cPn[n_{s}]$ with
\begin{enumerate}[label=(\roman*)]
    \item $\P[P_{n_{s}}]{\tilde{\ball}_{n_{s}} \in \widetilde{\mathcal{B}}_{post, n_{s}}} \to p^{*} \in (0,1]$; 
    \item $\Apostk\tballPns \to \apostk^{*} \in [-\infty, +\infty]^{\dim(c_{k})}$ for all $k$;
    \item $\Sigmaall(\Pns) \to \Sigmaall^{*} \in \mathcal{M} = \curly{\Sigmaall: 1/\blam \leq \lambda_{min}(\Sigmaall) \leq \lambda_{max}(\Sigmaall) \leq \blam}$; and
    \item $\cpostkPns \to \cpostk^{*} \in [-\blam_{c}, \blam_{c}]^{\dim(c_{k})}$ for all $k$;
\end{enumerate}
we have that
\begin{align*}
    \lim_{s \to \infty} \P[\Pns]{\tquantns \geq l_{post}'\tballPns \mid \tballns \in \tBpostns} = \alpha.
\end{align*}
\end{lemma}

\begin{proof}[Proof of Lemma \ref{d2:app:lemma:sufficient}]
$ $ \newline
Note $|x| = \max\{-x,x\}$. Thus, to show that
\begin{align*}
    \limsup_{n \to \infty} \sup_{P \in \cPn} \bigabs{\P[P]{\tquantn \geq l_{post}'\tballPn \mid \tballn \in \tBpostn} - \alpha}\P[P]{\tballn \in \tBpostn} = 0,
\end{align*}
it suffices to show both that
\begin{align}\label{d2:app:eq:liminf.condition} 
    \liminf_{n \to \infty} \inf_{P \in \cPn} \paren{\P[P]{\tquantn \geq l_{post}'\tballPn \mid \tballn \in \tBpostn} - \alpha}\P[P]{\tballn \in \tBpostn} \geq 0
\end{align}
and that
\begin{align}\label{d2:app:eq:limsup.condition}
    \limsup_{n \to \infty} \sup_{P \in \cPn} \paren{\P[P]{\tquantn \geq l_{post}'\tballPn \mid \tballn \in \tBpostn} - \alpha}\P[P]{\tballn \in \tBpostn} \leq 0.
\end{align}
We claim that \eqref{d2:app:eq:liminf.condition} holds if for all subsequences satisfying conditions (i)-(iv) of this lemma, 
\begin{align}\label{d2:app:eq:subliminf.condition}
    \liminf_{s \to \infty} \P[\Pns]{\tquantns \geq l_{post}'\tballPns \mid \tballns \in \tBpostns} \geq \alpha.
\end{align}
Likewise, we claim that \eqref{d2:app:eq:limsup.condition} holds if for all subsequences satisfying conditions (i)-(iv) of this lemma, 
\begin{align}\label{d2:app:eq:sublimsup.condition}
    \limsup_{s \to \infty} \P[\Pns]{\tquantns \geq l_{post}'\tballPns \mid \tballns \in \tBpostns} \leq \alpha.
\end{align}
Suppose these two claims are true. Then, to show that both \eqref{d2:app:eq:liminf.condition} and \eqref{d2:app:eq:limsup.condition} hold, it suffices to show that for all subsequences satisfying conditions (i)-(iv) of this lemma,
\begin{align*}
    \lim_{s \to \infty} \P[\Pns]{\tquantns \geq l_{post}'\tballPns \mid \tballns \in \tBpostns} = \alpha.
\end{align*}
In particular, the above equation implies \eqref{d2:app:eq:subliminf.condition} and \eqref{d2:app:eq:sublimsup.condition}. We provide the explicit argument for the first claim that \eqref{d2:app:eq:subliminf.condition} implies \eqref{d2:app:eq:liminf.condition}.

Towards a contradiction, suppose that \eqref{d2:app:eq:subliminf.condition} holds for all subsequences satisfying conditions (i)-(iv) of this lemma, but that \eqref{d2:app:eq:liminf.condition} fails. If \eqref{d2:app:eq:liminf.condition} fails, then there exists $\e>0$ such that
\begin{align*}
    \liminf_{n \to \infty} \inf_{P \in \cPn} \paren{\P[P]{\tquantn \geq l_{post}'\tballPn \mid \tballn \in \tBpostn} - \alpha}\P[P]{\tballn \in \tBpostn} = L < -3\e.
\end{align*}
By properties of the limit inferior, given $\e>0$, there exists subsequence $\{n_{q}\} \subseteq \{n\}$ such that
\begin{align*}
    \inf_{P \in \mathcal{P}_{n_{q}}} \paren{\P[P]{\tilde{\mu}_{\alpha, n_{q}}^{*} \geq l_{post}'\tilde{\ball}_{n_{q}}(P) \mid \tilde{\ball}_{n_{q}} \in \widetilde{\mathcal{B}}_{post, n_{q}}} - \alpha}\P[P]{\tilde{\ball}_{n_{q}} \in \widetilde{\mathcal{B}}_{post, n_{q}}} = L_{q} < L + \e, \quad \forall q.
\end{align*}
By properties of the infimum, given $\e>0$ and $q$, there exists $P_{n_{q}} \in \mathcal{P}_{n_{q}}$ such that
\begin{align*}
    \paren{\P[P_{n_{q}}]{\tilde{\mu}_{\alpha, n_{q}}^{*} \geq l_{post}'\tilde{\ball}_{n_{q}}(P_{n_{q}}) \mid \tilde{\ball}_{n_{q}} \in \widetilde{\mathcal{B}}_{post, n_{q}}} - \alpha}\P[P_{n_{q}}]{\tilde{\ball}_{n_{q}} \in \widetilde{\mathcal{B}}_{post, n_{q}}} < L_{q} + \e.
\end{align*}
Since $L < -3\e$ by hypothesis and $L_{q} < L + \e$, the above implies 
\begin{align}\label{d2:app:eq:condition3fodder}
    \limsup_{q \to \infty}\paren{\P[P_{n_{q}}]{\tilde{\mu}_{\alpha, n_{q}}^{*} \geq l_{post}'\tilde{\ball}_{n_{q}}(P_{n_{q}}) \mid \tilde{\ball}_{n_{q}} \in \widetilde{\mathcal{B}}_{post, n_{q}}} - \alpha}\P[P_{n_{q}}]{\tilde{\ball}_{n_{q}} \in \widetilde{\mathcal{B}}_{post, n_{q}}} < -\e.
\end{align}
We will show there must exist a subsequence $\curly{n_{s}} \subseteq \curly{n_{q}}$ that satisfies conditions (i)-(iv) of this lemma. Note that, if we can find such a subsequence, then \eqref{d2:app:eq:subliminf.condition} would yield
\begin{align*}
    0 \leq \liminf_{s \to \infty} \paren{\P[\Pns]{\tquantns \geq l_{post}'\tballPns \mid \tballns \in \tBpostns}-\alpha},
\end{align*}
which implies
\begin{align}\label{d2:app:eq:contradiction}
    0 \leq \liminf_{s \to \infty} \paren{\P[\Pns]{\tquantns \geq l_{post}'\tballPns \mid \tballns \in \tBpostns}-\alpha}\P[P_{n_{s}}]{\tilde{\ball}_{n_{s}} \in \widetilde{\mathcal{B}}_{post, n_{s}}},
\end{align}
which yields a contradiction, since equation \eqref{d2:app:eq:condition3fodder} implies that the above sequence is eventually bounded above by a negative constant. 

\paragraph{Condition (i).} Equation \eqref{d2:app:eq:condition3fodder} implies 
\begin{align*}
    \liminf_{q \to \infty}\paren{1 - \frac{\P[P_{n_{q}}]{\tilde{\mu}_{\alpha, n_{q}}^{*} \geq l_{post}'\tilde{\ball}_{n_{q}}(P_{n_{q}}) \mid \tilde{\ball}_{n_{q}} \in \widetilde{\mathcal{B}}_{post, n_{q}}}}{\alpha}}
    \P[P_{n_{q}}]{\tilde{\ball}_{n_{q}} \in \widetilde{\mathcal{B}}_{post, n_{q}}} = L' > \frac{\e}{\alpha}.
\end{align*}
By properties of the limit inferior, given $\e/(2\alpha) >0$, there exists $q_{0}$ such that
\begin{align*}
    \paren{1 - \frac{\P[P_{n_{q}}]{\tilde{\mu}_{\alpha, n_{q}}^{*} \geq l_{post}'\tilde{\ball}_{n_{q}}(P_{n_{q}}) \mid \tilde{\ball}_{n_{q}} \in \widetilde{\mathcal{B}}_{post, n_{q}}}}{\alpha}}
    \P[P_{n_{q}}]{\tilde{\ball}_{n_{q}} \in \widetilde{\mathcal{B}}_{post, n_{q}}} > L' - \frac{\e}{2\alpha}, \quad \forall q \geq q_{0}. 
\end{align*}
For such $q$, the term in parentheses on the LHS must be positive, and hence in $(0, 1]$. Thus,
\begin{align*}
    \P[P_{n_{q}}]{\tilde{\ball}_{n_{q}} \in \widetilde{\mathcal{B}}_{post, n_{q}}} > L' - \frac{\e}{2\alpha} > \frac{\e}{2\alpha}, \quad \forall q \geq q_{0}.
\end{align*}
The set $[\e/(2\alpha),1]$ is compact and contains $\P[P_{n_{q}}]{\tilde{\ball}_{n_{q}} \in \widetilde{\mathcal{B}}_{post, n_{q}}}$ for all $q \geq q_{0}$. Thus, there must exist a subsequence $\{n_{r}\} \subseteq \{n_{q}: q \geq q_{0}\}$ such that 
\begin{align*}
    \P[P_{n_{r}}]{\tilde{\ball}_{n_{r}} \in \widetilde{\mathcal{B}}_{post, n_{r}}} \to p^{*} \in [\e/(2\alpha),1].
\end{align*}

\paragraph{Condition (ii).} Consider $a = (a_{1}', \ldots, a_{K}')' \in [-\infty, +\infty]^{\dim(c)}$. Under the standard order topology on $[-\infty, +\infty]$, the induced product topology on $[-\infty, +\infty]^{\dim(c)}$ is generated by the metric 
\begin{align*}
    d(a^{1}, a^{2}) = \max_{j}|\Phi(a^{1}_{j}) - \Phi(a^{2}_{j})|, \quad a^{1}, a^{2} \in [-\infty, \infty]^{\dim(c)},
\end{align*}
where $z \mapsto \Phi(z)$ denotes the standard normal CDF with $\Phi(-\infty) = 0$ and $\Phi(+\infty) = 1$. The set $[-\infty, +\infty]^{\dim(c)}$ is compact under $d$ and contains $A_{post}\tballn[n_{r}](\Pn[n_{r}])$ for all $r$. Thus, there exists a further subsequence $\{n_{t}\} \subseteq \{n_{r}\}$ and limit $\apost^{*}$ such that
\begin{align*}
    d(A_{post}\tballn[n_{t}](\Pn[n_{t}]), a_{post}^{*}) \to 0, \quad a_{post}^{*} \in [-\infty, +\infty]^{\dim(c)}.
\end{align*}
In particular, $\Apostk\tballn[n_{t}](\Pn[n_{t}]) \to \apostk^{*} \in [-\infty, +\infty]^{\dim(c_{k})}$ for all $k$.

\paragraph{Condition (iii).} By assumption, $\SigmaallP \in \mathcal{M}$ for all $P$. Thus, $\Sigmaall(P_{n_{t}}) \in \mathcal{M}$ for all $t$. By compactness of $\mathcal{M}$ in the spectral norm, there then must exist a subsequence $\curly{n_{u}} \subseteq \curly{n_{t}}$ such that $\Sigmaall(P_{n_{u}}) \to \Sigmaall^{*} \in \mathcal{M}$.

\paragraph{Condition (iv).}
By assumption, we have $\norm{\cpostkP} \leq \blam_{c}$, for all $P$ and $k$. Thus, $\cpost(\Pn[n_{u}]) \in [-\blam_{c}, \blam_{c}]^{\dim(c)}$ for all $u$. The set $[-\blam_{c}, \blam_{c}]^{\dim(c)}$ is compact, so there must exist a subsequence $\{n_{s}\} \subseteq \{n_{u}\}$ and limit $\cpost^{*}$ such that
\begin{align*}
    \cpost(\Pn[n_{s}]) \to \cpost^{*} \in [-\blam_{c}, \blam_{c}]^{\dim(c)}.
\end{align*}
In particular, $\cpostk(\Pn[n_{s}]) \to \cpostk^{*} \in [-\blam_{c}, \blam_{c}]^{\dim(c_{k})}$ for all $k$.

\paragraph{Contradiction.} We have obtained a subsequence $\{n_{s}\} \subseteq \{n_{q}\}$ that satisfies conditions (i)-(iv) of this lemma. By hypothesis, equation \eqref{d2:app:eq:subliminf.condition} holds for this subsequence:
\begin{align*}
    \liminf_{s \to \infty} \P[\Pns]{\tquantns \geq l_{post}'\tballPns \mid \tballns \in \tBpostns} \geq \alpha.
\end{align*}
But this leads to a contradiction of equation \eqref{d2:app:eq:condition3fodder}, as we described under equation \eqref{d2:app:eq:contradiction}.
\end{proof}

\paragraph{Representing the Coverage Event.} To use Lemma \ref{d2:app:lemma:sufficient} to prove Proposition \ref{d2:app:prop:valid}, we first derive a representation of the coverage event $\{\tquantn \geq l_{post}'\tballPn\}$ in terms of appropriately normalized (centered and scaled) estimators. In what follows, recall that $\tballn^{*} = \tballn - \tballPn = \sqrt{n}(\hballn - \ballP)$. Moreover, define the normalized residual 
\begin{align*}
    \tilde{r}_{n}^{*} = \tilde{r}_{n} + \tcoefn l_{post}'\tballPn,
\end{align*}
the normalized truncation quantities 
\begin{align*}
    \stZminuspostkn &= \max_{j \in \tilde{J}_{k,n}^{-}} \frac{(\tcpostkn)_{j} - (\Apostk \tilde{r}_{n}^{*})_{j}}{(\Apostk\tcoefn)_{j}} = \tZminuspostkn - l_{post}'\tballPn, \\
    \stZpluspostkn &= \min_{j \in \tilde{J}_{k,n}^{+}} \frac{(\tcpostkn)_{j} - (\Apostk \tilde{r}_{n}^{*})_{j}}{(\Apostk \tcoefn)_{j}} = \tZpluspostkn - l_{post}'\tballPn, \\ 
    \stZzeropostkn &= \min_{j \in \tilde{J}_{k,n}^{0}} (\tcpostkn)_{j} - (\Apostk \tilde{r}_{n}^{*})_{j} = \tZzeropostkn,
\end{align*}
and the normalized truncation sets
\begin{equation*}
    \ttruncpostn^{*} = \bigcup_{k=1}^{K} \ttruncpostkn^{*}, \quad  
    \ttruncpostkn^{*} =
    \begin{cases}
        [\stZminuspostkn, \stZpluspostkn], & \stZzeropostkn \geq 0 \\
        \hfil \varnothing, & \stZzeropostkn < 0
    \end{cases}.
\end{equation*}
Note that $\ttruncpostn^{*} = \ttruncpostn - l_{post}'\tballPn$.

\begin{lemma}\label{d2:app:lemma:normalized.CDF}
For all $n \geq 1$, we have
\begin{align*}
    \tquantn \geq l_{post}'\tballPn  \iff 1-\alpha \leq F_{TN}(l_{post}'\tballn^{*};0,\tsigmapostn^{2},\ttruncpostn^{*}).
\end{align*}
\end{lemma}

\begin{proof}[Proof of Lemma \ref{d2:app:lemma:normalized.CDF}]
$ $ \newline
Noting that $\tquantn = \sqrt{n}\hquantn$, we obtain
\begin{align*}
    F_{TN}(l_{post}'\tballn;\tquantn,\tsigmapostn^{2},\ttruncpostn) &= \frac{\displaystyle \int_{-\infty}^{l_{post}'\tballn} \phi\paren{\frac{t-\tquantn}{\tsigmapostn}}\1\curly{t \in \ttruncpostn} dt}{\displaystyle \int_{-\infty}^{\infty}\phi\paren{\frac{t-\tquantn}{\tsigmapostn}}\1\curly{t \in \ttruncpostn} dt} \\
    &= \frac{\displaystyle \int_{-\infty}^{l_{post}'\sqrt{n}\hballn} \phi\paren{\frac{t/\sqrt{n}-\hquantn}{\hsigmapostn}}\1\curly{t/\sqrt{n} \in \htruncpostn} dt}{\displaystyle \int_{-\infty}^{\infty}\phi\paren{\frac{t/\sqrt{n}-\hquantn}{\hsigmapostn}}\1\curly{t/\sqrt{n} \in \htruncpostn} dt} \\
    &= \frac{\displaystyle \int_{-\infty}^{l_{post}'\hballn} \phi\paren{\frac{u-\hquantn}{\hsigmapostn}}\1\curly{u \in \htruncpostn} du}{\displaystyle \int_{-\infty}^{\infty}\phi\paren{\frac{u-\hquantn}{\hsigmapostn}}\1\curly{u \in \htruncpostn} du} \\
    &= F_{TN}(l_{post}'\hballn;\hquantn,\hsigmapostn^{2},\htruncpostn),
\end{align*}
where the third equality follows from substituting $u = t/\sqrt{n}$ and $du = dt/\sqrt{n}$. Since $\hquantn$ uniquely solves
$F_{TN}(l_{post}'\hballn;\hquantn,\hsigmapostn^{2},\htruncpostn) = 1-\alpha$ given $\hballn \in \hBpostn$, then given $\tballn \in \tBpostn$ we have that $\tquantn$ uniquely solves
\begin{align*}
     F_{TN}(l_{post}'\tballn;\tquantn,\tsigmapostn^{2},\ttruncpostn) = 1-\alpha.
\end{align*}
Since the truncated normal CDF is strictly decreasing in $\mu$, the above implies
\begin{align*}
     \tquantn \geq l_{post}'\tballPn \iff 1-\alpha \leq F_{TN}(l_{post}'\tballn;l_{post}'\tballPn,\tsigmapostn^{2},\ttruncpostn).
\end{align*}
Observe that
\begin{align*}
    F_{TN}(l_{post}'\tballn;l_{post}'\tballPn,\tsigmapostn^{2},\ttruncpostn)
    &= \frac{\displaystyle \int_{-\infty}^{l_{post}'\tballn} \phi\paren{\frac{t-l_{post}'\tballPn}{\tsigmapostn}}\1\curly{t \in \ttruncpostn} dt}{\displaystyle \int_{-\infty}^{\infty}\phi\paren{\frac{t-l_{post}'\tballPn}{\tsigmapostn}}\1\curly{t \in \ttruncpostn} dt}  \\
    &= \frac{\displaystyle \int_{-\infty}^{l_{post}'\tballn^{*}} \phi\paren{\frac{u}{\tsigmapostn}}\1\curly{u \in \ttruncpostn - l_{post}'\tballPn} du}{\displaystyle \int_{-\infty}^{\infty}\phi\paren{\frac{u}{\tsigmapostn}}\1\curly{u \in \ttruncpostn - l_{post}'\tballPn} du} \\
    &= F_{TN}(l_{post}'\tballn^{*};0,\tsigmapostn^{2},\ttruncpostn - l_{post}'\tballPn),
\end{align*}
where the second equality follows from substituting $u = t-l_{post}'\tballPn$ and $du = dt$. Thus,
\begin{align*}
     \tquantn \geq l_{post}'\tballPn \iff 1-\alpha \leq F_{TN}(l_{post}'\tballn^{*};0,\tsigmapostn^{2},\ttruncpostn - l_{post}'\tballPn),
\end{align*}
as desired.
\end{proof}

Lemma \ref{d2:app:lemma:normalized.CDF} provides a representation of the coverage event $\{\tquantn \geq l_{post}'\tballPn\}$ in terms of the quantity $F_{TN}(l_{post}'\tballn[n]^{*};0,\tsigmapostn[n]^{2},\ttruncpostn[n]^{*})$, which depends on normalized variables 
\begin{align*}
    (\tballn^{*}, \tSigmaalln, \ttruncpostn[n]^{*}) = (\sqrt{n}(\hballn - \ball(P)), n\hSigmaalln, \sqrt{n}(\htruncpostn - l_{post}'\ball(P))).
\end{align*}
We now characterize the limiting/asymptotic behavior of the RHS as $n \to \infty$. This will allow us to leverage Lemma \ref{d2:app:lemma:sufficient} to prove Proposition \ref{d2:app:prop:valid}.

\paragraph{Limit Problem.} For subsequences $\curly{n_{s}} \subseteq \curly{n}$ and $\curly{P_{n_{s}}} \in \times_{s=1}^{\infty}\cPn[n_{s}]$ satisfying the conditions of Lemma \ref{d2:app:lemma:sufficient}, we show that the corresponding sequence of finite sample problems, where we observe samples of size $n_{s}$ from distributions $P_{n_{s}}$ and employ conditional inference procedures, is well approximated by a limit problem where we observe a single normal vector with known covariance matrix and employ the analogous conditional inference procedures. In considering this limit problem, we will be able to understand the asymptotic behavior of
\begin{align*}
    (\tballn[n_{s}]^{*}, \tSigmaalln[n_{s}], \ttruncpostn[n_{s}]^{*}) = (\sqrt{n_{s}}(\hballn[n_{s}] - \ball(\Pns)), n_{s}\hSigmaalln[n_{s}], \sqrt{n_{s}}(\htruncpostn[n_{s}] - l_{post}'\ball(\Pns))),
\end{align*}
and hence the asymptotic behavior of $F_{TN}(l_{post}'\tballn[n_{s}]^{*};0,\tsigmapostn[n_{s}]^{2},\ttruncpostn[n_{s}]^{*})$.

To derive the limit problem, note that under Assumptions \ref{d2:app:ass:BL}-\ref{d2:app:ass:consistent.cutoff}, and for subsequences that satisfy the conditions of Lemma \ref{d2:app:lemma:sufficient}, we obtain the following convergence results, which we will use in the proof of Proposition \ref{d2:app:prop:valid}:
\begin{align*}
    \tballns^{*} &\to[d] N(0, \Sigmaall^{*}), & \Apostk\tballPns &\to \apostk^{*}, \\
    \tSigmaalln[n_{s}] &\to[p] \Sigmaall^{*}, & \tcpostkns &\to[p] \cpostk^{*}, \\
    \Sigmaall(\Pns) &\to \Sigmaall^{*}, & \cpostkPns &\to \cpostk^{*},
\end{align*}
where $\Sigmaall^{*} \in \mathcal{M}$, $\apostk^{*} \in [-\infty, +\infty]^{\dim(c_{k})}$, $\cpostk^{*} \in [-\blam_{c}, \blam_{c}]^{\dim(c_{k})}$, and $k=1,\ldots, K$. Thus, stacking the vectors and matrices over $k$, we obtain
\begin{align}\label{d2:app:eq:stacked.CLT}
    \begin{pmatrix}
        \tballns^{*} \\
        \tSigmaalln[n_{s}] \\
        A_{post}\tballPns \\
        \tilde{c}_{post,n_{s}}
    \end{pmatrix} 
    \to[d]
    \begin{pmatrix}
        \xi^{*} \\
        \Sigmaall^{*} \\
        a_{post}^{*} \\
        c_{post}^{*}
    \end{pmatrix}, \quad
    \xi^{*} \sim N(0, \Sigmaall^{*}).
\end{align}
Letting $(\Apostk)_{j}$ denote rows of $\Apostk$, the continuous mapping theorem (CMT) implies
\begin{align}\label{d2:app:eq:CMT.base}
\begin{split}
    \max_{k}\1\curly{\Apostk \tballns \leq \tilde{c}_{post,k,n_{s}}} &= \max_{k}\1\curly{\Apostk \tballns^{*} \leq \tilde{c}_{post,k,n_{s}} - \Apostk \tballPns} \\
    &= \max_{k}\min_{j} \1\curly{(\Apostk)_{j} \tballns^{*} \leq (\tilde{c}_{post,k,n_{s}})_{j} - (\Apostk)_{j} \tballPns} \\
    &\to[d] \max_{k}\min_{j} \1\curly{(\Apostk)_{j} \xi^{*} \leq (\cpostk^{*})_{j} - (\apostk^{*})_{j}} \\
    &= \max_{k}\1\curly{\Apostk \xi^{*} \leq \cpostk^{*} - \apostk^{*}},
\end{split}
\end{align}
where we use the continuity of $(\Apostk)_{j}\xi^{*} \sim N(0, (\Apostk)_{j}\Sigmaall^{*}(\Apostk)_{j}')$, which we obtain from $(\Apostk)_{j} \neq 0$ for all $(k,j)$. Consider the union of polyhedra
\begin{align*}
    \mathcal{B}_{post}^{*} = \bigcup_{k=1}^{K} \{\xi^{*}: \Apostk \xi^{*} \leq \cpostk^{*} - \apostk^{*}\}.
\end{align*}
By equation \eqref{d2:app:eq:CMT.base} and the conditions of Lemma \ref{d2:app:lemma:sufficient}, we obtain
\begin{align}\label{d2:app:eq:probability.convergence}
\begin{split}
    \P[\Pns]{\tballns \in \tBpostns} &= \EP[\Pns]{\max_{k}\1\curly{\Apostk \tballns \leq \tilde{c}_{post,k,n_{s}}}} \\
    &\to \EP[]{\max_{k}\1\curly{\Apostk \xi^{*} \leq \cpostk^{*} - \apostk^{*}}}\\
    &= \P[]{\xi^{*} \in \mathcal{B}_{post}^{*}} \\
    &= p^{*} > 0.
\end{split}
\end{align}
For the limit problem residual step, define
\begin{align*}
    r^{*} = \xi^{*} - \gamma^{*}l_{post}'\xi^{*}, \quad \gamma^{*} = \frac{\Sigmaall^{*}l_{post}}{\sigma_{post}^{*2}}, \quad \sigma_{post}^{*2} = l_{post}'\Sigmaall^{*}l_{post}. 
\end{align*}
For conditioning, define truncation quantities
\begin{align*}
    \Zminuspostk &= \max_{j \in J_{k}^{-}} \frac{(\cpostk^{*} - \apostk^{*})_{j} - (\Apostk r^{*})_{j}}{(\Apostk \gamma^{*})_{j}}, & J_{k}^{-} &= \{j: (\Apostk \gamma^{*})_{j} < 0\}, \\
    \Zpluspostk &= \min_{j \in J_{k}^{+}} \frac{(\cpostk^{*} - \apostk^{*})_{j} - (\Apostk r^{*})_{j}}{(\Apostk \gamma^{*})_{j}}, & J_{k}^{+} &= \{j: (\Apostk \gamma^{*})_{j} > 0\}, \\
    \Zzeropostk &= \min_{j \in J_{k}^{0}} (\cpostk^{*} - \apostk^{*})_{j} - (\Apostk r^{*})_{j} & J_{k}^{0} &= \{j: (\Apostk \gamma^{*})_{j} = 0\},
\end{align*}
such that
\begin{align*}
    \truncpost^{*} = \bigcup_{k=1}^{K}\truncpostk^{*}, \quad \truncpostk^{*} = 
    \begin{cases}
        [\Zminuspostk, \Zpluspostk], & \Zzeropostk \geq 0, \\
        \hfil \varnothing, & \Zzeropostk < 0.
    \end{cases}
\end{align*}
In summary, the limit problem consists of a normal draw $\xi^{*} \sim N(0, \Sigmaall^{*})$ with (unknown) mean zero and (known) covariance matrix $\Sigmaall^{*} \in \mathcal{M}$, and considers inference for post estimates $l_{post}'\xi^{*}$ conditional on $\xi^{*}$ falling into a union of polyhedra:
\begin{align*}
    \xi^{*} \in \mathcal{B}_{post}^{*}, \quad \mathcal{B}_{post}^{*} = \bigcup_{k=1}^{K} \{\xi^{*}: \Apostk \xi^{*} \leq \cpostk^{*} - \apostk^{*}\},
\end{align*}
where the conditioning event $\{\xi^{*} \in \mathcal{B}_{post}^{*}\}$ is equivalent to event $\{l_{post}'\xi^{*} \in \truncpost^{*}\}$. This conditioning event has positive probability, as noted in \eqref{d2:app:eq:probability.convergence}. 

\paragraph{Representing Inclusion in Truncation Set.} As a last step before proving the main result, it is useful to represent membership $\1\{t \in \truncpost^{*}\}$ in the truncation set for a given $t \in \R$ in terms of the limiting variables $(\xi^{*}, \Sigmaall^{*}, a_{post}^{*}, c_{post}^{*})$. To do so, note that if $t \in \truncpostk^{*}$ for some $k$, then this is equivalent to (i) $\Zzeropostk \geq 0$ and (ii)
\begin{align*}
    \max_{j \in J_{k}^{-}} \frac{(\cpostk^{*} - \apostk^{*})_{j} - (\Apostk r^{*})_{j}}{(\Apostk\gamma^{*})_{j}} \leq t \leq \min_{j \in J_{k}^{+}} \frac{(\cpostk^{*} - \apostk^{*})_{j} - (\Apostk r^{*})_{j}}{(\Apostk \gamma^{*})_{j}}.
\end{align*}
Condition (i) is equivalent to
\begin{align*}
    (\Apostk\gamma^{*})_{j}t + (\Apostk r^{*})_{j} &\leq (\cpostk^{*} - \apostk^{*})_{j}, \quad \forall j \in J_{k}^{0}.
\end{align*}
Condition (ii) is equivalent to the combination of these two statements:
\begin{align*}
    \frac{(\cpostk^{*} - \apostk^{*})_{j} - (\Apostk r^{*})_{j}}{(\Apostk\gamma^{*})_{j}} &\leq t, \quad \forall j \in J_{k}^{-}, \\
    \frac{(\cpostk^{*} - \apostk^{*})_{j} - (\Apostk r^{*})_{j}}{(\Apostk\gamma^{*})_{j}} &\geq t, \quad \forall j \in J_{k}^{+}.
\end{align*}
Together, the equivalent expressions for conditions (i) and (ii) show that, if $t \in \truncpostk^{*}$ for some $k$, then we equivalently have $(\Apostk\gamma^{*})_{j} t + (\Apostk r^{*})_{j} \leq (\cpostk^{*} - \apostk^{*})_{j}$ for all $j$. In other words,
\begin{align*}
    \1\curly{t \in \truncpostk^{*}} = \1\curly{\Apostk(\gamma^{*} t + r^{*}) \leq \cpostk^{*} - \apostk^{*}}.
\end{align*}
Since $t \in \truncpost^{*}$ is equivalent to the existence of $k$ such that $t \in \truncpostk^{*}$, we obtain
\begin{align}\label{d2:app:eq:membership.Z.to.B}
    \1\curly{t \in \truncpost^{*}} = \max_{k}\1\curly{\Apostk(\gamma^{*} t + r^{*}) \leq \cpostk^{*} - \apostk^{*}}.
\end{align}

\paragraph{Proving the Main Result.} We are now prepared to prove Proposition \ref{d2:app:prop:valid}.

\begin{proof}[Proof of Proposition \ref{d2:app:prop:valid}]
$ $ \newline
Under Assumptions \ref{d2:app:ass:BL}-\ref{d2:app:ass:consistent.cutoff}, to prove that
\begin{align*}
    \lim_{n \to \infty} \sup_{P \in \cPn} \bigabs{\P[P]{\hquantn \geq l_{post}'\ballP \mid \hballn \in \hBpostn} - \alpha}\P[P]{\hballn \in \hBpostn} = 0,
\end{align*}
it suffices to show that for all subsequences $\curly{n_{s}} \subseteq \curly{n}$ and $\curly{P_{n_{s}}} \in \times_{s=1}^{\infty}\cPn[n_{s}]$ that satisfy the conditions of Lemma \ref{d2:app:lemma:sufficient}, we have
\begin{align*}
    \lim_{s \to \infty} \P[\Pns]{\tquantns \geq l_{post}'\tballPns \mid \tballns \in \tBpostns} = \alpha.
\end{align*}
By Lemma \ref{d2:app:lemma:normalized.CDF}, the above is equivalent to 
\begin{align*}
    \lim_{s \to \infty} \P[\Pns]{F_{TN}(l_{post}'\tballn[n_{s}]^{*};0,\tsigmapostn[n_{s}]^{2},\ttruncpostn[n_{s}]^{*}) \geq 1-\alpha \mid \tballns \in \tBpostns} = \alpha.
\end{align*}
To verify this convergence, note that (i) $F_{TN}(l_{post}'\xi^{*};0,\sigma_{post}^{*2},\truncpost^{*})|\xi^{*} \in \mathcal{B}_{post}^{*} \sim U(0,1)$, as noted in the main text and (ii) $\P[]{\xi^{*} \in \mathcal{B}_{post}^{*}} > 0$, as per equation \eqref{d2:app:eq:probability.convergence}. Thus, if we can show that
\begin{align*}
    &\frac{\EP[\Pns]{\1\curly{F_{TN}(l_{post}'\tballn[n_{s}]^{*};0,\tsigmapostn[n_{s}]^{2},\ttruncpostn[n_{s}]^{*}) \geq 1-\alpha}\1\curly{\tballns \in \tBpostns}}}{\P[\Pns]{\tballns \in \tBpostns}} \\
    &\to \frac{\EP[]{\1\curly{F_{TN}(l_{post}'\xi^{*};0,\sigma_{post}^{*2},\truncpost^{*}) \geq 1-\alpha}\1\curly{\xi^{*} \in \mathcal{B}_{post}^{*}}}}{\P[]{\xi^{*} \in \mathcal{B}_{post}^{*}}} \\
    &= \P[]{F_{TN}(l_{post}'\xi^{*};0,\sigma_{post}^{*2},\truncpost^{*}) \geq 1-\alpha \mid \xi^{*} \in \mathcal{B}_{post}^{*}} = \alpha,
\end{align*}
then we will have verified the desired convergence. 

Note that $\P[\Pns]{\tballns \in \tBpostns} \to \P[]{\xi^{*} \in \mathcal{B}_{post}^{*}} > 0$ by equation \eqref{d2:app:eq:probability.convergence}. Thus, if suffices to show that the LHS numerator converges to the RHS numerator.

To establish convergence of the numerator, it suffices to show
\begin{align}\label{d2:app:eq:ultimate.target.for.CMT}
\begin{split}
    &\1\curly{F_{TN}(l_{post}'\tballn[n_{s}]^{*};0,\tsigmapostn[n_{s}]^{2},\ttruncpostn[n_{s}]^{*})\1\curly{\tballns \in \tBpostns} \geq 1-\alpha}\1\curly{\tballns \in \tBpostns} \\
    &\to[d] \1\curly{F_{TN}(l_{post}'\xi^{*};0,\sigma_{post}^{*2},\truncpost^{*})\1\curly{\xi^{*} \in \mathcal{B}_{post}^{*}} \geq 1-\alpha}\1\curly{\xi^{*} \in \mathcal{B}_{post}^{*}}.
\end{split}
\end{align}
To this end, note that $\P{F_{TN}(l_{post}'\xi^{*};0,\sigma_{post}^{*2},\truncpost^{*})\1\curly{\xi^{*} \in \mathcal{B}_{post}^{*}} \neq 1-\alpha} = 1$, since
\begin{align*}
    F_{TN}(l_{post}'\xi^{*};0,\sigma_{post}^{*2},\truncpost^{*})|\xi^{*} \in \mathcal{B}_{post}^{*} \sim U(0,1).
\end{align*}
Thus, if we can show that
\begin{align}\label{d2:app:eq:base.target.for.CMT}
\begin{split}
&(F_{TN}(l_{post}'\tballn[n_{s}]^{*};0,\tsigmapostn[n_{s}]^{2},\ttruncpostn[n_{s}]^{*})\1\curly{\tballns \in \tBpostns}, \1\curly{\tballns \in \tBpostns}) \\
&\to[d] (F_{TN}(l_{post}'\xi^{*};0,\sigma_{post}^{*2},\truncpost^{*})\1\curly{\xi^{*} \in \mathcal{B}_{post}^{*}}, \1\curly{\xi^{*} \in \mathcal{B}_{post}^{*}}),
\end{split}
\end{align}
then CMT will yield the convergence in \eqref{d2:app:eq:ultimate.target.for.CMT}.

To this end, recall the convergence in \eqref{d2:app:eq:stacked.CLT} yields
\begin{align}\label{d2:app:convergent.variables}
    \begin{pmatrix}
        \tballns^{*} \\
        \tSigmaalln[n_{s}] \\
        A_{post}\tballPns \\
        \tilde{c}_{post,n_{s}}
    \end{pmatrix} 
    \to[d]
    \begin{pmatrix}
        \xi^{*} \\
        \Sigmaall^{*} \\
        a_{post}^{*} \\
        c_{post}^{*}
    \end{pmatrix}.
\end{align}
Let $\mathcal{V} = \R^{\dim(\ball)} \times \mathcal{M} \times [-\infty, +\infty]^{\dim(c)} \times [-\blam_{c}, \blam_{c}]^{\dim(c)}$. Given some function $g: \mathcal{V} \to \R^{\dim(g)}$, let $\mathcal{C}(g) \subseteq \mathcal{V}$ denote the set of points at which $g$ is continuous. If we can represent the RHS of \eqref{d2:app:eq:base.target.for.CMT} as a function $g(\xi^{*}, \Sigmaall^{*}, a_{post}^{*}, c_{post}^{*})$ for which the set of continuity points $\mathcal{C}(g)$ has probability one under $\xi^{*} \sim N(0, \Sigmaall^{*})$, $\Sigmaall^{*} \in \mathcal{M}$, $\apost^{*} \in [-\infty, +\infty]^{\dim(c)}$, and $\cpost^{*} \in [-\blam_{c}, \blam_{c}]^{\dim(c)}$ in that
\begin{align*}
    \P{(\xi^{*}, \Sigmaall^{*}, a_{post}^{*}, c_{post}^{*}) \in \mathcal{C}(g)} = 1,
\end{align*}
then the convergence in \eqref{d2:app:convergent.variables} combined with CMT yields the desired convergence in \eqref{d2:app:eq:base.target.for.CMT}.

To work towards a continuous representation, we first expand the truncated CDF term. By equation \eqref{d2:app:eq:membership.Z.to.B}, we have
\begin{align*}
    &F_{TN}(l_{post}'\xi^{*};0,\sigma_{post}^{*2},\truncpost^{*})\1\curly{\xi^{*} \in \mathcal{B}_{post}^{*}} = \frac{\displaystyle \int_{-\infty}^{l_{post}'\xi^{*}} \phi\paren{\frac{t}{\sigma_{post}^{*}}}\1\curly{t \in \truncpost^{*}} dt}{\displaystyle \int_{-\infty}^{\infty}\phi\paren{\frac{t}{\sigma_{post}^{*}}}\1\curly{t \in \truncpost^{*}} dt}\1\curly{\xi^{*} \in \mathcal{B}_{post}^{*}} \\
    &= \frac{\displaystyle \int_{-\infty}^{\infty} \1\curly{t \leq l_{post}'\xi^{*}}\phi\paren{\frac{t}{\sigma_{post}^{*}}}\max_{k}\1\curly{\Apostk(\gamma^{*} t + r^{*}) \leq \cpostk^{*} - \apostk^{*}} dt}{\displaystyle \int_{-\infty}^{\infty}\phi\paren{\frac{t}{\sigma_{post}^{*}}}\max_{k}\1\curly{\Apostk(\gamma^{*} t + r^{*}) \leq \cpostk^{*} - \apostk^{*}} dt}\1\curly{\xi^{*} \in \mathcal{B}_{post}^{*}}.
\end{align*}
We now derive continuous representations for the numerator and denominator. To this end, we consider sequences $(\xi^{m}, \Sigmaall^{m}, a_{post}^{m}, c_{post}^{m}) \in \mathcal{V}$ for which
\begin{align*}
    (\xi^{m}, \Sigmaall^{m}, a_{post}^{m}, c_{post}^{m}) \to (\xi^{0}, \Sigmaall^{0}, a_{post}^{0}, c_{post}^{0}),
\end{align*}
where $(\xi^{0}, \Sigmaall^{0}, a_{post}^{0}, c_{post}^{0})$ is such that the corresponding set $\R \setminus \mathcal{T}_{g_{1}}^{0}$ given below is finite.

For the numerator, consider the function
\begin{align*}
&g_{1}(\xi^{*}, \Sigmaall^{*}, a_{post}^{*}, c_{post}^{*}) \\
&= \displaystyle \int_{-\infty}^{\infty} \1\curly{t \leq l_{post}'\xi^{*}}\phi\paren{\frac{t}{\sigma_{post}^{*}}}\max_{k}\min_{j}\1\curly{(\Apostk)_{j}(\gamma^{*} t + r^{*}) \leq (\cpostk^{*})_{j} - (\apostk^{*})_{j}} dt.
\end{align*}
For $t \in \mathcal{T}_{g_{1}}^{0} = \{t \in \R: t \neq l_{post}'\xi^{0}, (\Apostk)_{j} \gamma^{0} t \neq (\cpostk^{0})_{j} - (\apostk^{0})_{j} - (\Apostk)_{j} r^{0}, \forall k,j\}$, we have
\begin{align*}
&\1\curly{t \leq l_{post}'\xi^{m}}\phi\paren{\frac{t}{\sigma_{post}^{m}}}\max_{k}\min_{j}\1\curly{(\Apostk)_{j}(\gamma^{m} t + r^{m}) \leq (\cpostk^{m})_{j} - (\apostk^{m})_{j}} \\
&\to \1\curly{t \leq l_{post}'\xi^{0}}\phi\paren{\frac{t}{\sigma_{post}^{0}}}\max_{k}\min_{j}\1\curly{(\Apostk)_{j}(\gamma^{0} t + r^{0}) \leq (\cpostk^{0})_{j} - (\apostk^{0})_{j}},
\end{align*}
since the latter is continuous at $(\xi^{0}, \Sigmaall^{0}, a_{post}^{0}, c_{post}^{0})$. Since $\R \setminus \mathcal{T}_{g_{1}}^{0}$ is finite, its Lebesgue measure is zero; moreover, under the eigenvalue bounds for $\mathcal{M}$, the terms $\phi(t/\sigma_{post}^{m})$ are bounded above by a finite constant, uniformly over $(t,m)$. Therefore, dominated convergence theorem implies
\begin{align*}
    &g_{1}(\xi^{m}, \Sigmaall^{m}, a_{post}^{m}, c_{post}^{m}) \\
    &= \displaystyle \int_{-\infty}^{\infty} \1\curly{t \leq l_{post}'\xi^{m}}\phi\paren{\frac{t}{\sigma_{post}^{m}}}\max_{k}\min_{j}\1\curly{(\Apostk)_{j}(\gamma^{m} t + r^{m}) \leq (\cpostk^{m})_{j} - (\apostk^{m})_{j}} dt \\
    &\to \displaystyle \int_{-\infty}^{\infty} \1\curly{t \leq l_{post}'\xi^{0}}\phi\paren{\frac{t}{\sigma_{post}^{0}}}\max_{k}\min_{j}\1\curly{(\Apostk)_{j}(\gamma^{0} t + r^{0}) \leq (\cpostk^{0})_{j} - (\apostk^{0})_{j}} dt \\
    &= \displaystyle \int_{-\infty}^{\infty} \1\curly{t \leq l_{post}'\xi^{0}}\phi\paren{\frac{t}{\sigma_{post}^{0}}}\max_{k}\1\curly{\Apostk(\gamma^{0} t + r^{0}) \leq \cpostk^{0} - \apostk^{0}} dt \\
    &= g_{1}(\xi^{0}, \Sigmaall^{0}, a_{post}^{0}, c_{post}^{0}).
\end{align*}
The above argument holds for any $(\xi^{0}, \Sigmaall^{0}, a_{post}^{0}, c_{post}^{0})$ where $\R \setminus \mathcal{T}_{g_{1}}^{0}$ is finite. This finiteness holds whenever
\begin{align*}
    (\cpostk^{0}-\apostk^{0})_{j} - (\Apostk)_{j} r^{0} \neq 0, \quad \forall j:(\Apostk\gamma^{0})_{j} = 0, \quad \forall k. 
\end{align*}
Thus the set of continuity points for $g_{1}(\xi^{*}, \Sigmaall^{*}, a_{post}^{*}, c_{post}^{*})$ contains
\begin{align*}
    &\curly{(\xi^{*}, \Sigmaall^{*}, a_{post}^{*}, c_{post}^{*}): (\Apostk)_{j} r^{*} \neq (\cpostk^{*})_{j}-(\apostk^{*})_{j}, \forall j \in J_{k}^{0}, \forall k} \\
    &= \curly{(\xi^{*}, \Sigmaall^{*}, a_{post}^{*}, c_{post}^{*}): (\Apostk)_{j} \xi^{*} \neq (\cpostk^{*})_{j}-(\apostk^{*})_{j}, \forall j \in J_{k}^{0}, \forall k} \\
    &\subseteq
    \mathcal{C}(g_{1}). 
\end{align*}
For the denominator, analogous arguments show that the set of continuity points for
\begin{align*}
&g_{2}(\xi^{*}, \Sigmaall^{*}, a_{post}^{*}, c_{post}^{*}) \\
&= \displaystyle \int_{-\infty}^{\infty} \phi\paren{\frac{t}{\sigma_{post}^{*}}}\max_{k}\min_{j}\1\curly{(\Apostk)_{j}(\gamma^{*} t + r^{*}) \leq (\cpostk^{*})_{j} - (\apostk^{*})_{j}} dt
\end{align*}
also contains 
\begin{align*}
    \curly{(\xi^{*}, \Sigmaall^{*}, a_{post}^{*}, c_{post}^{*}): (\Apostk)_{j} \xi^{*} \neq (\cpostk^{*})_{j}-(\apostk^{*})_{j}, \forall j \in J_{k}^{0}, \forall k}
    \subseteq
    \mathcal{C}(g_{2}).
\end{align*}
To use these results, notice $\1\curly{\xi^{*} \in \mathcal{B}_{post}^{*}} = \max_{k}\min_{j}\1\curly{(\Apostk)_{j} \xi^{*} \leq (\cpostk^{*})_{j} - (\apostk^{*})_{j}}$, and consider that
\begin{align*}
    g_{3}(\xi^{*}, \Sigmaall^{*}, a_{post}^{*}, c_{post}^{*}) = \max_{k}\min_{j}\1\curly{(\Apostk)_{j} \xi^{*} \leq (\cpostk^{*})_{j} - (\apostk^{*})_{j}}
\end{align*}
is continuous over the set
\begin{align*}
\{(\xi^{*}, \Sigmaall^{*}, \apost^{*}, \cpost^{*}): (\Apostk)_{j} \xi^{*} \neq (\cpostk^{*})_{j} - (\apostk^{*})_{j}, \forall k,j\} \subseteq \mathcal{C}(g_{3}).
\end{align*}
Under $\xi^{*} \in \mathcal{B}_{post}^{*}$, we have $0 < g_{1} \leq g_{2}$ so that
\begin{align*}
&(F_{TN}(l_{post}'\xi^{*};0,\sigma_{post}^{*2},\truncpost^{*})\1\curly{\xi^{*} \in \mathcal{B}_{post}^{*}}, \1\curly{\xi^{*} \in \mathcal{B}_{post}^{*}}) \\
&= \paren{\frac{g_{1}(\xi^{*}, \Sigmaall^{*}, a_{post}^{*}, c_{post}^{*})}{g_{2}(\xi^{*}, \Sigmaall^{*}, a_{post}^{*}, c_{post}^{*})}\1\curly{\xi^{*} \in \mathcal{B}_{post}^{*}}, \1\curly{\xi^{*} \in \mathcal{B}_{post}^{*}}} \\
&= \paren{\frac{g_{1}(\xi^{*}, \Sigmaall^{*}, a_{post}^{*}, c_{post}^{*})}{g_{2}(\xi^{*}, \Sigmaall^{*}, a_{post}^{*}, c_{post}^{*})}g_{3}(\xi^{*}, \Sigmaall^{*}, a_{post}^{*}, c_{post}^{*}), g_{3}(\xi^{*}, \Sigmaall^{*}, a_{post}^{*}, c_{post}^{*})} \\
&= g(\xi^{*}, \Sigmaall^{*}, a_{post}^{*}, c_{post}^{*}).
\end{align*}
By the above analysis, the set of continuity points for $g$ contains
\begin{align*}
    \mathcal{C}(g_{1}) \cap \mathcal{C}(g_{2}) \cap \mathcal{C}(g_{3}) =\{(\xi^{*}, \Sigmaall^{*}, a_{post}^{*}, c_{post}^{*}): (\Apostk)_{j} \xi^{*} \neq (\cpostk^{*})_{j}-(\apostk^{*})_{j}, \forall k, j\} \subseteq \mathcal{C}(g).
\end{align*}
Thus, 
\begin{align*}
    &\P{(\xi^{*}, \Sigmaall^{*}, a_{post}^{*}, c_{post}^{*}) \in \mathcal{C}(g)} \\
    &\geq \P{(\Apostk)_{j} \xi^{*} \neq (\cpostk^{*})_{j}-(\apostk^{*})_{j}, \forall k, j} \\
    &= 1-\P{\exists (k,j): (\Apostk)_{j} \xi^{*} = (\cpostk^{*})_{j}-(\apostk^{*})_{j}} \\
    &\geq 1 - \sum_{(k,j)}\P{(\Apostk)_{j} \xi^{*} = (\cpostk^{*})_{j}-(\apostk^{*})_{j}} \\
    &= 1,
\end{align*}
where the last equality follows from continuity of $(\Apostk)_{j}\xi^{*} \sim N(0, (\Apostk)_{j}\Sigmaall^{*}(\Apostk)_{j}')$, since $(\Apostk)_{j} \neq 0$ for all $(k,j)$. To conclude, CMT combined with \eqref{d2:app:convergent.variables} yields the desired convergence in \eqref{d2:app:eq:base.target.for.CMT}: 
\begin{align*}
\begin{split}
&(F_{TN}(l_{post}'\tballn[n_{s}]^{*};0,\tsigmapostn[n_{s}]^{2},\ttruncpostn[n_{s}]^{*})\1\curly{\tballns \in \tBpostns}, \1\curly{\tballns \in \tBpostns}) \\
&= g(\tballns^{*},  \tSigmaalln[n_{s}],  A_{post}\tballPns, \tilde{c}_{post,n_{s}}) \\
&\to[d] g(\xi^{*}, \Sigmaall^{*}, a_{post}^{*}, c_{post}^{*}) \\
&= (F_{TN}(l_{post}'\xi^{*};0,\sigma_{post}^{*2},\truncpost^{*})\1\curly{\xi^{*} \in \mathcal{B}_{post}^{*}}, \1\curly{\xi^{*} \in \mathcal{B}_{post}^{*}}).
\end{split}
\end{align*}
\end{proof}

\end{document}